\documentclass[11pt]{article}
\usepackage[margin=0.8in]{geometry}

\usepackage{graphicx}
\usepackage{framed}
\usepackage[normalem]{ulem}

\usepackage{xfrac, amsthm,thm-restate,thmtools}
 \usepackage{amssymb}
\usepackage{amsbsy}
\usepackage{amsfonts}
\usepackage{cancel}
\usepackage{enumerate}
\usepackage{etoolbox}
\usepackage{subcaption}
\usepackage{url}
\usepackage{enumitem}
\usepackage{mathtools}
\usepackage{comment}
\usepackage[font={small,it}]{caption}
\usepackage{bbm}
\usepackage[usenames,dvipsnames]{color}
\usepackage{xcolor}
\usepackage{algorithmic} 
\usepackage{algorithm}
\usepackage{array}
\usepackage{nicefrac}
\usepackage{bm}
\usepackage{tikz}
\usepackage{mdframed}
\definecolor{gray230}{RGB}{240,240,240}
\usepackage[colorlinks=true, allcolors=black]{hyperref}
\hypersetup{    
  colorlinks   = true, 
  urlcolor     = blue, 
  linkcolor    = blue, 
  citecolor   = red 
}
\providecommand{\keywords}[1]{\textbf{\textit{Index terms---}} #1}

\usepackage{thmtools,thm-restate}

\DeclareMathOperator*{\argmax}{arg\,max}
\DeclareMathOperator*{\argmin}{arg\,min}

\newtheorem{theorem}{Theorem}
\newtheorem{lemma}{Lemma}

\newtheorem{corollary}{Corollary}
\newtheorem{definition}{Definition}
\newtheorem{observation}{Observation}
\newtheorem*{example}{Example}

\newtheorem{proposition}{Proposition}

\title{\bfseries Influence before Hiring: A Two-tired Incentive Compatible Mechanism for IoT-based Crowdsourcing in Strategic Setting}

\author{Chattu Bhargavi\thanks{\textcolor{blue}{School of Computer Science and Engineering, VIT-AP University, Amaravati, India.} {\tt \textcolor{blue}{bhargavi.chattu506@gmail.com}}}~\href{https://orcid.org/0000-0003-4481-827X}{\includegraphics[scale=0.0045]{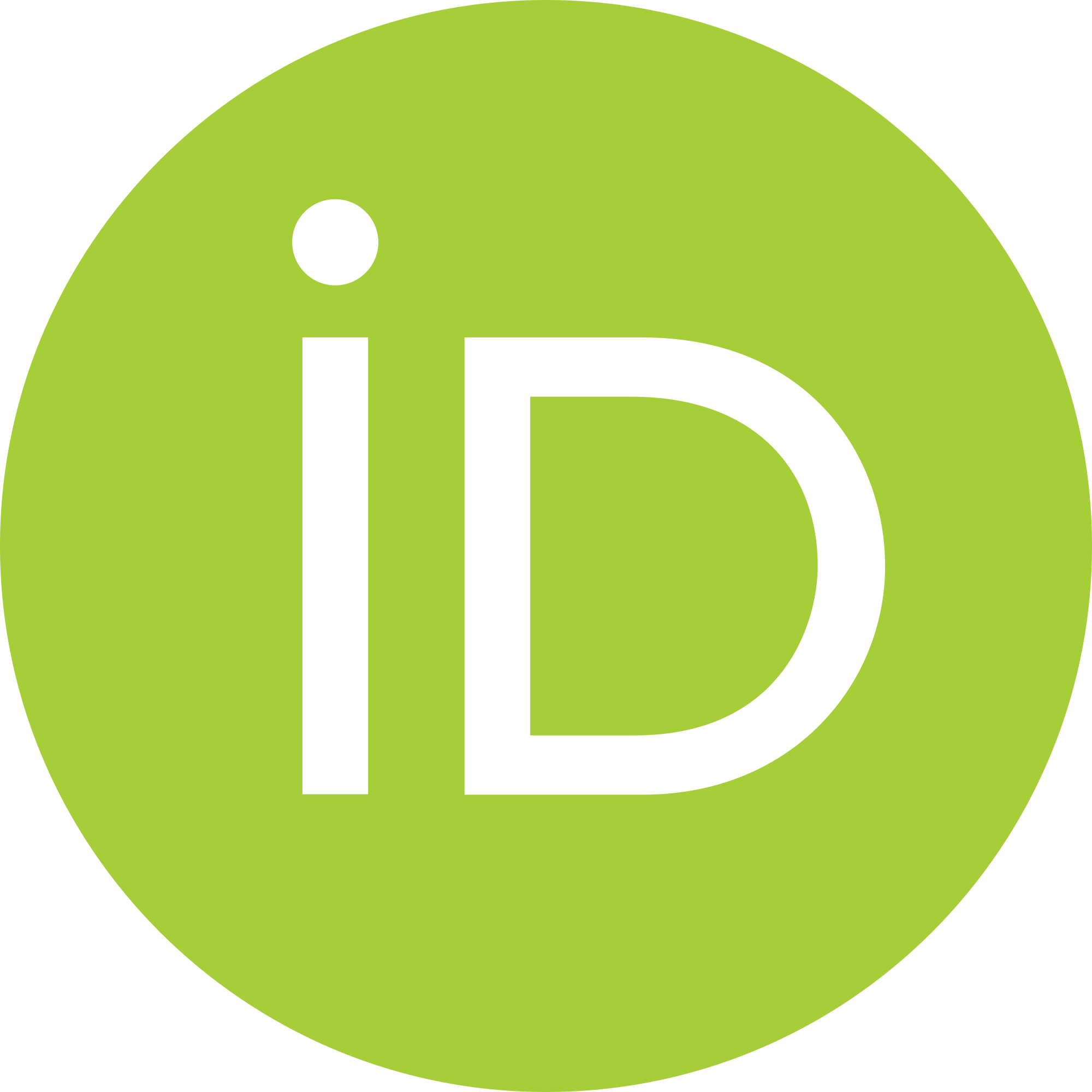}}\and Vikash Kumar Singh\thanks{\textcolor{blue}{School of Computer Science and Engineering, VIT-AP University, Amaravati, India.} {\tt \textcolor{blue}{vikash.singh@vitap.ac.in}}}~\href{https://orcid.org/0000-0002-8747-1627}{\includegraphics[scale=0.0047]{orcid.png}} }
\date{}
\begin{document}
\maketitle

\begin{abstract}
In crowdsourcing, a group of common people is asked to execute the tasks and in return will receive some incentives (maybe monetary benefits or getting social recognition). In this article, one of the crowdsourcing scenarios with multiple heterogeneous tasks and multiple IoT devices (as task executors) is studied as a two-tiered process.\\
\indent In the first tier of the proposed model, it is assumed that a substantial number of IoT devices are not aware of the hiring process and are made aware by utilizing their social connections. Each of the IoT devices is endowed with a cost (\emph{private value}) that it will charge in return for its services. The participating IoT devices are \emph{rational} and \emph{strategic} in nature. The objective of the first tier is to select the subset of IoT devices as \emph{initial notifiers} (helps in spreading awareness among the IoT devices about the task execution process) such that the total number of IoT devices notified is maximized with the stopping condition that the total payment offered to the notifiers is less than or equals to the available budget. For this purpose, an incentive compatible mechanism is proposed that also ensures the total payment made to the initial notifiers is less than or equal to the budget. Once the substantial number of IoT devices got intimated about the hiring process, in the second tier, a subset of quality IoT devices is determined by utilizing the idea of \emph{single-peaked preferences}. Once the quality of IoT devices is determined, the next objective of the second tier is to hire quality IoT devices for the floated tasks. For this purpose, each of the quality IoT devices reports private valuation along with their favorite bundle of tasks that they are interested in executing. In the second tier, it is assumed that the valuation of the IoT devices satisfies \emph{gross substitute} criteria and is \emph{private}. For the second tier, the truthful mechanisms are designed independently for determining the quality of IoT devices and for hiring quality IoT devices and deciding their payment respectively.\\
\indent Theoretical analysis is carried out for the two tiers independently. It is shown that the proposed mechanisms are \emph{computationally efficient}, \emph{truthful}, \emph{correct}, and \emph{budget feasible} (only in the case of the mechanism proposed in the first tier). Further, the probabilistic analysis is carried out to have an estimate of the expected number of IoT devices got notified about the task execution process. The simulation is done to measure the efficacy of the proposed mechanisms with the benchmark mechanisms based on \emph{truthfulness}, \emph{budget feasibility}, and \emph{running time}.  
 \end{abstract}
\section{Introduction}
\label{s:intro}
\emph{Crowdsourcing} is a process of completing the floated tasks by a group of common people through an open call \cite{liu2022budget, vahdat2022survey, amour2022crowdsourcing, ang2022towards}. It mainly consists of players such as: (1) \textcolor{red}{\emph{task requester(s)}}, (2) \textcolor{blue}{\emph{platform}} (or \emph{third party}), and (3)  \textcolor{green}{\emph{crowd workers}}. The workflow of the crowdsourcing system is, firstly, the task requester(s) will submit their tasks to some \emph{third party}. On receiving the tasks from the \emph{task requesters}, the \emph{third party} provides the tasks to the crowd workers that are present on the other side of the crowdsourcing market. The crowd workers execute the tasks and submit back the completed tasks to the platform. The third party returns the executed tasks to the respective task requester(s) and the crowd workers get some incentives (maybe monetary benefits or some social recognition) in exchange for their services (in this case executing the floated tasks). The above-discussed scenario is said to be ``\emph{crowdsourcing}" \cite{elsokkary2023crowdsourced,cricelli2021crowdsourcing, ang2022towards, Singh_2020, Singh2019}. However, when crowdsourcing is done using smart devices, it gives rise to a field called ``\emph{mobile crowdsourcing}" (or \emph{mobile crowdsensing} or \emph{participatory sensing} (PS)) \cite{Mukhopadhyay2021, singh2022quad, jaimes2012location, kim2022privacy,Fujihara2020PoWaPPO}.  \\
\indent One of the challenging aspects of crowdsourcing and PS is to have a large number of common people as the task executors in the system. Now, the question is: \emph{how to drag a large group of common people into such systems}? One of the solutions could be to provide them with some incentives (maybe money or some social recognition). In the past, the works have been carried out for designing the \emph{mechanisms} (a.k.a \emph{algorithms}) that will offer incentives to the crowd workers in return for their services, in \emph{strategic} setting \cite{singh2022quad, Mukhopadhyay2021, Singh_2020, Singh2019,DBLP:conf/hcomp/GoelNS14, 9741370}. In \cite{singh2022quad} for the set-up with multiples task requesters and task executors a \emph{truthful} mechanism is discussed. Each task requester is endowed with multiple homogeneous tasks and a bid. Both these quantities are private and are reported to the platform. On the other hand, the IoT devices in the mobile crowdsourcing market report the ask and the number of tasks they can execute, on the platform. The proposed truthful mechanism selects the subset of quality IoT devices (as task executors) for the set of tasks. In \cite{Singh_2020} a quality-based truthful mechanism is proposed for assigning a subset of tasks to the IoT devices in a mutually exclusive manner such that the sum of the valuations of the IoT devices gets maximized. In \cite{Singh2019} an effort has been made to design a quality-adaptive budget feasible truthful mechanism for the set-up consisting of multiple task requesters and multiple task executors. Each task requester has a single task along with the budget. On the other side, there are multiple IoT devices (as crowd workers) that report the bid values (cost they will charge in exchange for their services). Further, the more realistic flavor of the discussed set-up is studied where the tasks are divisible in nature. For the extended version of the problem, a non-truthful budget feasible mechanism is discussed. In \cite{Mukhopadhyay2021}, a mobile crowdsourcing scenario with a single task requester and multiple IoT devices (as task executors) is investigated in \emph{strategic} setting. The task requester has the set of tasks and the budget associated with the tasks. The task executors report a bid for executing the tasks and are private. In the proposed model, the overall budget is not available apriori and is made available in an incremental manner in multiple rounds. For this set-up, a \emph{truthful} mechanism is proposed that also considers that the total payment made to the task executors is within the budget. In \cite{DBLP:conf/hcomp/GoelNS14} an incentive-compatible mechanism is designed for one of the scenarios in crowdsourcing with a single task requester and multiple task executors. The task requester is having multiple tasks and a fixed budget. The goal is to select a set of task executors for executing the set of tasks such that the total payment offered to the selected task executors is less than or equal to the fixed budget. It is to be noted that, in the above-discussed scenarios, the task executors were already aware of the task execution process, but it may not be the case always. It means that only fewer task executors may be aware of the task execution process (or event). Now the question is, how to inform others about the ongoing event? To address the above-discussed realistic scenario, some works have been carried out in the past \cite{9741370, 10.1145/3487580, 8999584}. A two-tiered social crowdsourcing architecture is proposed in \cite{9741370} that allows the task executors to forward the floated tasks that are to be executed to their neighbors in the social connections. In this setup, the tasks are having different end times. For this scenario, the three different system models are discussed based on the arrival modes of the registered users and social neighbors. For the three different models, a truthful mechanism is proposed. In \cite{10.1145/3487580} to increase the crowd workers the floated tasks are diffused in the social network (representing the social connections of the crowd workers). The objective is to diffuse the tasks to as many crowd workers as possible with the constraint that the total payment made to the task diffusers is within the available budget. For the discussed set-up a truthful budget feasible mechanism is developed that takes into account the enhanced classic \emph{independent cascade model}. In \cite{8999584} an effort has been made to design a dynamic incentive mechanism that transfers information about the task execution process through the social connections of the task executors.  
  
\begin{figure*}
\centering
\includegraphics[scale=0.77]{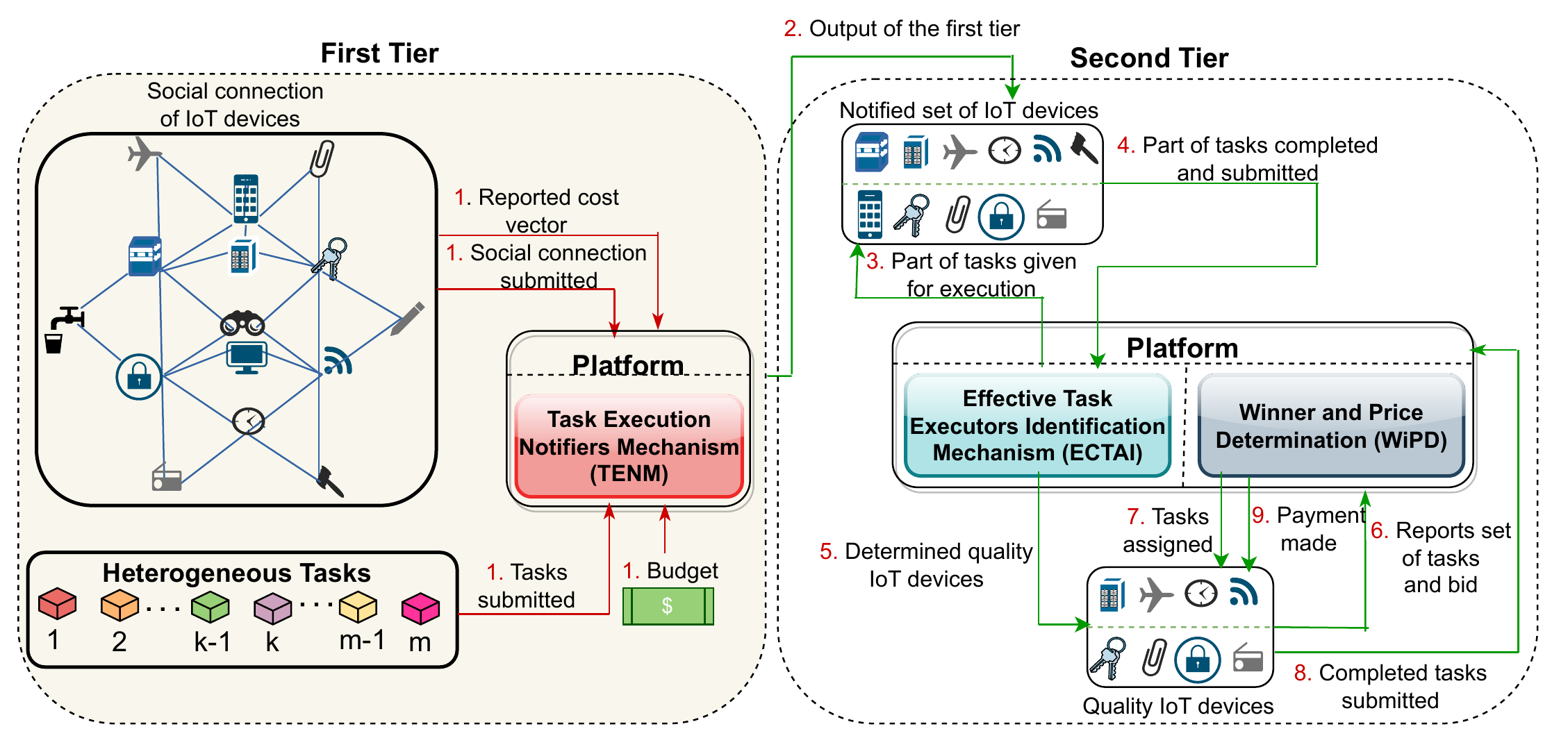}
\caption{A two-tiered framework for IoT-based crowdsourcing}
\label{fig:01}
\end{figure*}
Motivated by the above discussed crowdsourcing scenarios, in this paper, one of the scenarios of IoT-based crowdsourcing is studied as a two-tiered process in \emph{strategic} setting\footnote{\textcolor{blue}{By strategic, it is meant that the agents will try to manipulate the system by misreporting their private information.}} as shown in Figure \ref{fig:01}. In the proposed model, we have multiple heterogeneous tasks and multiple task executors (as \emph{IoT devices}). Firstly, the tasks are submitted to the platform for execution purposes. One of the assumptions is that an insufficient number of task executors are aware of the task execution process. So, on receiving the tasks, one of the challenges of the platform is: \emph{how to inform the sufficient number of task executors about the task execution process?} One of the solutions could be to utilize the social connections of the task executors for notifying (or informing) the sufficient number of task executors for the task execution process. Once notified, the next objective is to hire quality task executors for each of the tasks. For the above-discussed scenario, the discussed model is studied as a two-tier process. In the first tier, the social connections of the task executors are considered to inform the substantial number of task executors about the task execution process. The input to the first tier is the social connection of the task executors, the cost vector\footnote{\textcolor{blue}{The maximum price the IoT devices will charge in exchange for their services.}}, the set of tasks to be executed, and the available budget\footnote{\textcolor{blue}{The money to be invested as the payment for the notifiers in return for their `\emph{word of mouth}' in their social network.}}. The cost for notifying the task execution process to the task executors is \emph{private} and the task executors can act \emph{strategically} to gain. The objective of the first tier is to select the subset of task executors from the given social graph such that the number of task executors that got notified is \emph{maximized} with the constraint that the total payment made to the notifiers is within the fixed budget. Once the substantial number of task executors got notified about the task execution process, in the second tier, the challenges are (1) \emph{to determine the quality task executors among the notified task executors}, and (2) \emph{to hire the quality task executors and decide their payment in exchange for their services}. In the second tier, the challenges mentioned in points 1 and 2 are taken care of. Firstly, to have an idea about the quality of the task executors, an infinitesimally small part of the tasks are given to the task executors for execution purposes. Once the tasks get executed, the executed tasks are given to the peers for estimating the quality of the completed tasks by the task executors ($i.e.$ the quality of the task executors). Once the quality of the task executors is determined, the next objective is to handle the challenge mentioned in point 2 above. For that purpose, each of the quality task executors will be asked for a set of tasks and the valuation that they will charge for executing the requested set of tasks. Given the discussed setup, the output of the second tier is to allocate the tasks to the quality task executors and decide their payment.\\
\indent In this paper for the above-discussed set-up a \emph{\underline{\textbf{T}}wo-tiered \underline{\textbf{I}}ncentive \underline{\textbf{CO}}mpatible \underline{\textbf{M}}echanism} (TICOM) is proposed that consists of three components: (1) \emph{\underline{\textbf{T}}ask \underline{\textbf{E}}xecution \underline{\textbf{N}}otifiers \underline{\textbf{M}}echanism} (TENM), (2) \emph{\underline{\textbf{E}}ffe\underline{\textbf{C}}tive \underline{\textbf{T}}\underline{\textbf{A}}sk executors \underline{\textbf{I}}dentification mechanism} (ECTAI), and (3) \emph{\underline{\textbf{Wi}}nner and \underline{\textbf{P}}rice \underline{\textbf{D}}etermination} (WiPD). In Section \ref{s:prelim} the above-discussed scenario is formulated using mechanism design.  
\subsection{Our Contributions}
The contribution of this paper is:
\begin{enumerate}
    \item The task execution process in crowdsourcing is studied as a two-tiered process. Firstly, a sufficient number of IoT devices are made aware of the hiring process (or task execution process) by utilizing the social connections of the IoT devices (using Algorithm \ref{algo:1}). Further, the IoT devices that helped in spreading awareness about the task execution process in their social connections are paid such that the total payment made to them is within the fixed budget (using Algorithm \ref{algo:2}).  
    \item Once a sufficient number of IoT devices got notified about the task execution process, in the second tier, (1) the quality of the IoT devices is determined (using Algorithm \ref{algo:3}), and (2) for the floated tasks the quality IoT devices are selected and their payments are decided (using Algorithm \ref{algo:14}).      
    \item For the discussed set-up, TICOM is proposed that consists of the following components: (1) \underline{\textbf{T}}ask \underline{\textbf{E}}xecution \underline{\textbf{N}}otifiers \underline{\textbf{M}}echanism (TENM) (Algorithm \ref{algo:1} and Algorithm \ref{algo:2}), (2) \underline{\textbf{E}}ffe\underline{\textbf{C}}tive \underline{\textbf{T}}\underline{\textbf{A}}sk executors \underline{\textbf{I}}dentification mechanism (ECTAI) (Algorithm \ref{algo:3}), and (3) \underline{\textbf{Wi}}nner and \underline{\textbf{P}}rice \underline{\textbf{D}}etermination (WiPD) (Algorithm \ref{algo:14}).
    \item The theoretical analysis of the two tiers is carried out independently. Firstly, in the first tier, through theoretical analysis it is shown that TENM is \emph{computationally efficient} (Lemma \ref{lemma:p1}) \emph{correct} (Lemma \ref{lemma:p2}), \emph{truthful} (Corollary \ref{lemma:p3}), and \emph{budget feasible} (Corollary \ref{corr:p2}). Further, in the first tier, the probabilistic analysis is carried out to have an estimate of the number of IoT devices that got notified in expectation in a social graph. In the second tier, through theoretical analysis, it is shown that ECTAI and WiPD are \emph{computationally efficient} (Lemma \ref{lemma:sp1}), \emph{correct} (Lemmas \ref{lemma:sp2} and \ref{lemma:sp3}), and \emph{truthful} (Lemmas \ref{lemma:8} and \ref{lemma:ttu}).  
    \item The simulation for the two tiers is done independently. In the first tier, TENM is compared with two baseline mechanisms, namely, \emph{\textbf{\underline{n}}on-\textbf{\underline{t}}ruthful \textbf{\underline{b}}udget \textbf{\underline{f}}easible \textbf{\underline{m}}echanism} (NTBFM) and \emph{\textbf{\underline{p}}roportional \textbf{\underline{s}}hare \textbf{\underline{m}}echnaism} (PSM) \cite{Singer_2016}. The comparison is done based on (1) the utility of notifiers, (2) budget feasibility, (3) the number of IoT devices selected as initial notifiers in the social network, and (4) running time. In the second tier of the proposed framework, the experiments are carried out to compare ECTAI and WiPD with the already existing mechanism \emph{\textbf{\underline{a}}verage \textbf{\underline{v}}oting \textbf{\underline{r}}ule} (AVR) and greedy mechanism (outline of this mechanism is given in Subsection \ref{subsec:bl}) respectively on the ground of \emph{truthfulness}, \emph{individual rationality}, and \emph{running time}. 
\end{enumerate}

\subsection{Paper Organization}
 The remainder of the paper is structured as follows. The works carried out in the fields of crowdsourcing and mobile crowdsourcing is discussed in Section \ref{sec:rw}. Section \ref{s:prelim} describes the notations and preliminaries that are utilized throughout the paper. The proposed mechanisms are illustrated in section \ref{section:algorithm}. In section \ref{section:PM} the game theoretic and probabilistic analysis of the proposed mechanisms for the two tiers is carried out independently. The simulation and result analysis are carried out in Section \ref{sec:Sim}. Finally, the paper is concluded with the possible future directions in section \ref{se:conc}. 

\section{Related Prior Works}\label{sec:rw} In this section, the works carried out in crowdsourcing and PS in \emph{strategic} setting are discussed. The readers can go through \cite{10.1145/3494522, kim2022privacy, cricelli2021crowdsourcing, s20072055, abualsaud2018survey,  phuttharak2018review} to get an idea of the recent works carried out in crowdsourcing and PS. \\
\indent In \cite{s20072055} a comprehensive
review of different game theoretic solutions is done, that address the following issues in PS such as \emph{sensing cost}, \emph{quality of data}, and \emph{incentives}. In the past several incentive mechanisms are developed for different crowdsourcing and PS scenarios in \emph{strategic} setting \cite{qiao2022truthful, 8031314, DBLP:conf/hcomp/GoelNS14, 9369249, 10.1109/TCSS.2022.3149000, singh2022quad}. In \cite{qiao2022truthful} an incentive-compatible profit-oriented mechanism is designed for the setup with a single crowdsourcer and multiple workers. The workers will submit the bids (the amount they will charge in return for their services). Along with truthfulness, the proposed mechanism is \emph{individually rational} and \emph{computationally efficient}. In \cite{8031314} the incentive mechanisms were designed for IoT-based mobile crowdsourcing systems (MCSs) for surveillance applications. In \cite{DBLP:conf/hcomp/GoelNS14} paper the setup consists of a set of heterogeneous tasks such that it requires certain skills from the crowd workers to get completed. For this purpose, the crowd workers show interest in the set of tasks that they can perform based on their skills. The goal is to design a mechanism that along with \emph{truthfulness} satisfies \emph{budget feasibility}. In \cite{9369249} article the problem of allocating heterogeneous tasks with multiple skill requirements in crowdsourcing is tackled. The objective is to determine the mutually exclusive, quality set of workers who can successfully complete the tasks within a given \emph{deadline} and \emph{budget}. In \cite{10.1109/TCSS.2022.3149000} the goal is to crowdsource the small tasks such as \emph{image labeling} and \emph{voice recording} that gives rise to several challenges: (1) crowd workers may have different capacities for doing the works and may misreport it with their bid, (2) if the auction is running multiple times, then there is a chance that some sufficient number of workers may leave the market that reduces the competition in the system. To tackle the above challenges a \emph{truthful} mechanism is developed. In \cite{xiao2017mobile} the vehicles choose their capability for sensing the tasks based on sensing and transmission cost and the expected payment that will be received from the server. The Nash equilibrium (NE) of the static vehicular crowdsensing game had been determined for the sensing task and gave the condition that leads to the existence of NE. For the dynamic mobile crowdsensing game the solution is based on reinforcement learning.\\
\indent In \cite{fang2022selecting}, the crowdsourcing system is studied as a two-stage problem that consists of a task assignment stage and a truth discovery stage. Utilizing the prior knowledge about the domain of the tasks, firstly, the tasks are classified based on the domain and then allocated to the respective expert domains using a mechanism based on \emph{greedy} algorithm. To identify the copiers, the Bayesian model is utilized. Further for truth discovery, the iterative method is adopted.  A two-tiered social crowdsourcing architecture is proposed in \cite{9741370} that allows the task executors to forward the floated tasks that are to be executed to their neighbors in the social connections. In this setup, the tasks are having different end times. For this scenario, the three different system models are discussed based on the arrival modes of the registered users and social neighbors. For the three different models, a truthful mechanism is proposed. In \cite{10.1145/3487580} to increase the crowd workers the floated tasks are diffused in the social network (representing the social connections of the crowd workers). The objective is to diffuse the tasks to as many crowd workers as possible with the constraint that the total payment made to the task diffusers is within the available budget. For the discussed set-up a truthful budget feasible mechanism is developed that takes into account the enhanced classic \emph{independent cascade model}. In \cite{8999584} an effort has been made to design a dynamic incentive mechanism that transfers information about the task execution process through the social connections of the task executors. \cite{9416787} developed an incentive-based mechanism for truth discovery, with the primary objective being minimizing the copiers.\\ 
\indent Several quality-based incentive schemes are developed for different scenarios in crowdsourcing and PS \cite{Mukhopadhyay2021, Singh2019, Singh_2020, 8667369, singh2022quad}. In \cite{Mukhopadhyay2021} the setup consists of a single task requester and multiple task executors, where a task requester is endowed with multiple tasks and the budget. It is assumed that the overall budget is not available apriori and will be available in an incremental fashion. On the other side, we are having multiple task executors along with the charges that they will ask in return for their services. The objective is to select the subset of quality task executors for the given tasks such that the total payment made to the task executors is within the budget. In \cite{Singh2019} the heterogeneous task assignment problem is investigated in \emph{strategic} setting. The setup consists of multiple task requesters, each having a single task and multiple IoT devices (as task executors). In this, there is a publicly known budget that will be utilized for payment to the task executors in exchange for their services. The objective is to select the subset of quality  task executors for each task such that the total payment made to the task executors is within the budget. The setup with multiple task requesters and multiple task executors, where each task requester is endowed with multiple tasks is discussed in \cite{Singh_2020}. Here, each of the tasks has start and finish times associated with it. On the other side, we have multiple task executors that ask for the set of tasks they are interested in executing along with the cost they will charge. For the purpose of allocating the subset of task executors to each task in a non-conflicting manner a \emph{truthful} mechanism is proposed. Gong et al. \cite{8667369} considered the data quality and data accuracy, and proposed a truthful mechanism.  In \cite{singh2022quad} there are multiple task requesters and multiple IoT devices (as \emph{task executors}). Each task requester reports a set of homogeneous tasks and the bids (the amount they are willing to pay to the task executors in exchange for completing the tasks). On the other side, each of the  available IoT devices reports the number of tasks it can execute and the cost it will charge for imparting its services. The bids and asks of the task requesters and task executors respectively are \emph{private} information. For this scenario, a \emph{truthful} mechanism is proposed for allocating the quality IoT devices to the tasks carried by task requesters. Some other research works \cite{6120180, 10.1145/2796314.2745871, 7892023, electronics12040960} in crowdsourcing have focused on learning the data quality of crowd workers.  \\
\indent From the above-discussed literature reviews that the scenario discussed in this paper, in IoT-based crowdsourcing in strategic settings has not been considered in the past. In this paper, a \emph{truthful} mechanism is proposed that first provides awareness about the task execution process among its social connections. After that, each task is assigned quality IoT devices and payment for the quality IoT devices is decided.      

\section{Notation and Preliminaries}\label{s:prelim}
In this section, a crowdsourcing scenario discussed in this paper will be formulated using  \emph{mechanism design}. There are \emph{m} heterogeneous tasks and \emph{n} IoT devices (as \emph{task executors}). Here, $n \gg m$.  The set of tasks is given as $\boldsymbol{t}$, where $\boldsymbol{t}$ = $\{\boldsymbol{t}_1, \boldsymbol{t}_2, \ldots, \boldsymbol{t}_m\}$, and $\boldsymbol{t}_i$ represents $i^{th}$ task. The set of IoT devices is given as $\mathcal{I}$, where $\mathcal{I}$ = $\{\mathcal{I}_1, \mathcal{I}_2, \ldots, \mathcal{I}_n\}$, and $\mathcal{I}_j$ represents $j^{th}$ IoT device. Our proposed model is a two-tier model. Let us see each of the tiers one by one.\\

\noindent \textbf{\underline{First tier}:} In this model, it is assumed that a substantial number of IoT devices may not be aware of the \emph{task execution} event. For this reason, in the first tier, the social connection of the IoT devices is utilized to notify about the \emph{task execution event} to other IoT devices. In our model, the social connections of the IoT devices are depicted through a graph $\mathcal{G}(\boldsymbol{\mathcal{N}}^T, \boldsymbol{\mathcal{R}}^T)$, here, $\boldsymbol{\mathcal{N}}^T$ represents the set of IoT devices and are acting as the nodes (or vertices) of the graph, and $\boldsymbol{\mathcal{R}}^T$ represents the set of edges between the IoT devices in a graph $\mathcal{G}$. We say that $\mathcal{I}_i$ and $\mathcal{I}_j$ are socially connected, if and only if there exists an edge $(i,~j) \in \boldsymbol{\mathcal{R}}^T$, otherwise not. The \emph{notify function} is given as $h$, and is represented as $h:2^{\boldsymbol{\mathcal{N}}^T}\rightarrow\boldsymbol{\Re}$.  Given a set $U \subseteq \boldsymbol{\mathcal{N}}^T$ the value $h(U)$ represents the expected number of IoT devices getting notified about the \emph{task execution event} in the social graph $\mathcal{G}$. It is considered that the function $h:2^{\boldsymbol{\mathcal{N}}^T}\rightarrow\boldsymbol{\Re}$ is \emph{monotone (non-decreasing) submodular function}. By \emph{monotone}, It is meant that, for any $\mathcal{H} \subseteq \mathcal{J}$, $h(\mathcal{H}) \leq h(\mathcal{J})$.
\begin{definition}
$h:2^{\boldsymbol{\mathcal{N}}^T}\rightarrow\boldsymbol{\Re}$ is submodular if $h(\mathcal{H} \cup \{i\}) - h(\mathcal{H}) \geq h(\mathcal{J} \cup \{i\}) - h(\mathcal{J})$, $\forall$ $\mathcal{H} \subseteq \mathcal{J}$.
\end{definition}
\noindent Each IoT device $\mathcal{I}_i$ in graph  $\mathcal{G}(\boldsymbol{\mathcal{N}}^T, \boldsymbol{\mathcal{R}}^T)$ has a private \emph{bid} (or \emph{cost}) and is given as $c_i$. It is the amount that any $i^{th}$ IoT device will charge for being an initial \emph{notifier}. By private \emph{cost}, it is meant that the cost is only known to it and not known to other IoT devices and the mechanism designer. It is assumed that the IoT devices are \emph{strategic} and \emph{rational}. It means that they will try to manipulate their \emph{private} information (in this case, the \emph{cost}) to gain. For example, the reported cost by any $i^{th}$ IoT device could be $\hat{c}_i$ such that $\hat{c}_i=c_i$ or $\hat{c}_i \neq c_i$. $\hat{c}_i=c_i$ represents the fact that the IoT device $\mathcal{I}_i$ report bid in a \emph{truthful} manner. The \emph{cost vector} of the IoT devices is given as $c = \{c_1, c_2, \ldots, c_n\}$. For the purpose of notifying the IoT devices, the coverage influence model in social networks is utilized. In the graph, say, if each IoT device $\mathcal{I}_i$ is connected with the subset of IoT devices $\boldsymbol{\mathcal{Z}}_i$, then the number of IoT devices notified about the \emph{task execution event} by the subset $\boldsymbol{U} \subseteq \mathcal{I}$ is given as $h(\boldsymbol{U}) = \bigg|\bigcup\limits_{i\in \boldsymbol{U}} \boldsymbol{\mathcal{Z}}_i \bigg|$. The payment vector of the initial \emph{notifiers} is given as $\boldsymbol{\bar{\rho}}$, where $\boldsymbol{\bar{\rho}} = \{\boldsymbol{\bar{\rho}}_1, \boldsymbol{\bar{\rho}}_2, \ldots, \boldsymbol{\bar{\rho}}_n\}$ and $\boldsymbol{\bar{\rho}}_{i}$ is the payment of any $i^{th}$ notifier. The utility of any $i^{th}$ IoT device as notifier is given as: 
\begin{equation}
\label{equ:1aaaa}
    \mathcal{U}_i(c, \boldsymbol{\bar{\rho}}) = \begin{cases} \boldsymbol{\bar{\rho}}_i - c_i,\quad \text{if $\mathcal{I}_i$ acts as initial notifiers. }  \\
     0, \quad \text{otherwise }
  \end{cases}
  \end{equation}
Given the above-discussed scenario and the publicly known budget $\mathcal{B}$, the objective is to select the subset of IoT devices as the \emph{initial notifier} such that the total payment made to them is within the available budget $\mathcal{B}$. The output of the first tier is the \emph{subset of IoT devices as the initial notifiers}, the \emph{subset of notified IoT devices}, and the \emph{payment vector} that contains the payment of each of the IoT devices that are acting as the \emph{notifiers}.\\ 
\indent From the above discussion, it is clear that it is a single-parametric mechanism design problem, as each IoT device has only a single private information $i.e.$ \emph{cost}. So, for designing a \emph{truthful} mechanism for the above-discussed set-up, a \emph{greedy} technique-based mechanism is one of the viable solutions. It is due to the reason that the \emph{greedy} technique will be \emph{monotone} when sorted according to \emph{marginal notification} (see Definition \ref{def:mn}) per cost.
\begin{definition}[\textbf{Marginal notification \cite{Singer:2012:WFI:2124295.2124381}}]
\label{def:mn}
Marginal notification of j given $\mathcal{S}_{i-1}$ is: $h_{j| \mathcal{S}_{i-1}} = h(\mathcal{S}_{i-1} \cup \{j\}) - h(\mathcal{S}_{i-1})$. In this, $h_{j| \mathcal{S}_{i-1}}$ and $h_{\mathcal{I}_j| \mathcal{S}_{i-1}}$ will be used interchangeably where ever required. Here, $\mathcal{S}_{i-1}$ represents the set of $i-1$ IoT devices that are already selected using the same rule.   
\end{definition}

\begin{definition}
Marginal notification of IoT device $\mathcal{I}_i$ at position $k$ is $h_{i,k}$ = $h({T}_{k-1}\cup \{\mathcal{I}_i\}) - h({T}_{k-1})$ where $T_k$ denotes the subset of first $k$ IoT devices in the marginal notification-per-cost sorting over the subset of IoT devices $\mathcal{I}/\{\mathcal{I}_i\}.$
\end{definition}
\noindent Given the above-discussed set-up, for the first tier, a \emph{truthful} mechanism is proposed that satisfies the constraint that the total payment made to the initial notifiers is within the available budget.\\

\noindent \textbf{\underline{Second Tier}:} Once a sufficient number of IoT devices are made aware of the floated event, our next primary objective is to have a set of quality IoT devices from among the available ones. For that purpose, the idea of peer assessment is utilized. The general idea of the peer assessment is that the completed work(s) (or in our case the completed task(s)) by the IoT devices are assessed by their peers and the reports are submitted. Based on the submitted reports, the quality of IoT devices is determined. In the proposed model, the peer assessment is implemented by utilizing the idea of \emph{single-peaked preference}. In this, firstly, the IoT devices that are to be ranked are placed on the scale of $[0,1]$ randomly. After that infinitesimally small part of tasks are provided to these IoT devices for execution purposes. After execution, the completed tasks of each of the IoT devices are given to some randomly selected IoT devices (other than those present on the scale of $[0,1]$) for assessment purposes. As an assessment process, each $i^{th}$ IoT device provides a peak value $\alpha_i \in [0,1]$ that is \emph{private}. The reported peak value $\alpha_i$ of any $i^{th}$ IoT device will be the peak value at which its favorite IoT device is placed or closer to its favorite IoT device. Let us take an example to understand it in a better way. For example, let us say, there are 4 IoT devices $\mathcal{I}_{i-1}$, $\mathcal{I}_i$, $\mathcal{I}_{j}$, and $\mathcal{I}_{j+1}$ that are placed at 0.34, 0.47, 0.52, and 0.65 respectively on the scale of $[0,1]$ as shown in Figure \ref{fig:21a}. The peak values of the other 4 IoT devices $i.e.$ $\mathcal{I}_k$, $\mathcal{I}_{k+1}$, $\mathcal{I}_{i+1}$, and $\mathcal{I}_{j-1}$ are given as 0.37, 0.34, 0.58, and 0.65 respectively as shown in Figure \ref{fig:21b}. It means that $\mathcal{I}_{i+1}$'s most preferred IoT device is $\mathcal{I}_j$, $\mathcal{I}_{k+1}$'s most preferred IoT device is $\mathcal{I}_{i-1}$ and likewise. After getting the peak values from the IoT devices, in this paper, the quality of IoT devices is determined using Algorithm \ref{algo:3} (see Subsection \ref{sec:42}). 
  \begin{figure}[H]
     \centering
     \begin{subfigure}[h]{0.48\textwidth}
         \centering
         \includegraphics[scale = 0.85]{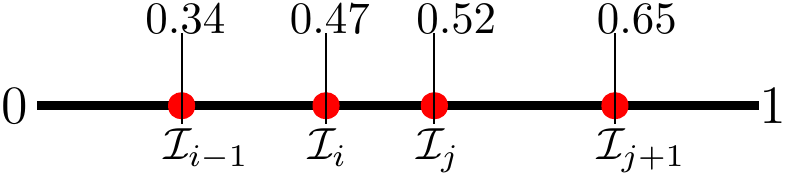}
         \caption{IoT devices placed on $[0,1]$ scale}
         \label{fig:21a}
     \end{subfigure}
     \begin{subfigure}[h]{0.49\textwidth}
         \centering
         \includegraphics[scale = 0.85]{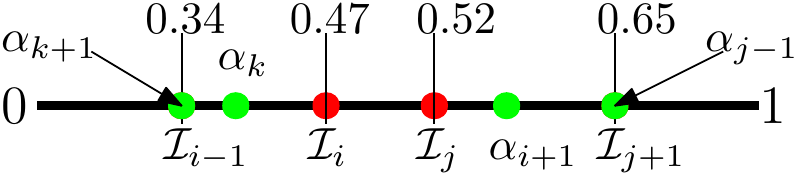}
         \caption{Peak values reported by $\mathcal{I}_k$, $\mathcal{I}_{k+1}$, $\mathcal{I}_{i+1}$, and $\mathcal{I}_{j-1}$}
         \label{fig:21b}
     \end{subfigure}
     \caption{Illustrating the meaning of peak values}
     \label{fig:21}
\end{figure}
\noindent After the determination of quality IoT devices, further challenges that need to be handled are:
 \begin{enumerate}
   \item Hire the subset of IoT devices from the available quality IoT devices.
     \item Distribute the set of available tasks to the subset of quality IoT devices for execution purposes.
     \item  What pricing strategy is to be followed for deciding the payment of the winning quality IoT devices?
 \end{enumerate}
The second step of the second tier of the proposed model takes care of the above-coined questions. It is assumed that each IoT device requests for a bundle of tasks at a time instead of a single task. Each IoT device $\mathcal{I}_i$ has a private valuation $v_i(S)$ for each bundle $S\subseteq \mathcal{F}$ of tasks that it might receive. The valuation function utilized in this step of the second tier satisfies the monotonicity condition i.e., $v_i(S) \leq v_i(\mathcal{F})$ for $S\subseteq \mathcal{F}$\footnote{\textcolor{blue}{It means that more tasks can only be better.}}. For an empty set of tasks, $v_i(\phi)=0$. For the discussed setup for a mechanism to work, it is assumed  that the valuation function satisfies the \emph{gross-substitute} condition (See Definition \ref{def:gs1}). 
 \begin{definition}[\textbf{Gross Substitute (GS) \cite{T.roughgarden_20145, NNisa_Pre_2007}}]
 \label{def:gs1}
  For any IoT device $\mathcal{I}_i$, the valuation $v_i$ satisfies GS condition if and only if for every price vector $\boldsymbol{\rho}$, some set $S\in D_i(\boldsymbol{\rho})$ and for every price vector $\boldsymbol{r}\geq \boldsymbol{\rho}$,  $\exists$ $T$ with  
  \begin{equation*}
     (S \setminus \boldsymbol{\chi}) \bigcup T \in D_i(\boldsymbol{r})
  \end{equation*}
  where $\boldsymbol{\chi} =  \bigg \{j:\boldsymbol{r}(j)>\boldsymbol{\rho}(j)\bigg \}$ is the set of tasks whose prices have gone up, and $ S\setminus\boldsymbol{\chi}$ is the set of tasks for which the prices remains same and $\mathcal{I}_i$ still wants them. $D_i(\boldsymbol{\rho})$ is the supply of IoT device $\mathcal{I}_i$ at price vector $\boldsymbol{\rho}$.
\end{definition}
\noindent For each IoT device $\mathcal{I}_i$ the utility at price vector $\boldsymbol{\rho}$ is given as:
\begin{equation}
\label{equ:1a}
     u_i(S,\boldsymbol{\rho}) = \begin{cases} \sum\limits_{j \in S} \boldsymbol{\rho}_i(j) - v_i(S),\quad \text{if $S$ is assigned to IoT device $\mathcal{I}_i$. }  \\
     0, \quad \text{otherwise }
  \end{cases}
  \end{equation}
  Here, $\boldsymbol{\rho}_i(j)$ is the price paid to the IoT device $\mathcal{I}_i$ for executing its $j^{th}$ assigned task. The utility of IoT device $\mathcal{I}_i$ is 0 if it does not receive the requested set of tasks. The supply by any IoT device $\mathcal{I}_i$ at the given price vector $\boldsymbol{\rho}$ is given as:  
  \begin{equation}
  \label{equ:1b}
     D_i(\boldsymbol{\rho})=argmin \bigg \{ \sum_{j \in S} \boldsymbol{\rho}_i(j) - v_i(S) \bigg \}_{S \subseteq \boldsymbol{t}}
  \end{equation}
As the participating IoT devices are \emph{strategic} in nature so they will try to maximize their utility by misreporting their private valuation (in this case the valuation $v_i(S)$ of IoT device $\mathcal{I}_i$ for a set of tasks $\mathcal{S}$). The objective of the second tier is to design a mechanism in presence of \emph{strategic} agents, such that, it returns an allocation and payment vectors with high social welfare (see \textbf{Definition} \ref{def:sw}). An allocation for the discussed set-up is allocation vector $\mathcal{A} =  \{\mathcal{A}_1, \mathcal{A}_2, \ldots, \mathcal{A}_{\mathcal{N}} \}$, where $\mathcal{A}_i = (\mathcal{S}_i,  \mathcal{I}_i)$ and $\boldsymbol{\rho} = \{\boldsymbol{\rho}_1, \boldsymbol{\rho}_2, \ldots, \boldsymbol{\rho}_{\mathcal{N}}\}$. Here, $\mathcal{N}$ is the number of quality IoT devices.\\
\indent Given the above-discussed scenario of IoT-based crowdsourcing in a strategic setting, the goal is to design a mechanism that takes care of the following: (1) determine the set of initial notifiers that will notify the substantial number of task executors about the task execution process with the constraint that the total payment made to the notifiers is within the budget. (2) Determine a set of quality task executors. (3) Allocating the set of tasks to the quality task executors for execution purposes and deciding their payment.\\
\indent In the upcoming section, each of the components of the proposed mechanism $i.e$ TICOM is discussed and presented in a detailed manner.
  \begin{table}[t]
\caption{Notations used}
\label{tab:first table}
\centering
\begin{tabular}{|c|c|}
\hline
\textbf{Symbols} & \textbf{Descriptions}\\
\hline
$\emph{m}$ & Number of heterogeneous tasks\\
$\emph{n}$ & Number of IoT devices\\
$\boldsymbol{t}$  & $\boldsymbol{t}$ = $\{\boldsymbol{t}_1, \boldsymbol{t}_2, \ldots, \boldsymbol{t}_m\}$ : Set of heterogeneous tasks.\\
$\boldsymbol{t}_i$ & $i^{th}$ task.\\
$\mathcal{I}$ & $\mathcal{I}$ = $\{\mathcal{I}_1, \mathcal{I}_2, \ldots, \mathcal{I}_n\}$ : Set of available IoT devices.\\
$\mathcal{I}_j$ & $j^{th}$ IoT device.\\
$\mathcal{G}$ & $\mathcal{G}(\boldsymbol{\mathcal{N}}^T, \boldsymbol{\mathcal{R}}^T)$ : Represents the social connection of IoT devices.\\ 
$\boldsymbol{\mathcal{N}}^T$ &  A set of IoT devices representing the nodes of a graph $\mathcal{G}$.\\
$\boldsymbol{\mathcal{R}}^T$ & A set of edges between the IoT devices in a graph $\mathcal{G}$.\\
$h$ & $h:2^{\boldsymbol{\mathcal{N}}^T}\rightarrow\boldsymbol{\Re}$ :  Represents a notify function.\\ 
$\boldsymbol{\mathcal{Z}}_i$ & Set of IoT devices notified by $i^{th}$ IoT device.\\
$c_i$ & True cost of $i^{th}$ IoT device.\\
$\hat{c}_i$ & Reported cost of $i^{th}$ IoT device.\\
$c$ & $c = \{c_1, c_2, \ldots, c_n\}$ : Cost vector of the IoT devices for \\ & being the initial notifiers.\\
$\mathcal{B}$ & Available budget.\\
$\boldsymbol{\bar{\rho}}$ &  $\boldsymbol{\bar{\rho}} = \{\boldsymbol{\bar{\rho}}_1, \boldsymbol{\bar{\rho}}_2, \ldots, \boldsymbol{\bar{\rho}}_n\}$ : Payment vector of IoT devices that are competing \\ & for being the initial notifiers.\\
$\boldsymbol{\bar{\rho}}_{i}$ &  Represents the payment of any $i^{th}$ IoT device\\ & as an initial notifier.\\
$\mathcal{U}_i(c, \boldsymbol{\bar{\rho}})$ & Utility of any $i^{th}$ IoT device given cost and payment vectors. \\
$\alpha_i$ & $\alpha_i \in [0,1]$: Peak value of $i^{th}$ IoT device.\\
$v_i(\mathcal{S})$ & Private valuation of $i^{th}$ IoT device for set of tasks $\mathcal{S}$. \\
$\mathcal{N}$ & It is the number of quality IoT devices.\\
$\mathcal{A} $ &  $\mathcal{A} =  \{\mathcal{A}_1, \mathcal{A}_2, \ldots, \mathcal{A}_{\mathcal{N}}\}$: Allocation vector.\\
$\mathcal{A}_i$ & $\mathcal{A}_i= (\mathcal{S}_i,\mathcal{I}_i)$: Allocation made to $i^{th}$ IoT device.\\
$\boldsymbol{\rho}$ & $\boldsymbol{\rho} = \{\boldsymbol{\rho}_1, \boldsymbol{\rho}_2, \ldots, \boldsymbol{\rho}_{\mathcal{N}}\}$ : Payment vector of IoT devices selected \\ & for executing the tasks.\\
$\boldsymbol{\rho}_{i}$ &  Represents the payment of any $i^{th}$ IoT device.\\
$u_i(S, \boldsymbol{\rho})$ &  Utility of $i^{th}$ IoT device given the price vector $\boldsymbol{\rho}$.\\
$D_i(\boldsymbol{\rho})$ &  Supply by any $i^{th}$ IoT device given the price vector $\boldsymbol{\rho}$.\\
\hline
\end{tabular}
\end{table}

  \subsection{Additional Required Definitions}
\begin{definition}[\textbf{Truthful or Incentive Compatible (IC) \cite{NNisa_Pre_2007}}] A mechanism is said to be truthful or IC if for any $i^{th}$ IoT device~ $\mathcal{U}_i(c, \boldsymbol{\bar{\rho}})= \boldsymbol{\bar{\rho}}_i - c_i \geq \boldsymbol{\bar{\rho}}_i - \hat{c}_i =\hat{\mathcal{U}}_i(\hat{c}, \boldsymbol{\bar{\rho}})$.
  \end{definition}
  \begin{definition}[\textbf{Budget feasible \cite{singer2012win, Singer_2016}}]
 A mechanism is said to be budget feasible, if the total payment made to the winning initial notifiers is less than equal to the available budget $\mathcal{B}$ $i.e.$  $\displaystyle\sum_{i\in \mathcal{I}} \boldsymbol{\bar{\rho}}_i\leq \mathcal{B}.$ 
  \end{definition}
  
  \begin{definition}[\textbf{Individual rationality \cite{NNisa_Pre_2007}}]
  \label{def:2}
  A mechanism is said to be individually rational if, for every participating IoT device in the crowdsourcing market, the utility is non-negative. In other words $\mathcal{U}_i(c, \boldsymbol{\bar{\rho}}) \geq 0$ (for first tier) or $u_i(S,~\boldsymbol{\rho}) \geq 0$ (for second tier).
    \end{definition}
  
\begin{definition}[\textbf{Social Welfare \cite{NNisa_Pre_2007}}]
\label{def:sw}
It is the sum of the valuations of the IoT devices for their preferred set of tasks. Mathematically, it is given as: 
  \begin{equation*}
  \displaystyle\sum_{i=1}^{\mathcal{N}} v_i(S)
 \end{equation*}
 where, $S$ is the requested set of task from each IoT device $\mathcal{I}_i \in \mathcal{I}$.
 \end{definition}
 \begin{definition}[\textbf{Computational Efficiency}]  
 A  mechanism (a.k.a algorithm) is said to be computationally efficient if each step of the mechanism takes polynomial time.
  \end{definition}
\begin{figure}[H]
\centering
\includegraphics[scale = 0.88]{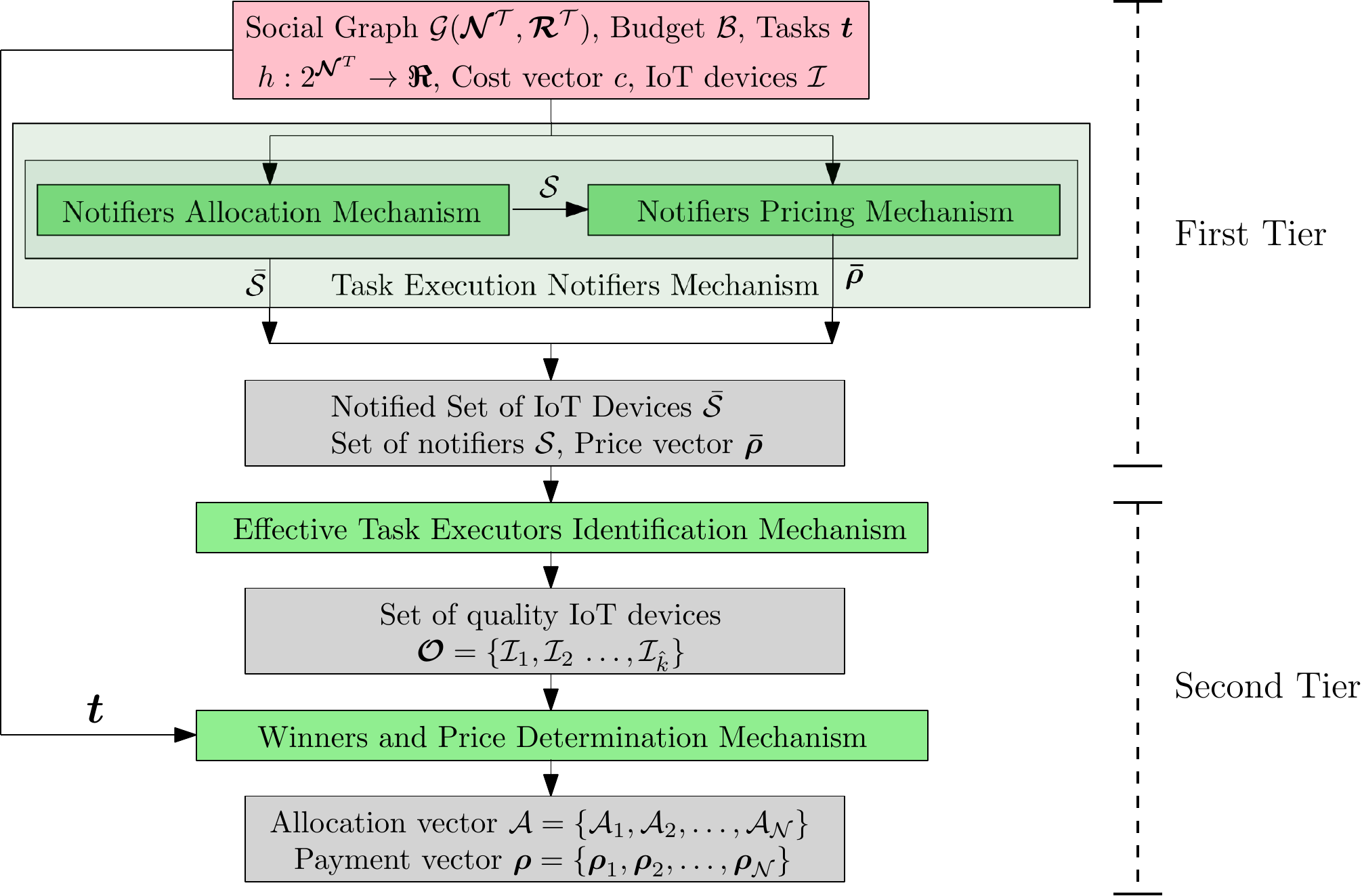}
\caption{Work Flow of Two-tiered Incentive COmpatible Mechanism (TICOM)}
\label{fig:21111}
\end{figure}
\section{\textsc{Two-tiered Incentive Compatible Mechanism (TICOM)}}
\label{section:algorithm}
In this section, TICOM is discussed in a detailed manner. The TICOM consists of:
\begin{itemize}
    \item \textcolor{blue}{\textbf{Task execution notifiers mechanism (TENM)}} $-$ It is useful in determining the set of initial notifiers that will notify a substantial number of IoT devices for the task execution event in the crowdsourcing market and decide their payment. The payment made to the initial notifiers is such that the total payment made to them should be within the available budget.    
    \item \textcolor{blue}{\textbf{Effective task executors identification mechanism (ECTAI)}} $-$ It helps to determine the quality/effective task executors among the available task executors in the crowdsourcing market.
    \item \textcolor{blue}{\textbf{Winners and price determination mechanism (WiPD)}} $-$ Hires the quality IoT devices and decide their payments.
\end{itemize}
The overall workflow of TICOM is depicted in Figure \ref{fig:21111}. 
In the upcoming subsections, each of the components of TICOM is discussed step-by-step in a detailed manner. 
\subsection{\textsc{Task Execution Notifiers Mechanism}}
In this section, a truthful budget feasible mechanism is proposed that is used to figure out the task executors in the given social graph that will propagate the event of task execution to the task executors in their social connection. It will also determine their payment in return for acting as the initial notifier. The task execution notifier mechanism consists of (1) \emph{notifier allocation mechanism} (NAM), and (2) \emph{notifier pricing mechanism} (NPM). The two components are discussed below.

\subsubsection{\textsc{Notifiers Allocation Mechanism} (NAM)}
The input to the NAM is the social graph $\mathcal{G}$, budget $\mathcal{B}$, the cost vector of the task executors $\hat{c}$, and the set of task executors $\mathcal{I}$. The output of NAM is the set of task executors as notifiers and the set of task executors that got notified about the task execution event.

\begin {algorithm}[H]
\caption{\textsc{Notifiers Allocation Mechanism} ($\mathcal{G}$, $\mathcal{B}$, $\hat {c}$, $\mathcal{I})$}
\label{algo:1}
\noindent
\textbf{Output:} $\mathcal{S} \leftarrow \phi$, $\bar{\mathcal{S}} \leftarrow \phi$
\begin{algorithmic}[1]
\FOR{each $\mathcal{I}_i \in \mathcal{I}$}
\STATE $h_{{i}|\mathcal{S}_{i-1}}$ = $h$($\mathcal{S}_{i-1}\cup \{i\}$) - $h(\mathcal{S}_{i-1})$ \COMMENT{\textcolor{blue}{Marginal notification of $i^{th}$ IoT device, given the set $\mathcal{S}_{i-1}$.}}
\ENDFOR
\STATE  $\boldsymbol{\Delta} \gets 2$ \COMMENT{\textcolor{blue}{$\boldsymbol{\Delta}$ is initialized to 2.}}
\STATE $i \gets argmax_{\mathcal{I}_k \in \mathcal{I} \setminus\mathcal {S}} \bigg(\frac{h_{k| \mathcal{S}_{k-1}}}{c_k}\bigg)$ \COMMENT{\textcolor{blue}{Selects an IoT device with maximum marginal notification per cost, given set $\mathcal{S}_{k-1}$.}}
\WHILE{$c_i$ $\leq$ $\frac{\mathcal{B}}{\boldsymbol{\Delta}}$   \bigg($\frac{h_{\mathcal{I}_i| \mathcal{S}_{i-1}}}{(h(\mathcal{S}_{i-1}) + h_{\mathcal{I}_i| \mathcal{S}_{i-1}})}$\bigg)}
\STATE $\mathcal{S} $ $\leftarrow$ $\mathcal{S} \cup \{\mathcal{I}_i\}$ \COMMENT{\textcolor{blue}{ $\mathcal{S}$ holds the set of IoT devices that satisfies the stopping condition in line 6.}}
\STATE$\tilde { \mathcal{S}} \leftarrow \tilde{\mathcal{S}} \cup \{\mathcal {\boldsymbol{\mathcal{Z}}}_i\}$ \COMMENT{\textcolor{blue}{ $\bar{\mathcal{S}}$ holds the set of IoT devices that got notified by the IoT devices in $\mathcal{S}$.}}
\STATE $i \gets argmax_{\mathcal{I}_k \in \mathcal{I} \setminus\mathcal {S}} \bigg(\frac{h_{\mathcal{I}_k| \mathcal{S}_{k-1}}}{c_k}\bigg)$ \COMMENT{\textcolor{blue}{Selects an IoT device with maximum marginal notification per cost among the available ones, given set $\mathcal{S}$.}} 
\ENDWHILE
\STATE return $\mathcal{S}$, $\bar{\mathcal{S}}$ \COMMENT{\textcolor{blue}{Returns the set of task executors as notifiers and the notified set of task executors.}}
\end{algorithmic}
\end{algorithm}
\noindent In Algorithm \ref{algo:1}, lines 1-3 calculate the marginal notification by each of the task executors $\mathcal{I}_i \in \mathcal{I}$. In line 2 the marginal notification of $\mathcal{I}_i$ given an already selected set of task executors as notifiers $\mathcal{S}_i$ is calculated. In line 4 variable $\boldsymbol{\Delta}$ is initialized to 2. In line 5, the task executor with maximum marginal notification per cost is selected among the available ones. In lines 6-10, the idea of determining the task executors as notifiers are presented. In line 6 the stopping condition checks if the cost of $i^{th}$ task executor is less than or equal to $\frac{\mathcal{B}}{\boldsymbol{\Delta}}$ times the ratio of the marginal notification of $i^{th}$ task executor to the number of task executors notified by the set of task executors in set $\mathcal{S} \cup \{\mathcal{I}_i\}$. If the stopping condition in line 6 is true, then in line 7 the task executor $\mathcal{I}_i$ is held in $\mathcal{S}$. In line 8 $\bar{\mathcal{S}}$ holds the set of IoT devices that got notified by the IoT devices in $\mathcal{S}$. For the next iteration, the task executor with the highest marginal notification per cost is selected among the available ones in line 9. Lines 6-10 iterate until the stopping condition in line 6 is true. In line 11 the set of task executors as notifiers and the notified set of task executors are returned.   

\subsubsection{\textsc{Notifiers Pricing Mechanism} (NPM)}
The input to the NPM is the set initial notifier $\mathcal{S}$, budget $\mathcal{B}$, and cost vector $\hat{c}$. The output of the NPM is the price vector $\boldsymbol \bar{\boldsymbol {\rho}}$  of the initial notifiers. In line 1 of Algorithm \ref{algo:2}, $\mathcal{S}'$ is initialized to $\phi$. In lines 2-19 the payment calculation for the \emph{initial notifiers} in set $\mathcal{S}$ is done. In line 3, a set $\boldsymbol{\xi}$ holds the set of all the IoT devices except $\mathcal{I}_j \in \mathcal{S}$. After that, the task executor with the highest marginal notification per cost is selected from $\boldsymbol{\xi}$ and is stored in variable $i$, as shown in line 4. Lines 5-9 determine the set of initial notifiers when task executor $\mathcal{I}_j$ is dragged out of the crowdsourcing market.  Now, any  $i^{th}$ IoT device will be added in the set $\mathcal{S}'$ only when the stopping condition in line 5 is true. Once added, the selected IoT device $\mathcal {I}_i$ will be removed from $\boldsymbol{\xi}$. For the next iteration, a task executor will be selected from the available task executors and will be held in $i$. Lines 5-9 iterate until the stopping condition in line 5 is true. Once terminated, in lines 10-16, the two quantities are calculated $\boldsymbol{\nabla}_{j,k}$ and $\boldsymbol{\rho}_{j,k}$. In line 11, the quantity $\boldsymbol{\nabla}_{j,k}$ is calculated and is held in $\boldsymbol{\nabla}_{j}$ in line 12. The quantity $\boldsymbol{\rho}_{j,k}$ is calculated in line 13 and is held in $\boldsymbol{\rho}_{j}$ in line 14. In line 15, the minimum of two quantities is determined. Line 17 gives us the maximum value present in $\boldsymbol{\rho'}_{j}$ and held in $\boldsymbol{\bar{\rho}}$.
 Finally, in line 20 the payment vector $\boldsymbol{\bar{\rho}}$ is returned.
\begin{example}
\emph{Let us understand the \emph{task execution notifier mechanism} with the help of an example. As discussed earlier  it consists of two components: (1) \emph{notifier allocation mechanism}, and (2) \emph{notifier pricing mechanism}. Both the components are elaborated in the order discussed above}. \emph{The graph shown in Figure \ref{fig:ex1} represents the social connections of the task executors and will be helpful in notifying the substantial number of task executors about the task execution event. The value inside the square box is the cost that will be charged by IoT devices in exchange for notifying the IoT devices about the task execution event.}
\begin{figure}[H]
     \centering
     \begin{subfigure}[h]{0.35\textwidth}
         \centering
         \includegraphics[width=\textwidth]{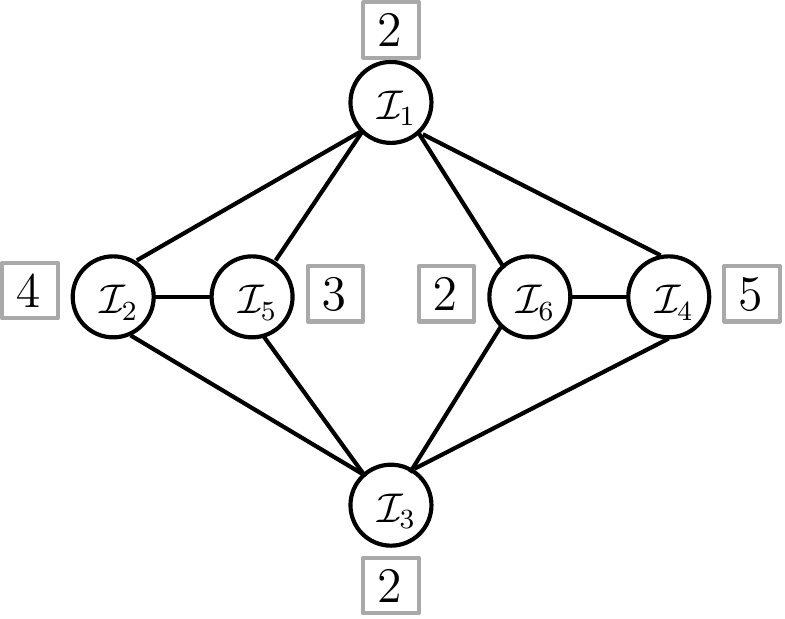}
         \caption{Graph Representing Social Connection of IoT Devices}
         \label{fig:ex1}
     \end{subfigure}
     \hfill
     \begin{subfigure}[h]{0.55\textwidth}
\begin{tabular}{c|c|c|c|c|c|c}
\hline
\textbf{Task Executors} & $\mathcal{I}_1$ & $\mathcal{I}_2$ & $\mathcal{I}_3$ & $\mathcal{I}_4$ & $\mathcal{I}_5$ & $\mathcal{I}_6$\\
\hline
$\boldsymbol{h_{i|S}}$ & 4 & 3 & 4 & 3 & 3 & 3\\
\hline
$\boldsymbol{\hat{c}}$ & 2 & 4 & 2 & 5 & 3 & 2\\
\hline
$\boldsymbol{\frac{h_{i|S}}{\hat{c}}}$ & 2 & 0.75 & 2 & 0.60 & 1 & 1.50\\
\hline
\end{tabular}
\caption{Calculation of Marginal notification given $\mathcal{S} = \phi$, Cost vector, and Marginal notification per cost}
         \label{fig:ex1b}
     \end{subfigure}
     \caption{Initial Set-up for Illustration of NAM}
\end{figure}
\begin {algorithm}[H]
\caption{\textsc{Notifier Pricing Mechanism} ($\mathcal{S}$, $\mathcal{B}$, $\hat {c}$)}
\label{algo:2}
\noindent
\textbf{Output:}  $\boldsymbol{\bar{\rho}} \leftarrow \phi$ 
\begin{algorithmic}[1]
\STATE $\mathcal{S}'\leftarrow \phi$
\FOR{each $ \mathcal {I}_j \in \mathcal{S}$}
\STATE $ \boldsymbol {\xi} \leftarrow \mathcal{I}\setminus \{\mathcal{I}_j\} $ \COMMENT{\textcolor{blue}{Removing $\mathcal{I}_j$ from the market and storing rest of the task executors in $\boldsymbol{\xi}$.}}
\STATE $i \gets \argmax\limits_{\mathcal{I}_k \in \boldsymbol{\xi}}  \bigg(\frac{h_{\mathcal{I}_k|\mathcal{S}'}}{\hat{c}_k}\bigg)$ \COMMENT{\textcolor{blue}{The $k^{th}$ task executor with maximum marginal notification to cost ratio is determined from $\boldsymbol{\xi}$ and is stored in $i$.}}
\WHILE{$\frac{\hat{c}_i}{\mathcal{B}}$ $\leq$    \bigg($\frac{h_{\mathcal{I}_i|\mathcal{S}'}}{h_{\mathcal{I}_i|\mathcal{S}'} + h({\mathcal{S}'})}$\bigg)} 
\STATE $\mathcal{S}' $ $\leftarrow$ $\mathcal{S}' \cup \{\mathcal{I}_i\}$ \COMMENT{\textcolor{blue}{$\mathcal{S}'$ holds the set of IoT devices that satisfies the stopping condition in line 5.}}
\STATE$ \boldsymbol{\xi} \leftarrow \boldsymbol{\xi} \setminus \{{\mathcal{I}_i}\}$ \COMMENT{\textcolor{blue}{$\mathcal{I}_i$ removed from $\boldsymbol{\xi}$.}}
\STATE $i \gets \argmax\limits_{\mathcal{I}_k \in \boldsymbol{\xi}}  \bigg(\frac{h_{\mathcal{I}_k|\mathcal{S}'}}{c_k}\bigg)$ \COMMENT{\textcolor{blue}{Selects an IoT device with maximum marginal notification per cost
among the available ones, given set $\mathcal{S}'$.}}
\ENDWHILE
\FOR{ $k\gets 1$ to $|\mathcal{S}'|+1 $ }
\STATE$\boldsymbol{\nabla}_{j, k}$ $\gets$ $h_{j,k}$ $\cdot$ $\bigg (\frac{c_k}{h'_{{k}|T_{k-1}}}\bigg)$ \COMMENT{\textcolor{blue}{Cost that any $j^{th}$ IoT device would have revealed when considered in place of IoT device that is already present at any $k^{th}$ position.}}
\STATE$\boldsymbol{\nabla}_{j}$ $\gets$ $\boldsymbol{\nabla}_{j}$ $\cup$ \{$\boldsymbol{\nabla}_{j, k}\}$ 
\STATE $\boldsymbol {\rho}_{j,k} $ $\gets$ $\mathcal{B}$ $\cdot$ $\bigg(\frac{ h_{j,k}}{h({T}_{k-1} \cup  \{\mathcal{I}_j\})}\bigg)$ \COMMENT{\textcolor{blue}{The fraction of budget that will be utilized as the payment of the $j^{th}$ IoT device at some $k^{th}$ position.}}
\STATE $\boldsymbol {\rho}_{j} $ $\gets$ $\boldsymbol {\rho}_{j}$ $\cup \{ \boldsymbol {\rho}_{j, k}\}$ \COMMENT{\textcolor{blue}{The payment of any $j^{th}$ IoT device at position $k$ is stored in $\boldsymbol {\rho}_{j}$.}}
\STATE $\boldsymbol {\rho'}_{j} $ $\gets$ $\min \{\boldsymbol{\nabla}_{j}, \boldsymbol {\rho}_{j}$\} \COMMENT{\textcolor{blue}{Minimum of $\boldsymbol{\nabla}_j$ and $\boldsymbol{\rho}_j$ is determined and is stored in $\boldsymbol{\rho}_j'$.}}
\ENDFOR
\STATE $\boldsymbol{\bar{\rho}}_j$ $\gets$ $\max \{\boldsymbol{\rho'}_{j}$\} \COMMENT{\textcolor{blue}{Determining the maximum value from $j = 1$ to $|\mathcal{S}'|+1$ and stored in $\bar{\boldsymbol{\rho}_j}$.}}
\STATE $\boldsymbol{\bar{\rho}} \gets$ $\boldsymbol{\bar{\rho}} \cup \{\boldsymbol{\bar{\rho}}_j\}$ \COMMENT{\textcolor{blue}{Payment of each of the $j^{th}$ IoT device in winning set is determined and is stored in $\bar{\boldsymbol{\rho}}$}}
\ENDFOR
\STATE return $\boldsymbol{\bar{\rho}}$ \COMMENT{\textcolor{blue}{The payment vector for the winning IoT devices is returned.}}
\end{algorithmic}
\end{algorithm}

\emph{From Figure \ref{fig:ex1} it is evident that the task executor $\mathcal{I}_1$ has social connection with task executors $\mathcal{I}_2$, $\mathcal{I}_4$, $\mathcal{I}_5$, and $\mathcal{I}_6$. In the similar way, task executor $\mathcal{I}_2$ has social connection with $\mathcal{I}_1$, $\mathcal{I}_3$, and $\mathcal{I}_5$ and so on. Further calculations in the running example will be done by considering the available budget $\mathcal{B}$ as 12.}
\begin{itemize}
    \item \emph{\textcolor{blue}{\underline{\textsc{\textbf{Notifiers Allocation Mechanism}}}}: Applying Algorithm \ref{algo:1} to  Figure \ref{fig:ex1}, the marginal notification of each of the task executors in the graph shown in Figure \ref{fig:ex1} are calculated given $\mathcal{S} = \phi$. After that, the ratio $h_{i|\mathcal{S}}$ per cost is calculated using line 2 of Algorithm \ref{algo:1}. The calculated values are depicted in Figure \ref{fig:ex1b} in tabular form. Following line 5, one of the two task executors $\mathcal{I}_1$ and $\mathcal{I}_3$ will be considered as both have the highest value for marginal notification per cost but are the same. Let us say $\mathcal{I}_1$ is considered randomly. For task executor $\mathcal{I}_1$ the stopping condition $2 \leq \large(\frac{12}{2}\large) \cdot \frac{4}{4} = 6$ in line 6 is true and is selected. So, $S= \{\mathcal{I}_1\}$, and $\tilde{S} = \{\mathcal{I}_2, \mathcal{I}_4, \mathcal{I}_5, \mathcal{I}_6\}$. In the next iteration, the graph configuration shown in Figure \ref{fig:ex2a} will be considered. In the second iteration, the marginal notification of each of the task executors is calculated given $\mathcal{S} = \{\mathcal{I}_1\}$ as shown in Figure \ref{fig:ex2b}. After that, the ratio $h_{i|\mathcal{S}}$ per cost is calculated. The calculated values are shown in Figure \ref{fig:ex2b}. Following lines 6-10, a task executor $\mathcal{I}_6$ will be considered as it is having the highest marginal notification per cost value among the available IoT devices. For task executor $\mathcal{I}_6$ the stopping condition $2 \leq \large(\frac{12}{2}\large) \cdot \frac{2}{6} = 2$ in line 6 is true and is selected. So, $S= \{\mathcal{I}_1, \mathcal{I}_6\}$, and $\tilde{S} = \{\mathcal{I}_1, \mathcal{I}_2, \mathcal{I}_3, \mathcal{I}_4, \mathcal{I}_5, \mathcal{I}_6\}$.}
    
    \begin{figure}[H]
     \centering
     \begin{subfigure}[h]{0.35\textwidth}
         \centering
         \includegraphics[scale = 0.70]{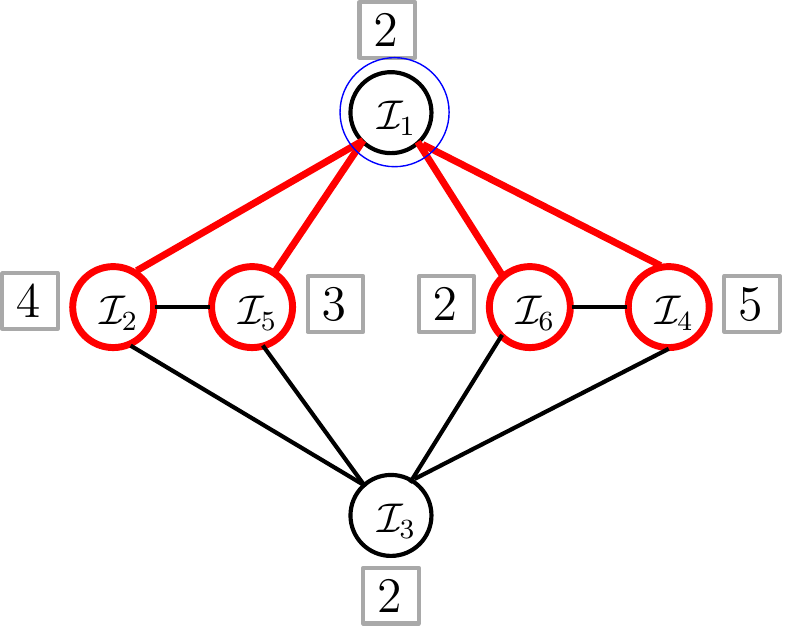}
         \caption{Graph Representing Social Connection of IoT Devices and Selected IoT Device $\mathcal{I}_1$}
         \label{fig:ex2a}
     \end{subfigure}
     \hfill
     \begin{subfigure}[h]{0.55\textwidth}
\begin{tabular}{c|c|c|c|c|c|c}
\hline
\textbf{Task Executors} & $\mathcal{I}_1$ & $\mathcal{I}_2$ & $\mathcal{I}_3$ & $\mathcal{I}_4$ & $\mathcal{I}_5$ & $\mathcal{I}_6$\\
\hline
$\boldsymbol{h_{i|S}}$ & 0 & 2 & 0 & 2 & 2 & 2\\
\hline
$\boldsymbol{\hat{c}}$ & 2 & 4 & 2 & 5 & 3 & 2\\
\hline
$\boldsymbol{\frac{h_{i|S}}{\hat{c}}}$ & 0 & 0.5 & 0 & 0.40 & 0.66 & 1\\
\hline
\end{tabular}
\caption{Calculation of Marginal notification given $\mathcal{S} = \{\mathcal{I}_1\}$, Cost vector, and Marginal notification per cost}
         \label{fig:ex2b}
     \end{subfigure}
     \caption{Illustration of $2^{nd}$ Iteration of While Loop of NAM}
     \label{fig:ex2}
\end{figure}
\emph{For the next iteration the configuration shown in Figure \ref{fig:ex3a} will be considered. In the third iteration, the marginal notification of each of the task executors is calculated given $\mathcal{S} = \{\mathcal{I}_1, \mathcal{I}_6\}$ as shown in Figure \ref{fig:ex3b}. After that, the quantity $h_{i|\mathcal{S}}$ per cost is calculated. The calculated values are shown in Figure \ref{fig:ex3b}. Following lines 6-10, for none of the task executor the stopping condition in line 6 will be satisfied as the marginal notification per cost value for the available task executors are 0. Hence, the \emph{while} loop in lines 6-10 will terminate, and line 11 will return $\mathcal{S} =\{\mathcal{I}_1, \mathcal{I}_6\}$, and $\tilde{S} = \{\mathcal{I}_1, \mathcal{I}_2, \mathcal{I}_3, \mathcal{I}_4, \mathcal{I}_5, \mathcal{I}_6\}$.}
 
    \begin{figure}[H]
     \centering
     \begin{subfigure}[h]{0.35\textwidth}
         \centering
         \includegraphics[scale = 0.70]{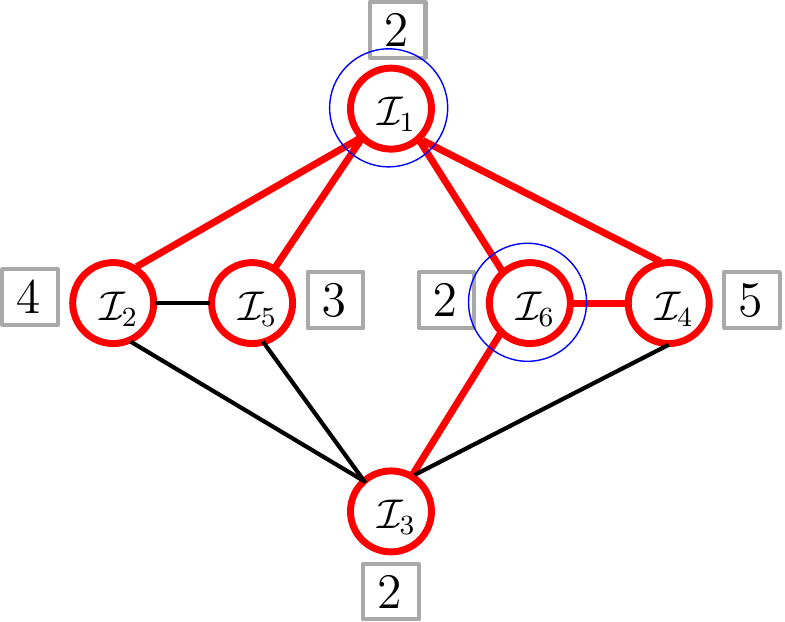}
         \caption{Graph Representing Social Connection of IoT Devices and Selected IoT Device $\mathcal{I}_1$ and $\mathcal{I}_6$}
         \label{fig:ex3a}
     \end{subfigure}
     \hfill
     \begin{subfigure}[h]{0.55\textwidth}
\begin{tabular}{c|c|c|c|c|c|c}
\hline
\textbf{Task Executors} & $\mathcal{I}_1$ & $\mathcal{I}_2$ & $\mathcal{I}_3$ & $\mathcal{I}_4$ & $\mathcal{I}_5$ & $\mathcal{I}_6$\\
\hline
$\boldsymbol{h_{i|S}}$ & 0 & 0 & 0 & 0 & 0 & 0\\
\hline
$\boldsymbol{\hat{c}}$ & 2 & 4 & 2 & 5 & 3 & 2\\
\hline
$\boldsymbol{\frac{h_{i|S}}{\hat{c}}}$ & 0 & 0 & 0 & 0 & 0 & 0\\
\hline
\end{tabular}
\caption{Calculation of Marginal notification given $\mathcal{S} = \{\mathcal{I}_1, \mathcal{I}_6\}$, Cost vector, and Marginal notification per cost}
         \label{fig:ex3b}
     \end{subfigure}
     \caption{Illustration of $3^{rd}$ Iteration of While Loop of NAM}
     \label{fig:ex3}
\end{figure}
\item \textsc{\textcolor{blue}{\underline{\textbf{Notifiers Pricing Mechanism}}:}} \emph{Using Algorithm \ref{algo:2}, the payment of IoT devices $\mathcal{I}_1$ and $\mathcal{I}_6$ will be calculated. From the construction of Algorithm \ref{algo:2}, to calculate payment of $\mathcal{I}_1$, the task executor $\mathcal{I}_1$ will be placed out of the social graph. After dragging out IoT device $\mathcal{I}_1$, the marginal notification of each of the task executors is calculated given nobody is selected $i.e.$ $\mathcal{S}' = \phi$. After that, the ratio $h_{i|\mathcal{S}'}$ per cost is calculated. The calculated values are shown in Figure \ref{fig:ex4b}. Following line 4, a task executor $\mathcal{I}_3$ will be picked up as it is having the maximum marginal notification per cost value among the available IoT devices. The stopping condition $2 \leq 12 \cdot (\frac{4}{4})$ = 12 in line 5 is satisfied. So, $\mathcal{I}_3$ is selected}. \emph{ For the next iteration, the configurations shown in Figure \ref{fig:ex5a} will be considered. In the second iteration, the marginal notification of each of the task executors given $\mathcal{S}' = \{\mathcal{I}_3\}$ is calculated. After that, the ratio $h_{i|\mathcal{S}'}$ per cost is calculated. The calculated values are shown in Figure \ref{fig:ex5b}. Following line 4, a task executor $\mathcal{I}_6$ will be considered as it is having the maximum marginal notification per cost value among the available IoT devices. For task executor $\mathcal{I}_6$ the stopping condition $2 \leq 12 \cdot (\frac{1}{5}) = 2.4$ in line 5 is true and is selected.} 

\begin{figure}[H]
     \centering
     \begin{subfigure}[h]{0.35\textwidth}
         \centering
         \includegraphics[scale = 0.72]{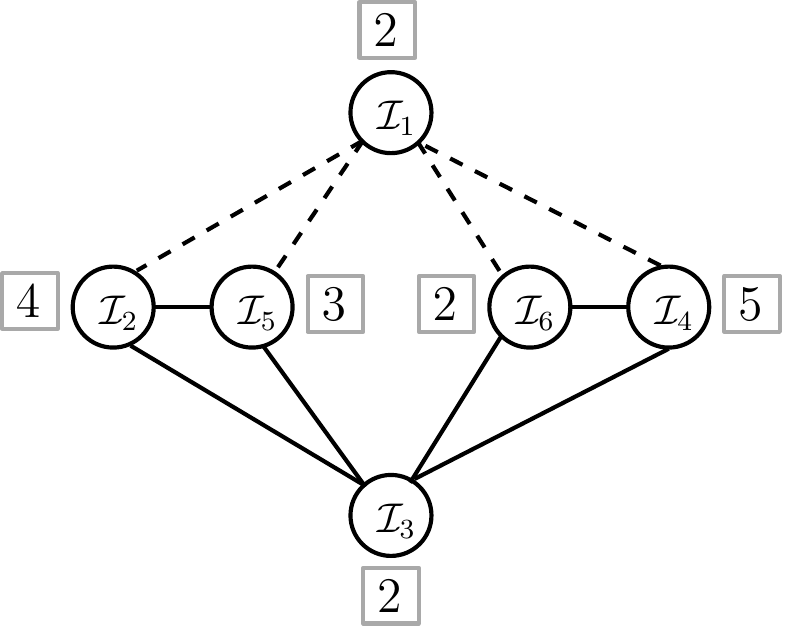}
         \caption{Graph Representing Social Connection of IoT Devices When $\mathcal{I}_1$ is Out of the Market}
         \label{fig:ex4a}
     \end{subfigure}
     \hfill
     \begin{subfigure}[h]{0.55\textwidth}
\begin{tabular}{c|c|c|c|c|c}
\hline
\textbf{Task Executors} & $\mathcal{I}_2$ & $\mathcal{I}_3$ & $\mathcal{I}_4$ & $\mathcal{I}_5$ & $\mathcal{I}_6$\\
\hline
$\boldsymbol{h_{i|S'}}$ & 2 & 4 & 2 & 2 & 2\\
\hline
$\boldsymbol{\hat{c}}$ & 4 & 2 & 5 & 3 & 2\\
\hline
$\boldsymbol{\frac{h_{i|S'}}{\hat{c}}}$ & 0.5 & 2 & 0.4 & 0.66 & 1\\
\hline
\end{tabular}
\caption{Calculation of Marginal notification given $\mathcal{S}' = \phi$, Cost vector, and Marginal notification per cost}
         \label{fig:ex4b}
     \end{subfigure}
     \caption{Illustration of $1^{st}$ Iteration of NPM}
     \label{fig:ex4}
\end{figure}
\emph{In the third iteration, the configuration shown in Figure \ref{fig:ex6b} will be utilized. The marginal notification of each of the task executors given $\mathcal{S}' = \{\mathcal{I}_3, \mathcal{I}_6\}$ are calculated. After that, the ratio $h_{i|\mathcal{S}'}$ per cost is calculated. The calculated values are shown in Figure \ref{fig:ex6b}. For this configuration, none of the task executors will be selected as the marginal notification of all the IoT devices given $\mathcal{S}' = \{\mathcal{I}_3, \mathcal{I}_6\}$ is 0. So, when $\mathcal{I}_1$ is out of the market then $\mathcal{S}' = \{\mathcal{I}_3, \mathcal{I}_6\}$ will be acting as the initial notifiers. As two IoT devices are selected, so $|S'|$ value is 2. In our running example, the first loser can be anyone out of $\mathcal{I}_2$, $\mathcal{I}_4$, and $\mathcal{I}_5$. Let us say $\mathcal{I}_2$. Now, following lines 10-16 of Algorithm \ref{algo:2}, at each index $k \in [1..|S'|+1]$ the maximal cost of $\mathcal{I}_1$ and its payment is determined. The minimum of the two quantities is taken and then the maximum of each of the points will be the payment made by the $\mathcal{I}_1$. Let us say $\mathcal{I}_1$ is considered in place of $\mathcal{I}_3 \in \mathcal{S}'$.}
\begin{figure}[H]
     \centering
     \begin{subfigure}[h]{0.40\textwidth}
         \centering
         \includegraphics[scale = 0.72]{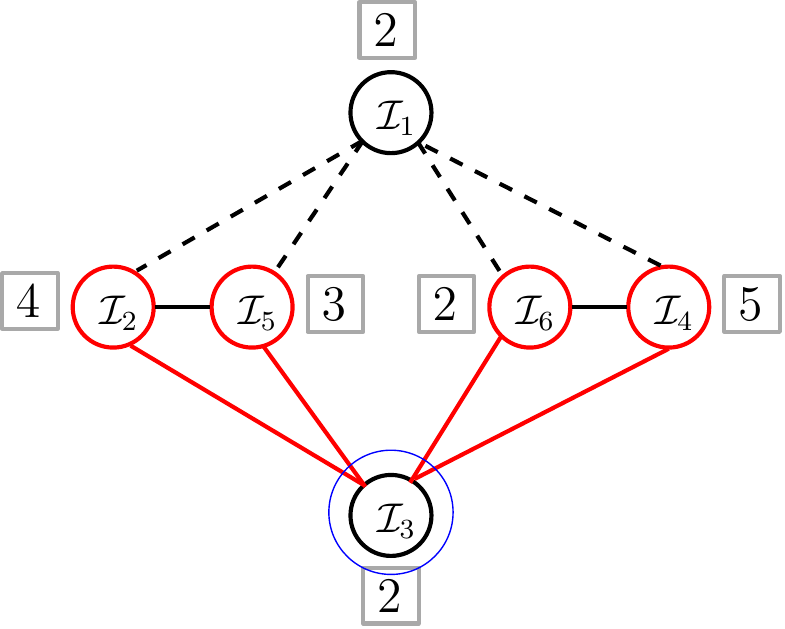}
         \caption{Graph Representing Social Connection of IoT Devices When $\mathcal{I}_1$ is Out of the Market and $\mathcal{I}_3$ got Selected}
         \label{fig:ex5a}
     \end{subfigure}
     \hfill
     \begin{subfigure}[h]{0.55\textwidth}
\begin{tabular}{c|c|c|c|c|c}
\hline
\textbf{Task Executors} & $\mathcal{I}_2$ & $\mathcal{I}_3$ & $\mathcal{I}_4$ & $\mathcal{I}_5$ & $\mathcal{I}_6$\\
\hline
$\boldsymbol{h_{i|S'}}$ & 1 & 0 & 1 & 1 & 1\\
\hline
$\boldsymbol{\hat{c}}$ & 4 & 2 & 5 & 3 & 2\\
\hline
$\boldsymbol{\frac{h_{i|S'}}{\hat{c}}}$ & 0.25 & 0 & 0.2 & 0.33 & 0.5\\
\hline
\end{tabular}
\caption{Calculation of Marginal notification given $\mathcal{S}' = \{\mathcal{I}_3\}$, Cost vector, and Marginal notification per cost}
         \label{fig:ex5b}
     \end{subfigure}
     \caption{Illustration of $2^{nd}$ Iteration of NPM}
     \label{fig:ex5}
\end{figure}
\emph{In such case, $\boldsymbol{\nabla}_{1, 1} = 4 \cdot (\frac{2}{4}) = 2$, and $\boldsymbol{\rho}_{1, 1} = 12 \cdot (\frac{4}{4}) = 12$. So, $\min\{2, 12\} = 2$. Next, if $\mathcal{I}_1$ is considered in place of $\mathcal{I}_6 \in \mathcal{S}'$ given $\mathcal{S}' = \{\mathcal{I}_3\}$. In such case, $\boldsymbol{\nabla}_{1, 2} = 0 \cdot (\frac{2}{4}) = 0$, and $\boldsymbol{\rho}_{1, 2} = 12 \cdot (\frac{0}{4}) = 0$. So, $\min\{0, 0\} = 0$. Finally, $\mathcal{I}_1$ is considered in place of $\mathcal{I}_2$ then $\boldsymbol{\nabla}_{1, 3} = 0 \cdot (\frac{4}{0}) = 0$, and $\boldsymbol{\rho}_{1, 3} = 12 \cdot (\frac{0}{4}) = 0$. So, $\min\{0, 0\} = 0$. Following line 17, we get $\max\{2, 0, 0\} = 2$. So, the payment of $\mathcal{I}_1 =2$. Similarly, the payment of $\mathcal{I}_6$ can be calculated and is given as 3. So, $\boldsymbol{\bar{\rho}} = \{2, 3\}$. The total payment is $2+3 = 5 \leq 12$.} 
\begin{figure}[H]
     \centering
     \begin{subfigure}[h]{0.35\textwidth}
         \centering
         \includegraphics[width=\textwidth]{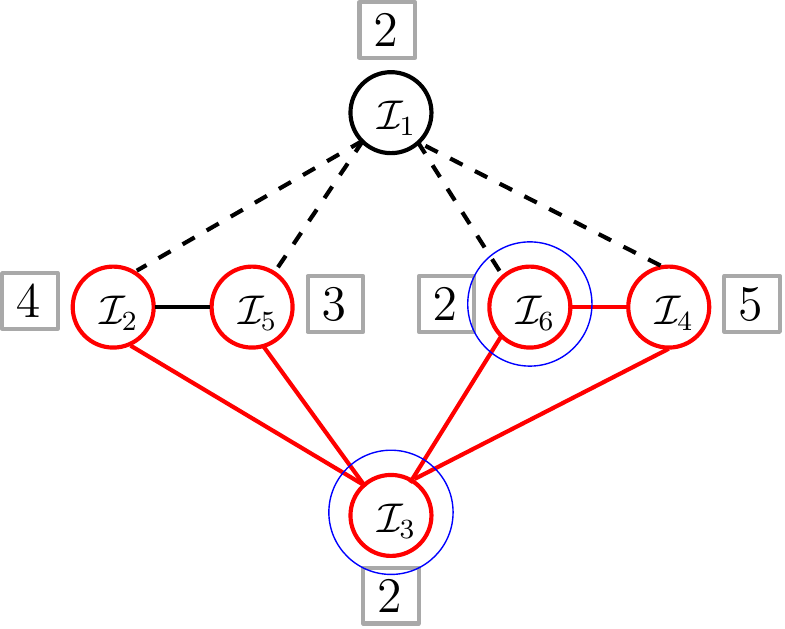}
         \caption{Graph Representing Social Connection of IoT Devices When $\mathcal{I}_1$ is Out of the Market and $\mathcal{I}_3$ and $\mathcal{I}_6$ got Selected}
         \label{fig:ex5a11}
     \end{subfigure}
     \hfill
     \begin{subfigure}[h]{0.55\textwidth}
\begin{tabular}{c|c|c|c|c|c}
\hline
\textbf{Task Executors} & $\mathcal{I}_2$ & $\mathcal{I}_3$ & $\mathcal{I}_4$ & $\mathcal{I}_5$ & $\mathcal{I}_6$\\
\hline
$\boldsymbol{h_{i|S'}}$ & 0 & 0 & 0 & 0 & 0\\
\hline
$\boldsymbol{\hat{c}}$ & 4 & 2 & 5 & 3 & 2\\
\hline
$\boldsymbol{\frac{h_{i|S'}}{\hat{c}}}$ & 0 & 0 & 0 & 0 & 0\\
\hline
\end{tabular}
\caption{Calculation of Marginal notification given $\mathcal{S}' = \{\mathcal{I}_3, \mathcal{I}_6\}$, Cost vector, and Marginal notification per cost}
         \label{fig:ex6b}
     \end{subfigure}
     \caption{Illustration of $3^{rd}$ Iteration of NPM}
     \label{fig:ex611}
\end{figure}
\end{itemize}
\end{example}

\subsection{\textsc{Effective Task Executors Identification Mechanism (ECTAI)}}
\label{sec:42}
Once the substantial number of IoT devices from the first tier is obtained, in the second tier, the first objective is to determine effective (or quality) IoT devices from among the available IoT devices. For this purpose, the idea of \emph{single-peaked preference} \cite{T.roughgarden_2016, 10.2307/30023824} is utilized. The general idea of the proposed mechanism $i.e.$ ECTAI is:

\begin{mdframed}[backgroundcolor=gray230]
\begin{center}\textbf{\underline{ECTAI}}\end{center}
\noindent In each iteration of the \emph{while} loop:
\begin{enumerate}
\item Firstly, some $f$ number of IoT devices are selected randomly from the set of IoT devices and are placed on the scale of $[0,1]$.
\item After that, an infinitesimally small part of tasks are given to those IoT devices that are placed on the scale of $[0,1]$.
\item The executed tasks of $f$ number of IoT devices are reviewed by $g$ number of other IoT devices and based on that the peak values are reported by $g$ number of IoT devices. 
\item Calculate the median of the peak values reported by $g$ number of IoT devices.
\item Determine the IoT device among the IoT devices that are placed on the scale of $[0,1]$ whose peak value lies closer to the median peak value. It will be considered a quality IoT device. 
\end{enumerate}
\end{mdframed}
Steps 1-5 are continued until the ranking tasks of each of the task executors are not ranked. The detailing of the above-discussed approach is presented in Algorithm \ref{algo:3}. The input to algorithm \ref{algo:3} is: (1) the set of IoT devices that are acting as the initial notifiers $i.e.$ $\mathcal{S}$, and (2) the set of IoT devices that got notified $i.e.$ $\bar{\mathcal{S}}$. In Algorithm \ref{algo:3}, line 1 initializes $\mathcal{I}''$ and $\mathcal {I}'$ by  $\mathcal {S} \cup \bar{\mathcal {S}}$, $\mathcal{R}$ to $\phi$, and  $\hat{\mathcal{R}}$ to 0. Lines 2-18 determine the quality of task executors. Lines 2-18 iterate until the condition in line 2 is satisfied. In line 3, the $f$ task executors that are to be ranked are randomly selected from the set $\mathcal {I}'$ and are held in $\eta$. From set $\mathcal{I}''\setminus \eta$, some task executors are randomly selected and are held in $\alpha$. The task executors in $\alpha$ rank the executed tasks of the task executors in $\eta$. For the ranking purpose, the idea of single peaked preference is utilized. The executed tasks by the task executors in $\eta$ are placed on a scale of 0 to 1 randomly.
\begin {algorithm}[!ht]
\caption{\textsc{Effective Task Executors Identification} ($\mathcal {S}$, $\bar{\mathcal{S}}$)}
\label{algo:3}
\noindent
\textbf{Output:}  $\boldsymbol{\mathcal{O}}  \leftarrow \phi $
\begin{algorithmic}[1]
\STATE $\mathcal{I}'' = \mathcal {I}' = \mathcal {S} \cup \bar{\mathcal{S}} ,R \gets \phi, \hat{R} \gets 0 $ 
\WHILE{$\mathcal{I}' \neq \phi$ }
\STATE $\eta \gets$ \textsc{Random Selection} ($\mathcal {I}', f)$ \COMMENT{\textcolor{blue}{Selects \textit{f} task executors from $\mathcal{I}'$ that are to be ranked.}}
\STATE $\alpha \gets$  \textsc{Random Selection} ($\mathcal{I}''\setminus \eta,~ g$) 
\COMMENT{\textcolor{blue}{Selects $g$ task executors from $\mathcal{I}'' \setminus \eta$ that will provide ranking over the completed tasks of task executors in $\eta$.}}
\FOR{each $\mathcal{I}_i \in \alpha$}
\STATE $\alpha_i \gets$ \textsc{Random}~$(0,1)$
\COMMENT{\textcolor{blue}{Returns a random number between 0 and 1.}}
\STATE $\mathcal{R}$ $\gets$  $\mathcal{R}$ $\cup \{\alpha_i\}$ \COMMENT{\textcolor{blue}{Generated random number is held in $\mathcal{R}$.}}
\ENDFOR
\STATE sort ($\mathcal{R}$)
\COMMENT{\textcolor{blue}{Sort the peak values present in $\mathcal{R}$, in ascending order.}}
\IF{$(|\alpha|$ mod $2) \neq 0 $ }
\STATE $\hat{\mathcal{R}} \gets \mathcal{R} \Big[\frac{|\alpha|+1}{2}\Big]$ \COMMENT{\textcolor{blue}{When the number of peak values is odd then the median is calculated and is stored in $\hat{\mathcal{R}}$.}}
\ELSE
\STATE $\hat{\mathcal{R}}$ $\gets$ $\Bigg(\frac{\mathcal{R} \Big[\frac{|\alpha|}{2}\Big] + \mathcal{R} \Big[\frac{|\alpha|}{2}+1\Big]}{2}\Bigg)$ \COMMENT{\textcolor{blue}{When the number of peak values is even then the median is calculated and is stored in $\hat{\mathcal{R}}$.}}
\ENDIF
\STATE $i \gets \argmin\limits_{\mathcal{I}_i \in \eta}  |\alpha_i- \hat{\mathcal{R}}|$ \COMMENT{\textcolor{blue}{Determines the nearest IoT device to the median peak value.}}
\STATE $\boldsymbol{\mathcal{O}} \gets \boldsymbol{\mathcal{O}} \cup \{i\}$ \COMMENT{\textcolor{blue}{The nearest IoT device to the median peak value is held in $\boldsymbol{\mathcal{O}}$.}}
\STATE $\mathcal{I}' \gets \mathcal{I}'\setminus \eta $ \COMMENT{\textcolor{blue}{Removes the already ranked IoT devices from $\mathcal{I}'$.}}
\ENDWHILE
\STATE return $\boldsymbol{\mathcal{O}}$ \COMMENT{\textcolor{blue}{Returns the set of quality IoT devices.}}
\end{algorithmic}
\end{algorithm}
  Now, each task executor in $\alpha$ will provide a peak (or value) between 0 and 1. The peak of each of the task executors is stored in $\mathcal{R}$, as depicted in lines 5-8. In line 9 the peak value by each of the task executors in $\alpha$ is sorted in ascending order. In lines 10-14, the resultant peak value is determined. Here, the two cases can happen: (1) $|\alpha|$ could be even, in that case, the resultant peak is calculated by line 11 and is stored in $\hat{R}$, (2) $|\alpha|$ could be odd, in that case, the resultant peak value is calculated by line 13 and is stored in $\hat{R}$. In line 15, the IoT device in $\eta$ closer to the resultant peak is returned and is stored in $i$. The closer the IoT device in $\eta$ to the resultant peak value of the IoT devices in $\alpha$ better will be the quality of the IoT device $\mathcal {I}_i \ in \eta$. On the other hand, the farther the resultant peak value from the peak value of the task executor $\mathcal{I_j}$, the poorer will be the quality of the task executor  $\mathcal {I}_j$. In line 16, the task executor $\mathcal{I}_i$, with high quality is placed in $\boldsymbol{\mathcal O}$ in each iteration. In line 17, the task executors that are already ranked are removed from the set $\mathcal {I}'$. The \emph{while} loop in lines 2-18 will iterate until all the IoT devices got ranked. Line 19 returns the set of quality IoT devices. 
\begin{example}
Let us consider an example to understand Algorithm \ref{algo:3}. Let us suppose that there are 12 IoT devices $ \mathcal{I}$ = $\{\mathcal{I}_1, \mathcal{I}_2, \ldots, \mathcal{I}_{12}\}$ that got notified about the task execution process by their social connection. Now, out of 12 IoT devices, the objective is to select the subset of quality IoT devices. For that purpose, let us apply Algorithm \ref{algo:3}. For simplicity purposes, we have considered $f=3$, and $g=5$.\\
\indent In the first iteration of \emph{while} loop of Algorithm \ref{algo:3}, $\eta = \{\mathcal{I}_2, \mathcal{I}_4, \mathcal{I}_9\}$ are considered that are to be ranked and $\alpha = \{\mathcal{I}_1, \mathcal{I}_{3},\mathcal{I}_7,\mathcal{I}_8, \mathcal{I}_{11}\}$ are the IoT devices that will provide ranking on the set $\eta$ as shown in Figure \ref{fig:2a}. The peak values reported by the IoT devices in set $\alpha$ are depicted in the table shown in Figure \ref{fig:2a}. Following lines 10-14 of Algorithm \ref{algo:3}, we get $|\alpha|=5$. So, line 10 of Algorithm \ref{algo:3} is true. Using line 11, $\hat{\mathcal{R}} \gets \mathcal{R} \Big[\frac{5+1}{2}\Big]$ = $\hat{\mathcal{R}} \gets \mathcal{R} \Big[\frac{6}{2}\Big]$ = $\hat{\mathcal{R}} \gets \mathcal{R} \Big[{3}\Big]$ = $\hat{\mathcal{R}} \gets 0.50 = \alpha_3$. Using line 15, we get $i\gets \mathcal{I}_4$. So, $\boldsymbol{\mathcal O} = \{{\mathcal{I}_4}\}$. In the next iteration of \emph{while} loop of Algorithm \ref{algo:3}, $\eta = \{\mathcal{I}_1, \mathcal{I}_3, \mathcal{I}_7\}$ are considered that are to be ranked and $\alpha = \{\mathcal{I}_2, \mathcal{I}_{4},\mathcal{I}_6,\mathcal{I}_8, \mathcal{I}_{11}\}$ are the IoT devices that will provide ranking on the set $\eta$ as shown in Figure \ref{fig:2b}. The peak values reported by the IoT devices in set $\alpha$ are depicted in the table shown in Figure \ref{fig:2b}. Following lines 10-14 of Algorithm \ref{algo:3}, we get $|\alpha|=5$. So, line 10 of Algorithm \ref{algo:3} is true. Using line 11, $\hat{\mathcal{R}} \gets \mathcal{R} \Big[\frac{5+1}{2}\Big]$ = $\hat{\mathcal{R}} \gets \mathcal{R} \Big[\frac{6}{2}\Big]$ = $\hat{\mathcal{R}} \gets \mathcal{R} \Big[{3}\Big]$ = $\hat{\mathcal{R}} \gets 0.50 = \alpha_2$. Using line 15, we get $i\gets \mathcal{I}_3$. So, $\boldsymbol{\mathcal O} = \{{\mathcal{I}_4, \mathcal{I}_3}\}$. In the third iteration of \emph{while} loop of Algorithm \ref{algo:3}, $\eta = \{\mathcal{I}_6, \mathcal{I}_{10}, \mathcal{I}_{12}\}$ are considered that are to be ranked and $\alpha = \{\mathcal{I}_2, \mathcal{I}_{7},\mathcal{I}_9,\mathcal{I}_{10}, \mathcal{I}_{11}\}$ are the IoT devices that will provide ranking on the set $\eta$ as shown in Figure \ref{fig:2c}. The peak values reported by the IoT devices in set $\alpha$ are depicted in the table shown in Figure \ref{fig:2c}. Following lines 10-14 of Algorithm \ref{algo:3}, we get $|\alpha|=5$. So, line 10 of Algorithm \ref{algo:3} is true. Using line 11, $\hat{\mathcal{R}} \gets \mathcal{R} \Big[\frac{5+1}{2}\Big]$ = $\hat{\mathcal{R}} \gets \mathcal{R} \Big[\frac{6}{2}\Big]$ = $\hat{\mathcal{R}} \gets \mathcal{R} \Big[{3}\Big]$ = $\hat{\mathcal{R}} \gets 0.45 = \alpha_2$. Using line 15, we get $i\gets \mathcal{I}_{10}$. So, $\boldsymbol{\mathcal O} = \{{\mathcal{I}_3, \mathcal{I}_4, \mathcal{I}_{10}}\}$. In the final iteration of \emph{while} loop of Algorithm \ref{algo:3}, $\eta = \{\mathcal{I}_8, \mathcal{I}_{11}, \mathcal{I}_{9}\}$ are considered that are to be ranked and $\alpha = \{\mathcal{I}_1, \mathcal{I}_{5},\mathcal{I}_7,\mathcal{I}_{9}, \mathcal{I}_{10}\}$ are the IoT devices that will provide ranking on the set $\eta$ as shown in Figure \ref{fig:2d}. The peak values reported by the IoT devices in set $\alpha$ are depicted in the table shown in Figure \ref{fig:2d}. Following lines 10-14 of Algorithm \ref{algo:3}, we get $|\alpha|=5$. So, line 10 of Algorithm \ref{algo:3} is true. Using line 11, $\hat{\mathcal{R}} \gets \mathcal{R} \Big[\frac{5+1}{2}\Big]$ = $\hat{\mathcal{R}} \gets \mathcal{R} \Big[\frac{6}{2}\Big]$ = $\hat{\mathcal{R}} \gets \mathcal{R} \Big[{3}\Big]$ = $\hat{\mathcal{R}} \gets 0.35 = \alpha_9$. Using line 15, we get $i\gets \mathcal{I}_{8}$. So, $\boldsymbol{\mathcal O} = \{{\mathcal{I}_3, \mathcal{I}_4, \mathcal{I}_{8}, \mathcal{I}_{10}}\}$. As all the IoT devices in our running example are ranked, the Algorithm \ref{algo:3} will terminate by returning the set of quality IoT devices $\boldsymbol{\mathcal O} = \{{\mathcal{I}_3, \mathcal{I}_4, \mathcal{I}_{8}, \mathcal{I}_{10}}\}$. 

\begin{figure}[!htbp]
     \centering
     \begin{subfigure}[h]{0.48\textwidth}
         \centering
         \includegraphics[scale = 0.90]{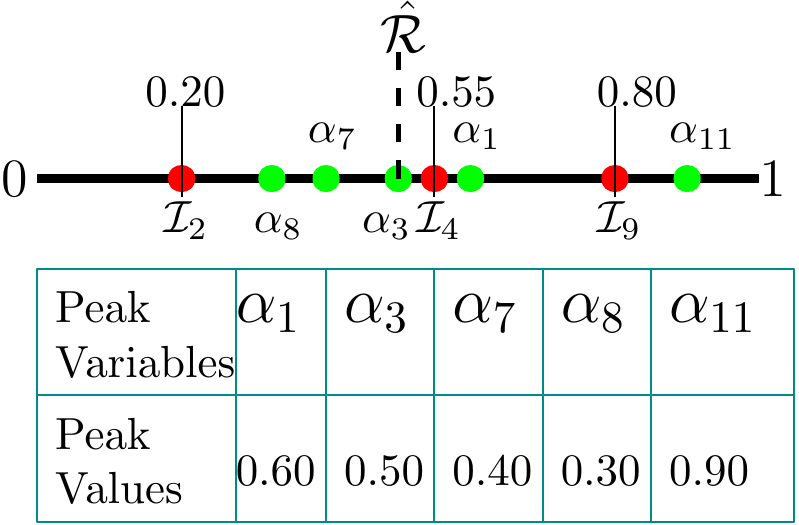}
         \caption{IoT devices represented on $[0,1]$ scale}
         \label{fig:2a}
     \end{subfigure}
     \begin{subfigure}[h]{0.49\textwidth}
         \centering
         \includegraphics[scale = 0.90]{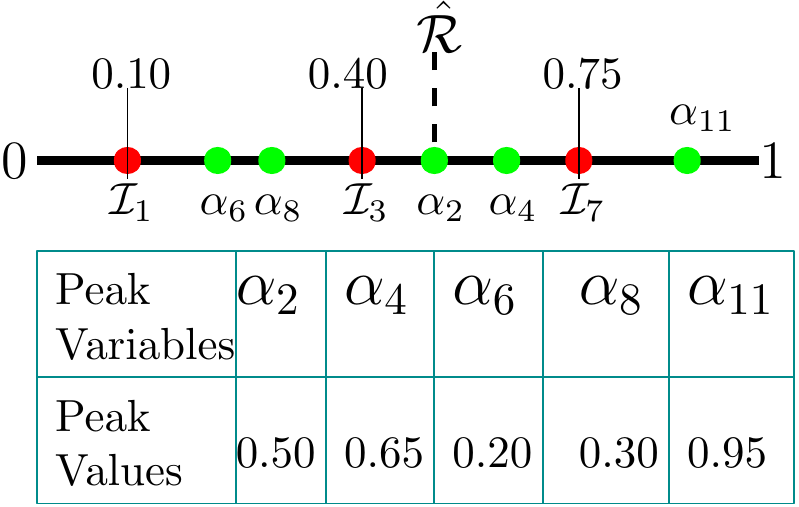}
         \caption{Peak values of $\mathcal{I}_j$ and $\mathcal{I}_k$}
         \label{fig:2b}
     \end{subfigure}
      \begin{subfigure}[h]{0.48\textwidth}
         \centering
         \includegraphics[scale = 0.90]{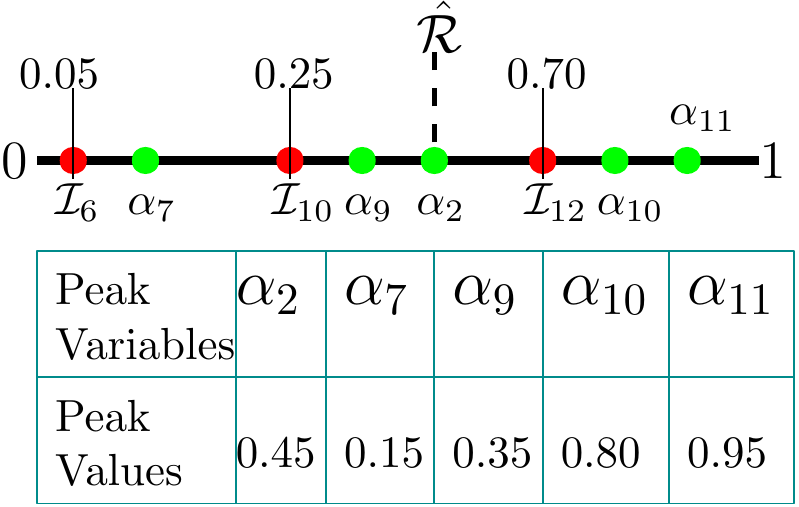}
         \caption{Task executors represented on $[0,1]$ scale}
         \label{fig:2c}
     \end{subfigure}
      \begin{subfigure}[h]{0.48\textwidth}
         \centering
         \includegraphics[scale = 0.90]{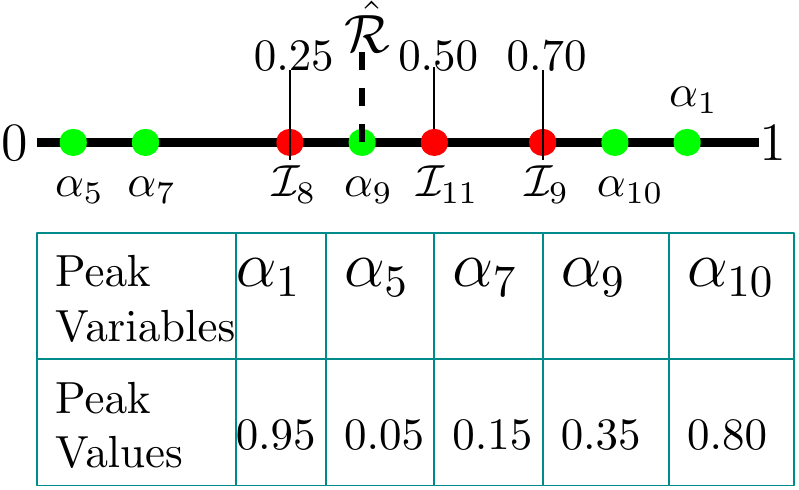}
         \caption{Task executors represented on $[0,1]$ scale}
         \label{fig:2d}
     \end{subfigure}
     \caption{Detailed Illustration of ECTAI}
     \label{fig:2}
\end{figure}
\end{example}

\subsection{Winners and Price Determination (WiPD)}
In subsection \ref{sec:42} using Algorithm \ref{algo:3}, the quality of IoT devices is determined. In this section, it is discussed that: (1) \emph{how the tasks will be allocated to the quality IoT devices?} and (2) \emph{what will be their payment?} The input to the Algorithm \ref{algo:14} is the set of quality IoT devices and the set of heterogeneous tasks. In line 1, the allocation and price vectors are set to $\phi$. In lines 2-4, the prices for all the tasks are set to 0. Using lines 5-7, the initial allocation and initial price for all the IoT devices are set to empty and 0 respectively.
\begin {algorithm}[!ht]
\caption{\textsc{Winners and Price Determination} ($\boldsymbol{\mathcal{O}}$, $\boldsymbol{t}$)} 
\label{algo:14}
\noindent
\begin{algorithmic}[1]
\STATE $\mathcal{A} \gets \phi$, $\boldsymbol {\rho} \gets \phi$ \COMMENT{\textcolor{blue}{Initially, the allocation  and price vectors are set to $\phi$.}}
\FOR{each $\boldsymbol{t_i} \in \boldsymbol{t}$}
\STATE $\boldsymbol {\rho}(i) \gets 0$ \COMMENT{\textcolor{blue}{Initially, the price of all the tasks is set to 0.}}
\ENDFOR
\FOR{each $\mathcal{I}_i \in \boldsymbol{\mathcal{O}}$ } 
\STATE $\mathcal{A}_i \gets \phi$, $\boldsymbol{\rho}_i \gets 0$ \COMMENT{\textcolor{blue}{Initially, the allocation and payment vectors of any $i^{th}$ IoT device is set to $\phi$ and $0$ respectively.}}
\ENDFOR
\WHILE{(True)}
\FOR{each $\mathcal{I}_i \in \boldsymbol{\mathcal{O}}$}
\STATE Ask for the preferred set of tasks not assigned to it, given the tasks they already have and the current prices $-$ an arbitrary set $\mathcal{F}_i$ in\\
\hspace*{15mm} $\argmin\limits_{\mathcal{F}_i \subseteq \boldsymbol{t} \setminus S_i }$ ${\Bigg\{\Big(\displaystyle\sum_{j\in S_i}}$  $\boldsymbol \rho(j)$ + ${\displaystyle\sum_{j\in \mathcal{F}_i } (\boldsymbol \rho(j)+\epsilon)\Big) }$ - $v_i(S_i \cup \mathcal{F}_i)\Bigg\}$ 
  \IF{$ \mathcal{F}_i=\phi $ }
  \STATE $\mathcal{A}_i \leftarrow \mathcal{S}_i$ \COMMENT{\textcolor{blue}{$\mathcal{A}_i$ holds the set of tasks assigned to $i^{th}$ IoT device.}}
\STATE $\mathcal{A} \leftarrow \mathcal{A} \cup \{\mathcal{A}_i\}$ \COMMENT{\textcolor{blue}{$\mathcal{A}$ holds the selected tasks of the respective IoT device.}}
\STATE $\boldsymbol{\rho}_i \leftarrow \displaystyle\sum_{i \in \mathcal{S}_i} \boldsymbol{\rho}(i)$ \COMMENT{\textcolor{blue}{The prices of all the tasks requested by IoT devices in $\mathcal{S}_i$ are added and is stored in $\boldsymbol{\rho}_i$.}}
\STATE $\boldsymbol {\rho} \leftarrow \boldsymbol{\rho} \cup \{\boldsymbol{\rho}_i\}$ \COMMENT{\textcolor{blue}{$\boldsymbol{\rho}$ holds the payment of the winning IoT devices.}}
\ELSE
\STATE  $\mathcal{S}_i \gets  \mathcal{S}_i \cup \mathcal{F}_i$ \COMMENT{\textcolor{blue}{The new set of tasks $\mathcal{F}_i$ is added to $\mathcal{S}_i$ and is stored in $\mathcal{S}_i$.}}
\STATE  ${S_l \gets S_l \setminus \mathcal{F}_i}$, $\forall l \neq i$ \COMMENT{\textcolor{blue}{$\mathcal{F}_i$ set of tasks is removed from the requested set of other task executors.}}\STATE $\boldsymbol{\rho}(j) \leftarrow \boldsymbol{\rho}(j) + \epsilon$, $\forall j \in \mathcal{F}_i$ \COMMENT{\textcolor{blue}{The prices of all the tasks held by IoT devices in $\mathcal{F}_i$ is increased by $\epsilon$.}}
\STATE $\boldsymbol{\rho}_i \leftarrow \displaystyle\sum_{i \in \mathcal{S}_i}\boldsymbol{\rho}(i)$ \COMMENT{\textcolor{blue}{The prices of all the tasks requested by IoT devices in $\mathcal{S}_i$ are added and is stored in $\boldsymbol{\rho}_i$.}}
\STATE $\mathcal{A}_i \leftarrow \mathcal{A}_i \cup \mathcal{S}_i$ \COMMENT{\textcolor{blue}{$\mathcal{A}_i$ holds the set of tasks assigned to IoT device $\mathcal{I}_i$.}}
\STATE $\mathcal{A} \leftarrow \mathcal{A} \cup \mathcal{A}_i$ and $\boldsymbol{\rho} \leftarrow \boldsymbol{\rho} \cup \boldsymbol{\rho}_i$ \COMMENT{\textcolor{blue}{$\mathcal{A}$ and $\boldsymbol{\rho}$ holds the set of assigned tasks of IoT devices and their payments respectively.}}
\ENDIF
\ENDFOR
\ENDWHILE
\STATE return $\mathcal{A}$, $\boldsymbol{\rho}$ \COMMENT{\textcolor{blue}{Returns $\mathcal{A}$ and $\boldsymbol{\rho}$ that holds the set of assigned tasks of IoT devices and their payments respectively.}}
\end{algorithmic}
\end{algorithm}
In \emph{while} loop in lines 8-25, it is asked from the IoT devices that \emph{given the tasks you already have at the given prices of the tasks, what set of tasks, in addition, would you want to bid on?} Now, if with the increase in price, it is seen that no IoT device is interested to modify its requested set of tasks then the \emph{while} loop in lines 8-25 terminates and the current allocations and the current payment vectors are returned in line 26. In line 10 on the other hand, it may happen that with the increase in price, some of the task executors may be ready to show interest in an additional set of tasks. If that is the case, then lines 16-23 of Algorithm \ref{algo:14} will be activated. In line 17, the overall demand of $i^{th}$ task executor is stored in $\mathcal{S}_i$. In line 18, the set of tasks $\mathcal{F}_i$ that got added in $\mathcal{S}_i$ is removed from the demand set of other IoT devices except for the demand set of $\mathcal{I}_i$. The price of the tasks in $\mathcal{F}_i$ is increased by $\epsilon$ again and is stored in $\boldsymbol{\rho}(j)$ in line 19. In line 20, the sum of the prices of tasks in $\mathcal{S}_i$ is held in $\boldsymbol{\rho}_i$. In line 21, $\mathcal{A}_i$ holds the set of tasks assigned to IoT device $\mathcal{I}_i$. In line 22, the tasks allocated to all the task executors are held in $\mathcal{A}$ and the payment vector of all the winning task executors is determined. Line 26 returns the final allocation and the payment vector.        

 \subsection{Illustrative Example}
In this subsection, WiPD is elaborated in a detailed manner with the help of an example. Let us say there are 8 heterogeneous tasks $\boldsymbol{t} = \{\boldsymbol{t}_1, \boldsymbol{t}_2, \boldsymbol{t}_3, \boldsymbol{t}_4, \boldsymbol{t}_5, \boldsymbol{t}_6, \boldsymbol{t}_7, \boldsymbol{t}_8\}$ and 3 quality IoT devices. Following lines 2-4 of Algorithm \ref{algo:14} the prices of the tasks in $\boldsymbol{t}$ are set to 0. In our running example, the $\epsilon$ value is taken as 1. Using lines 8-25 of Algorithm \ref{algo:14}, firstly IoT device $\mathcal{I}_1$ is asked that at price $\boldsymbol{\rho}(i) = 0$, for all $i = 1~to~8$, \emph{what are the tasks that you want to execute}? For $\mathcal{I}_1$ we have $v_1(\mathcal{S}_1) = 6$, where $S_1 = \{\boldsymbol{t}_1, \boldsymbol{t}_2, \boldsymbol{t}_3\}$. Next, the prices of tasks $\boldsymbol{t}_1$, $\boldsymbol{t}_2$, and $\boldsymbol{t}_3$ are increased by $\epsilon = 1$. So, at price $\boldsymbol{\rho}(i) = 1$, for $i = 1, 2,$ and $3$, and $\boldsymbol{\rho}(i) = 0$, for $i = 4, 5, 6, 7$ and $8$, the preferred set of tasks by $\mathcal{I}_2$ is asked. For $\mathcal{I}_2$ we have $v_2(\mathcal{S}_2) = 4$, where $S_2 = \{\boldsymbol{t}_4, \boldsymbol{t}_6, \boldsymbol{t}_8\}$. So, in the second iteration the tasks $\boldsymbol{t}_4$, $\boldsymbol{t}_6$, and $\boldsymbol{t}_8$ are given to $\mathcal{I}_2$. Next, the prices of the tasks $\boldsymbol{t}_4$, $\boldsymbol{t}_6$, and $\boldsymbol{t}_8$ are increased by $\epsilon = 1$ and prices became 1 for these tasks. So, at price $\boldsymbol{\rho}(i) = 1$, for $i = 1, 2, 3, 4, 6,$ and $8$, and $\boldsymbol{\rho}(i) = 0$, for $i = 5, 7$, the preferred set of tasks by $\mathcal{I}_3$ is asked. For $\mathcal{I}_3$ we have $v_3(\mathcal{S}_3) = 3$, where $S_3 = \{\boldsymbol{t}_5\}$. So, in the third iteration, the task $\boldsymbol{t}_5$ is given to $\mathcal{I}_3$. Next, the price of task  $\boldsymbol{t}_5$ is increased by $\epsilon = 1$ and the price became 1 for task $\boldsymbol{t}_5$. After the first iteration of the while loop, the allocation vector $\mathcal{A} = \{(\{\boldsymbol{t}_1, \boldsymbol{t}_2, \boldsymbol{t}_3\}, \mathcal{I}_1), (\{\boldsymbol{t}_4, \boldsymbol{t}_6, \boldsymbol{t}_8\}, \mathcal{I}_2), (\{\boldsymbol{t}_5\}, \mathcal{I}_3)\}$ and payment vector $\boldsymbol{\rho} = \{3, 3, 1\}$.
\begin{figure}[H]
\centering
\includegraphics[scale=0.85]{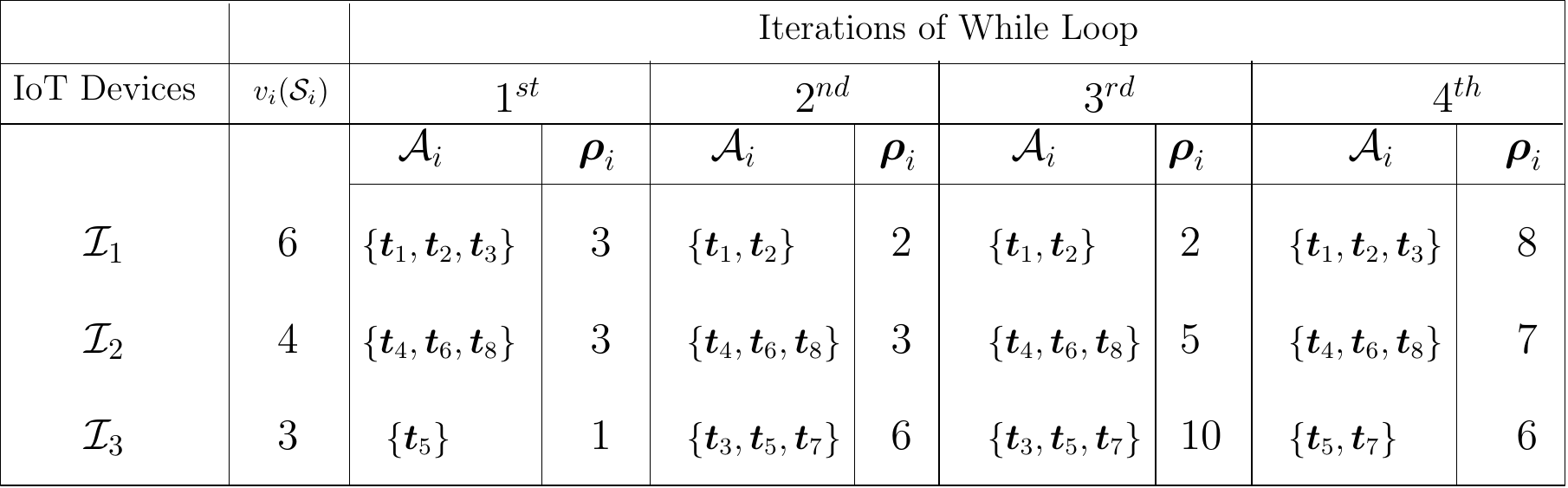}
\caption{Detailed Illustration of WiPD}
\label{fig:1}
\end{figure}
\noindent Now, at the given price vector $\boldsymbol{\rho}$, \emph{what are the additional task(s) you want}? The answer by IoT device $\mathcal{I}_1$ is: $\mathcal{F}_1 = \{\boldsymbol{t}_7\}$. So, $\mathcal{S}_1 = \{\boldsymbol{t}_1, \boldsymbol{t}_2, \boldsymbol{t}_3, \boldsymbol{t}_7\}$. Following line 19 of Algorithm \ref{algo:14}, the price of task $\boldsymbol{t}_7$ is increased by 1 and it became $\boldsymbol{\rho}(1) = 1$. Further, in the next iteration the IoT device $\mathcal{I}_2$ is asked about the additional task(s). $\mathcal{I}_2$ requested for additional task $\boldsymbol{t}_3$. So, $v_2(\mathcal{S}_2) = 4$, where $\mathcal{S}_2 = \{\boldsymbol{t}_3, \boldsymbol{t}_4, \boldsymbol{t}_6, \boldsymbol{t}_8\}$. As $\mathcal{F}_2 = \{\boldsymbol{t}_3\}$, so the price of task $\boldsymbol{t}_3$ will be increased by 1 and it became 2. In the next iteration of \emph{for} loop, IoT device $\mathcal{I}_3$ is asked about the additional task(s) and reported $\boldsymbol{t}_3$ and $\boldsymbol{t}_7$. So, $\mathcal{F}_3 = \{\boldsymbol{t}_3, \boldsymbol{t}_7\}$. Hence, $\mathcal{S}_3$ = $\{\boldsymbol{t}_3, \boldsymbol{t}_5, \boldsymbol{t}_7\}$. As $\mathcal{F}_3 = \{\boldsymbol{t}_3, \boldsymbol{t}_7\}$, so the prices of tasks $\boldsymbol{t}_3$ and $\boldsymbol{t}_7$ will be increased by 1 and it became 3 and 2 respectively. After the second iteration of the while loop, the allocation vector is  $\mathcal{A} = \{(\{\boldsymbol{t}_1, \boldsymbol{t}_2\}, \mathcal{I}_1), (\{\boldsymbol{t}_4, \boldsymbol{t}_6, \boldsymbol{t}_8\}, \mathcal{I}_2), (\{\boldsymbol{t}_3, \boldsymbol{t}_5, \boldsymbol{t}_7\}, \mathcal{I}_3)\}$ and payment vector $\boldsymbol{\rho} = \{2, 3, 6\}$.\\
\indent Now, at the given price vector $\boldsymbol{\rho}$, \emph{what are the additional tasks you want}? The answer by $\mathcal{I}_1$ IoT device is: $\mathcal{F}_1 = \{\boldsymbol{t}_3, \boldsymbol{t}_4, \boldsymbol{t}_6\}$. So, $\mathcal{S}_1 = \{\boldsymbol{t}_1, \boldsymbol{t}_2, \boldsymbol{t}_3, \boldsymbol{t}_4, \boldsymbol{t}_6\}$. Following line 19 of Algorithm \ref{algo:14}, the prices of tasks $\boldsymbol{t}_3$, $\boldsymbol{t}_4$, and $\boldsymbol{t}_6$ is increased by 1 and it became 4, 2, and 2 respectively. Further, in the next iteration the IoT device $\mathcal{I}_2$ is asked about the additional task(s). $\mathcal{I}_2$ requested for additional task $\boldsymbol{t}_7$. So, $v_2(\mathcal{S}_2) = 4$, where $\mathcal{S}_2 = \{\boldsymbol{t}_4, \boldsymbol{t}_6, \boldsymbol{t}_7, \boldsymbol{t}_8\}$. As $\mathcal{F}_2 = \{\boldsymbol{t}_7\}$, so the price of task $\boldsymbol{t}_7$ will be increased by 1 and it became 3. In the next iteration of \emph{for} loop, IoT device $\mathcal{I}_3$ is asked about the additional task(s) and reported $\boldsymbol{t}_3$ and $\boldsymbol{t}_7$. So, $\mathcal{F}_3 = \{\boldsymbol{t}_3, \boldsymbol{t}_7\}$. Hence, $\mathcal{S}_3$ = $\{\boldsymbol{t}_3, \boldsymbol{t}_5, \boldsymbol{t}_7\}$. As $\mathcal{F}_3 = \{\boldsymbol{t}_3, \boldsymbol{t}_7\}$, so the prices of tasks $\boldsymbol{t}_3$ and $\boldsymbol{t}_7$ will be increased by 1 and it became 5 and 4 respectively. After the third iteration of the while loop, the allocation vector is  $\mathcal{A} = \{(\{\boldsymbol{t}_1, \boldsymbol{t}_2\}, \mathcal{I}_1), (\{\boldsymbol{t}_4, \boldsymbol{t}_6, \boldsymbol{t}_8\}, \mathcal{I}_2), (\{\boldsymbol{t}_3, \boldsymbol{t}_5, \boldsymbol{t}_7\}, \mathcal{I}_3)\}$ and payment vector $\boldsymbol{\rho} = \{2, 5, 10\}$. \\
\indent Similarly, in the next iteration of while loop $\mathcal{F}_1 = \{\boldsymbol{t}_3, \boldsymbol{t}_7, \boldsymbol{t}_8\}$, $\mathcal{F}_2 = \{\boldsymbol{t}_8 \}$, and $\mathcal{F}_3 = \{\boldsymbol{t}_7\}$. After this iteration, the allocation and price vectors are $\mathcal{A} = \{(\{\boldsymbol{t}_1, \boldsymbol{t}_2, \boldsymbol{t}_3\}, \mathcal{I}_1), (\{\boldsymbol{t}_4, \boldsymbol{t}_6, \boldsymbol{t}_8\}, \mathcal{I}_2), (\{\boldsymbol{t}_5, \boldsymbol{t}_7\}, \mathcal{I}_3)\}$ and payment vector $\boldsymbol{\rho} = \{8, 7, 6\}$. In next iteration $\mathcal{F}_1$, $\mathcal{F}_2$, and $\mathcal{F}_3$ are $\phi$. So, the final allocation vector is $\mathcal{A} = \{(\{\boldsymbol{t}_1, \boldsymbol{t}_2, \boldsymbol{t}_3\}, \mathcal{I}_1), (\{\boldsymbol{t}_4, \boldsymbol{t}_6, \boldsymbol{t}_8\}, \mathcal{I}_2), (\{\boldsymbol{t}_5, \boldsymbol{t}_7\}, \mathcal{I}_3)\}$ and final payment vector is $\boldsymbol{\rho} = \{8, 7, 6\}$. The utility of $\mathcal{I}_1$, $\mathcal{I}_2$, and $\mathcal{I}_3$ is $u_1(\mathcal{S}_1, \boldsymbol{\rho}) = 8 - 6 = 2$, $u_2(\mathcal{S}_2, \boldsymbol{\rho}) = 7 - 4 = 3$, and $u_3(\mathcal{S}_3, \boldsymbol{\rho}) = 6 - 3 = 3$ respectively.    

\section{\textsc{Mechanism Analysis}}
\label{section:PM}
In this section, the analysis of the proposed mechanisms for the two tiers is discussed independently one by one. Firstly, in subsection \ref{subsec:AFT} the analysis of TENM is carried out. Next, the analysis of the proposed mechanisms for the second tier $i.e.$ ECTAI and WiPD are discussed in subsection \ref{subsec:AST}. 

\subsection{Analysis of First Tier}
\label{subsec:AFT}
In this subsection, the analysis of the first tier is depicted. In Lemma \ref{lemma:p1} it is proved that TENM runs in \emph{polynomial} time. The correctness of TENM has been discussed in Lemma \ref{lemma:p2}. It shows that on termination, TENM gives the desired output. Using Proposition \ref{propp:1}, Corollary \ref{lemma:p3} proves that IoT devices cannot improve their utility by misreporting their \emph{private} information. In other words, it shows that TENM is \emph{truthful} or \emph{incentive compatible}. By taking the help of Proposition \ref{prop:p1} in Corollary \ref{corr:p2} it is shown that TENM is \emph{budget feasible}.\\
\indent In Theorem \ref{the1} we are proving that the number of IoT devices that got notified about the task execution process by any $i^{th}$ IoT device in expectation is given as $E[X_i] = |\boldsymbol{\mathcal{Z}}_i| \cdot p$. Here, $X_i$ is an indicator random variable that captures the number of IoT devices that got notified about the task execution process by any $i^{th}$ IoT device, $|\boldsymbol{\mathcal{Z}}_i|$ is the number of IoT devices that are socially connected to $i^{th}$ IoT device and $p$ is the probability with which the $i^{th}$ IoT device will notify to any of its peers. This theorem will give us an estimate that the number of IoT devices notified by $i^{th}$ IoT device. Further, it is estimated in Theorem \ref{theorem:222} that the probability that any $i^{th}$ IoT device notifies about the task execution process to at least one IoT device out of $|\boldsymbol{\mathcal{Z}}_i|$ is given as $1- \bigg(\frac{1}{exp ({|\boldsymbol{\mathcal{Z}}_i| \cdot  \lceil ln |\boldsymbol{\mathcal{Z}}_i|\rceil})} \bigg)$. This theorem helps us to show with what probability any $i^{th}$ IoT device will be notifying at least one IoT device among the available ones. In a given social graph, the probability that at least $\sqrt{|\boldsymbol{\mathcal{Z}}_i|} \cdot \ln |\boldsymbol{\mathcal{Z}}_i|$ IoT devices got notified by any $i^{th}$ IoT device is given as $\frac{exp~\bigg( \sqrt{|\boldsymbol{\mathcal{Z}}_i|}\cdot \ln |\boldsymbol{\mathcal{Z}}_i| \bigg)}{\bigg( \sqrt{|\boldsymbol{\mathcal{Z}}_i|}\cdot \ln |\boldsymbol{\mathcal{Z}}_i|\bigg)^{ \sqrt{|\boldsymbol{\mathcal{Z}}_i|}\cdot \ln |\boldsymbol{\mathcal{Z}}_i|}}$. This theorem helps us to show that with what probability any $i^{th}$ IoT device will be notified to at least $\sqrt{|\boldsymbol{\mathcal{Z}}_i|} \cdot \ln |\boldsymbol{\mathcal{Z}}_i|$ IoT devices.              
\begin{lemma}
\label{lemma:p1}
TENM is computationally efficient.
\end{lemma}
\begin{proof}
The running time of TENM is the sum of the running time of Algorithm \ref{algo:1} and Algorithm \ref{algo:2}. So, let us determine the running time of each of the algorithms one by one.\\
\indent Lines 1$-$3 in Algorithm \ref{algo:1} will execute for $n$ times and is bounded above by $O(n)$. Line 4 will take $O(1)$. Line 5 is determining the IoT device with maximum $\bigg(\frac{h_{\mathcal{I}_k|\mathcal{S}_{k-1}}}{c_k}\bigg)$ and holding it in $i$. It will take $O(n)$ time. In the worst case, lines 6-10 may execute for $n$ times ($i.e.$ number of IoT devices). For each iteration of \emph{while} loop, lines 7 and 8 will take constant time. Line 9 will take $O(n)$. So \emph{while} loop is bounded above by $O(n^2)$. Line 11 will take $O(1)$. So, the running time of Algorithm \ref{algo:1} is $O(n) + O(1) + O(n^2) = O(n^2)$.\\
\indent Line 1 of Algorithm \ref{algo:2} takes $O(1)$ time. Lines 2-19 may run for $n$ times in the worst case. For each iteration of lines 2-19, lines 3 and 4 of Algorithm \ref{algo:2} will take $O(n)$ time. Lines 5-9 will take $O(n^2)$ as discussed in Algorithm \ref{algo:1}. Lines 10-16 will iterate for $n$-1 times in the worst case. For each iteration of \emph{for} loop in lines 10-16, lines 11-14 will take constant time $i.e.$ $O(1)$. Line 15 is bounded above by $O(n)$. So, lines 10-16 are bounded above by $O(n^2)$. Line 17 will take $O(n)$. Line 18 is bounded above by $O(1)$. So, each iteration of \emph{for} loop in lines 2-19 is bounded above by $O(1) + O(n) + O(n^2) + O(n^2) + O(n) = O(n^2)$. For $n$ iterations it will be $O(n^3)$. So, overall running time of Algorithm \ref{algo:2} is $O(1) + O(n^3) + O(1) = O(n^3)$.\\
  \indent Combining the running time of Algorithm \ref{algo:1} and Algorithm \ref{algo:2} $i.e.$ $O(n^2) + O(n^3) = O(n^3)$. TENM takes $O(n^3)$ time and hence is \emph{computationally efficient}.   
\end{proof}

\begin{lemma}
\label{lemma:p2}
TENM works correctly.
\end{lemma}
\begin{proof}
The proof of correctness of TENM is done using \emph{loop invariant} technique \cite{Coreman_2009}. To show that TENM is correct, it is to be shown that all the subroutines associated with TENM are correct. Let us prove that each of the subroutines of TENM is correct.\\

\indent \textbf{Proof of Correctness of Algorithm \ref{algo:1}:} To prove that Algorithm \ref{algo:1} is correct, the following loop invariant is considered:
    \begin{mdframed}[backgroundcolor=gray230]
\textbf{Loop invariant:} In each iteration of \emph{while} loop of lines 6-10, an IoT device  is added into the output array $\mathcal{S}$ as a notifier.
\end{mdframed}
\begin{itemize}
\item \textbf{Initialization:} We can start by showing that loop invariant holds before the first iteration of \emph{while} loop $i.e.$ when $\mathcal{S} = \phi$. The output set $\mathcal{S}$ has no IoT devices acting as notifiers before the first iteration. So, the loop invariant holds.
\item \textbf{Maintenance:} For the loop invariant to be true, it is to be shown that before any $l^{th}$ iteration of \emph{while} loop and after $l^{th}$ iteration of the \emph{while} loop the loop invariant holds. Before $l^{th}$ iteration, $i.e.$ till $(l-1)^{th}$ iteration there will be $(l-1)$  IoT devices in a set $\mathcal{S}$. After $l^{th}$ iteration, the number of IoT devices will be $\sum\limits_{i=1}^{l} 1 = l$ in $\mathcal{S}$. So, the loop invariant holds.
\item \textbf{Termination:} From the construction of Algorithm \ref{algo:1}, it is clear that the \emph{while} loop will terminate only when the condition in line 6 is not satisfied. It means that once \emph{while} loop terminates $\mathcal{S}$ contains the IoT devices that will act as initial notifiers.
\end{itemize}
As Algorithm \ref{algo:1} returns the desired output, it is correct.\\

\indent \textbf{Proof of Correctness of Algorithm \ref{algo:2}:} To prove that Algorithm \ref{algo:2} is correct, the following loop invariant is considered:
    \begin{mdframed}[backgroundcolor=gray230]
\textbf{Loop invariant:} In each iteration of \emph{for} loop of lines 2-19, a price of a single IoT device from among the IoT device in $\mathcal{S}$ is determined and is added in $\boldsymbol{\bar{\rho}}$.
\end{mdframed}
\begin{itemize}
\item \textbf{Initialization:} We can start by showing that the loop invariant holds before the first iteration of \emph{for} loop when $\boldsymbol{\bar{\rho}} = \phi$. The output set $\boldsymbol{\bar{\rho}}$ has no price value for notifiers before the first iteration. So, the loop invariant holds.
\item \textbf{Maintenance:} For the loop invariant to be true, it is to be shown that before any $l^{th}$ iteration of the \emph{for} loop and after $l^{th}$ iteration of the \emph{for} loop the loop invariant holds. Here, $l < |S|$. Before $l^{th}$ iteration, $i.e.$ till $(l-1)^{th}$ iteration there will be $(l-1)$ prices in a set $\boldsymbol{\bar{\rho}}$. After $l^{th}$ iteration, the price vector $\boldsymbol{\bar{\rho}}$ will contain the prices of $l$ IoT devices. So, the loop invariant holds.
\item \textbf{Termination:} From the construction of Algorithm \ref{algo:2}, it is clear that the \emph{for} loop will terminate only when the payment of all the IoT devices that acted as the initial notifiers are processed.  
\end{itemize}
As, Algorithm \ref{algo:2} returns the desired output, it is correct.\\

\indent As Algorithm \ref{algo:1} and Algorithm \ref{algo:2} are correct, so as the TENM.
\end{proof}

\begin{proposition}
\label{propp:1}
The mechanism discussed in \cite{singer2012win,Singer_2016, Singer:2012:WFI:2124295.2124381} is truthful.
\end{proposition}
\begin{corollary}
\label{lemma:p3}
TENM is IC.
\end{corollary}
\begin{proof}
From the construction of TENM, it can be seen that TENM consists of \emph{notifiers allocation mechanism} and \emph{notifiers pricing mechanism}. The \emph{notifiers allocation mechanism} determines the set of IoT devices that can act as the initial notifiers in the social graph. The \emph{notifiers pricing mechanism} is used to determine the payment of the IoT devices that are acting as the initial notifiers. The \emph{notifiers allocation mechanism} and \emph{notifiers pricing mechanism} of TENM are based on the allocation rule and payment characterization discussed in \cite{singer2012win} respectively. So, following Proposition \ref{propp:1} it can be inferred that TENM is IC.
\end{proof}

\begin{proposition}
\label{prop:p1}
The payment characterization (i.e. each winning IoT device as notifier will be paid how much) is discussed in \cite{Singer_2016, Singer:2012:WFI:2124295.2124381} is budget feasible.
\end{proposition}
\begin{corollary}
\label{corr:p2}
The total payment made to the notifiers using the notifiers pricing mechanism of TENM is within the available budget.
\begin{proof}
The notifiers pricing mechanism is utilized to determine the payment of the winning IoT devices (that are acting as the initial notifiers). The payment characterization of the notifiers pricing mechanism is based on the payment characterization of \cite{Singer_2016, Singer:2012:WFI:2124295.2124381}. By Proposition \ref{prop:p1}, we can infer that TENM is budget feasible. 
\end{proof}
\end{corollary}
\begin{theorem}
\label{the1}
In the given social graph $\mathcal{G}(\boldsymbol{\mathcal{N}}^{T},\boldsymbol{\mathcal{R}}^{T})$, the expected number of IoT devices notified by any $i^{th}$ IoT device for the task execution process is given as  $E[X_i]$ = $| \boldsymbol{\mathcal{Z}}_i| \cdot p$. Here, $X_i$ is the indicator random variable that keeps track of a number of IoT devices that got notified by any $i^{th}$ IoT device, $ \boldsymbol{\mathcal{Z}}_i$ is the set of IoT devices that are socially connected to $i^{th}$ IoT device and $p$ is the probability with which the $i^{th}$ IoT device will notify to its connections.
\end{theorem}
\begin{proof}
In this theorem, we wish to determine that in expectations \emph{how many IoT devices will get notified by any $i^{th}$ IoT device?} 
The sample space for the event is represented as $U$ and is given as : 
\begin{equation*}
 U = \{\underbrace{i^{th}~ IoT~ device~ notify~ to~ j^{th}~ 
IoT~ device}_Y,~ \underbrace{i^{th}~IoT~device~ do~ not~ 
 notify~ to~ j^{th}~ IoT~ device}_{\bar{Y}}\} 
\end{equation*}
 
 \noindent The probability that Y takes place is $p$ and the probability that Y does not takes place is $(1-p)$. Let $X_i$ be the random variable whose value will be equal to the number of IoT devices notified by any $i^{th}$ IoT device.
 We let $X_i^j$ be the indicator random variable associated with the event in which the $j^{th}$ IoT device is notified by $i^{th}$ IoT device. Thus, $X_i^j$ = $\mathcal{I}\{Y\}$
 \begin{equation*}
    X_i^j  = \begin{cases} 1   ,\quad \text{if Y happen.}  \\
     0, \quad \text{otherwise }
  \end{cases}
  \end{equation*}
 As it is known that the expected value of the indicator random variable capturing the event is equal to the probability of that event \cite{Coreman_2009}. So, we have 
  \begin{equation}
  \label{equ:t1}
   E[ X_i^j]  = pr\{Y\}
  \end{equation}
  The indicator random variable that we are interested in is given as: 
   \begin{equation*}
  X_i = \sum_{j=1}^{|\boldsymbol{\mathcal{Z}}_i|}X_i^j
 \end{equation*}
 Taking expectations from both sides, we get
  \begin{equation*}
   E[ X_i]  = E\Bigg[\sum_{j=1}^{|\boldsymbol{\mathcal{Z}}_i|}  X_i^j \bigg]
  \end{equation*}
  By linearity of expectation, we have
  \begin{equation}
  \label{equ:t2}
   E[ X_i]  = \sum_{j=1}^{|\boldsymbol{\mathcal{Z}}_i|}  E[ X_i^j]  \end{equation}
  Substituting the value of equation \ref{equ:t1} to equation \ref{equ:t2}, we have 
   \begin{equation*}
   E[ X_i]  = \sum_{j=1}^{|\boldsymbol{\mathcal{Z}}_i|} pr\{Y\} 
  \end{equation*}
\begin{equation*}
 \hspace*{1.5mm} = \sum_{j=1}^{|\boldsymbol{\mathcal{Z}}_i|} p
 \end{equation*}
 \begin{equation*}
\hspace*{5mm} = |\boldsymbol{\mathcal{Z}}_i| \cdot p
 \end{equation*}
 Hence proved.
\end{proof}
\begin{observation}
\label{th:obs}
If $p$ value is considered as $\frac{1}{2}$, then $E[X_i]$ will be $\frac{|\boldsymbol{Z}_i|}{2}$. It means that in expectation half of the socially connected IoT devices with $i^{th}$ IoT device will be notified by the $i^{th}$ IoT device. If we consider $p$ value as $\frac{1}{9}$ then $\frac{|\boldsymbol{Z}_i|}{9}$ IoT devices got notified about the task execution process by $i^{th}$ IoT device. It means that the higher the $p$ value higher will be $E[X_i]$ value.  
\end{observation}

\begin{theorem}
\label{theorem:222}
In TENM, the probability that any $i^{th}$ IoT device notifies about the task execution process to at least one IoT device is bounded above by $1- \bigg(\frac{1}{exp ({|\boldsymbol{\mathcal{Z}}_i| \cdot \lceil ln |\boldsymbol{\mathcal{Z}}_i|\rceil})} \bigg)$. Mathematically,  
\begin{equation*}
pr(X_i>1) \leq   1- \bigg(\frac{1}{exp ({|\boldsymbol{\mathcal{Z}}_i| \cdot  \lceil ln |\boldsymbol{\mathcal{Z}}_i|\rceil})} \bigg)
\end{equation*}
\begin{proof}
In this theorem, we are interested in determining \emph{what is the probability that any} $i^{th}$ \emph{IoT device will notify at least one of the IoT devices in its social connection}?  For this purpose, the proof and results presented in Theorem \ref{the1} will be utilized. Notifying $j^{th}$ IoT device by any $i^{th}$ IoT device is independent of notifying the other IoT devices in $\boldsymbol{\mathcal{Z}}_i$. The probability that the $i^{th}$ IoT device has not notified any of the IoT devices in $\boldsymbol{\mathcal{Z}}_i$ is:
  \begin{equation*}
  pr(X_i< 1) = (1-p) \times (1-p) \times \ldots  |\boldsymbol{\mathcal{Z}}_i|~ times
 \end{equation*}
 \begin{equation}
 \label{equ:123}
 \hspace*{-20mm} = (1-p)^{|\boldsymbol{\mathcal{Z}}_i|} 
 \end{equation}
Utilizing the standard inequality 1+$|\boldsymbol{\mathcal{Z}}_i|$ $\leq  exp~ ({|\boldsymbol{\mathcal{Z}}_i|})$, equation \ref{equ:123} can be written as: 
 
 \begin{equation*}
  pr(X_i< 1) \leq exp~\large({-|\boldsymbol{\mathcal{Z}_i|} \cdot p\large)}
 \end{equation*}
 
 \begin{equation}
 \label{equ:71}
  \hspace*{17mm} =\frac{1}{exp~\large({|\boldsymbol{\mathcal{Z}_i|} \cdot p\large)}}
 \end{equation}
 Given equation \ref{equ:71} the probability that any $i^{th}$ IoT device will notify at least one IoT device is given as:
 \begin{equation}
 \label{equ:5678}
  pr(X_i \geq 1) \leq 1-\bigg(\frac{1}{exp~\large({|\boldsymbol{\mathcal{Z}_i|} \cdot p\large)}}\bigg)
 \end{equation} 
 Now if we take $p =\lceil ln |\boldsymbol{\mathcal{Z}_i|} \rceil$ then equation \ref{equ:5678} will be 
  \begin{equation*}
 pr(X_i \geq 1) \leq  1-\bigg(\frac{1}{exp~\large({|\boldsymbol{\mathcal{Z}_i|}\cdot \lceil ln |\boldsymbol{\mathcal{Z}_i|\rceil}}\large)}\bigg)
 \end{equation*}
 Hence proved.
\end{proof}
\end{theorem}

\begin{corollary}
If the value of $|\boldsymbol{\mathcal{Z}}_i|$ for any $i^{th}$ IoT device is taken as say $5$, then  
 \begin{equation*}
 pr(X_i \geq 1) \leq  1-\bigg(\frac{1}{exp~\large({5}\cdot \lceil \ln 5\rceil \large)}\bigg)
 \end{equation*}
 \begin{equation*}
\hspace*{12mm} =  1-\frac{1}{3125}
 \end{equation*}

 \begin{equation*}
\hspace*{4.5mm} =  0.99
 \end{equation*}
 From the above calculation, it can be inferred that the probability that one of the IoT devices will be notified among $|\boldsymbol{Z}_i| = 5$ is very high.
\end{corollary}

\begin{proposition}[Chernoff Bounds \cite{T.roughgarden_201618}]
\label{prop:2}
Let $X_1, X_2, \ldots, X_n$ be random variables that have the common range [0,~1] and $X = \sum_{i=1}^{n} X_i$. Given the above set-up, for every $\kappa > 0$, we have  
\begin{equation}
\label{equ:chernoff}
Pr\bigg\{X > (1+\boldsymbol{\kappa})E[X] \bigg\} < \bigg(\frac{exp~((1+\boldsymbol{\kappa})E[X])}{(1+\boldsymbol{\kappa})^{(1+\boldsymbol{\kappa})E[X]}}\bigg) 
\end{equation}
\begin{proof}
As the proof is standard, it is omitted from this paper.
\end{proof}
\end{proposition}
\begin{lemma}
\label{lemma:1}
Given a social connection $\mathcal{G} (\boldsymbol{\mathcal{N}}^T, \boldsymbol{\mathcal{R}}^T)$, the probability that the number of IoT devices that are notified by any $i^{th}$ IoT device is greater than $\sqrt{|\boldsymbol{\mathcal{Z}}_i|} \cdot \ln |\boldsymbol{\mathcal{Z}}_i|$ is given as:
\begin{equation*}
Pr\bigg\{X_i > \sqrt{|\boldsymbol{\mathcal{Z}}_i|} \cdot \ln |\boldsymbol{\mathcal{Z}}_i|\bigg\} < \frac{exp~\bigg( \sqrt{|\boldsymbol{\mathcal{Z}}_i|}\cdot \ln |\boldsymbol{\mathcal{Z}}_i| \bigg)}{\bigg( \sqrt{|\boldsymbol{\mathcal{Z}}_i|}\cdot \ln |\boldsymbol{\mathcal{Z}}_i|\bigg)^{ \sqrt{|\boldsymbol{\mathcal{Z}}_i|}\cdot \ln |\boldsymbol{\mathcal{Z}}_i|}} 
\end{equation*}
\end{lemma}
where, $|\boldsymbol{\mathcal{Z}}_i|$ is the set of IoT devices that are socially connected to $i^{th}$ IoT device.
\begin{proof}
To prove the result of Lemma \ref{lemma:1}, the Chernoff bounds \cite{T.roughgarden_201618} (Proposition \ref{prop:2}) is utilized. Let us say $\boldsymbol{\kappa} = \sqrt{|\boldsymbol{\mathcal{Z}}_i|} \cdot \ln |\boldsymbol{\mathcal{Z}}_i| - 1$, and $E[X_i] = 1$. Here $|\boldsymbol{\mathcal{Z}}_i| >2$. Substituting the value of $\boldsymbol{\kappa}$ and $E[X_i]$ in equation \ref{equ:chernoff}, we get 
\begin{equation*}
Pr\bigg\{X_i > \bigg(1+\sqrt{|\boldsymbol{\mathcal{Z}}_i|}\cdot \ln |\boldsymbol{\mathcal{Z}}_i| - 1\bigg) \cdot 1 \bigg\} = Pr\bigg\{X_i > \sqrt{|\boldsymbol{\mathcal{Z}}_i|} \cdot \ln |\boldsymbol{\mathcal{Z}}_i|\bigg\} < \bigg(\frac{exp~((1+\boldsymbol{\kappa})E[X])}{(1+\boldsymbol{\kappa})^{(1+\boldsymbol{\kappa})E[X]}}\bigg) 
\end{equation*}

\begin{equation*}
= \frac{exp~\bigg(\bigg(1+\sqrt{|\boldsymbol{\mathcal{Z}}_i|}\cdot \ln |\boldsymbol{\mathcal{Z}}_i| - 1\bigg) \cdot 1\bigg)}{\bigg(1+\sqrt{|\boldsymbol{\mathcal{Z}}_i|}\cdot \ln |\boldsymbol{\mathcal{Z}}_i| - 1\bigg)^{\bigg(1+\sqrt{|\boldsymbol{\mathcal{Z}}_i|}\cdot \ln |\boldsymbol{\mathcal{Z}}_i| - 1\bigg) \cdot 1}} 
\end{equation*}

\begin{equation*}
= \frac{exp~\bigg( \sqrt{|\boldsymbol{\mathcal{Z}}_i|}\cdot \ln |\boldsymbol{\mathcal{Z}}_i| \bigg)}{\bigg( \sqrt{|\boldsymbol{\mathcal{Z}}_i|}\cdot \ln |\boldsymbol{\mathcal{Z}}_i|\bigg)^{ \sqrt{|\boldsymbol{\mathcal{Z}}_i|}\cdot \ln |\boldsymbol{\mathcal{Z}}_i|}} 
\end{equation*}
Hence proved.
\end{proof}

\begin{observation}
If we have $\sqrt{|\boldsymbol{\mathcal{Z}}_i|}\cdot \ln |\boldsymbol{\mathcal{Z}}_i| = 3$ then the probability that any $i^{th}$ IoT device is notifying to at least 3 IoT devices in its social connection is bounded above by is given as:
\begin{equation*}
Pr\{X_i > 3\} < \frac{exp~\bigg( \sqrt{|\boldsymbol{\mathcal{Z}}_i|}\cdot \ln |\boldsymbol{\mathcal{Z}}_i| \bigg)}{\bigg( \sqrt{|\boldsymbol{\mathcal{Z}}_i|}\cdot \ln |\boldsymbol{\mathcal{Z}}_i|\bigg)^{ \sqrt{|\boldsymbol{\mathcal{Z}}_i|}\cdot \ln |\boldsymbol{\mathcal{Z}}_i|}}     
\end{equation*}

\begin{equation*}
= \frac{exp~(3)}{(3)^3}     
\end{equation*}

\begin{equation*}
= 0.7438     
\end{equation*}
Here, it can be seen that with probability at most 0.7438 the $i^{th}$ IoT device will notify about the task execution process to at least 3 IoT devices in its social connection.     
\end{observation}
\subsection{Analysis of Second Tier}
\label{subsec:AST}
In this, the analysis of the second tier is depicted. In Lemma \ref{lemma:sp1}, it is proved that ECTAI and WiPD run in \emph{polynomial time}. The correctness of ECTAI and WiPD are discussed in Lemma \ref{lemma:sp2} and Lemma \ref{lemma:sp3} respectively. It is shown that on termination, ECTAI and WiPD give the correct output. In Lemma \ref{lemma:8} it is shown any $i^{th}$ IoT device cannot gain by misreporting its peak value $i.e.$ ECTAI is \emph{truthful}. By misreporting the private valuation for the set of tasks, the IoT devices cannot gain $i.e.$ WiPD is \emph{truthful} is shown in Lemma \ref{lemma:ttu}. 
\begin{lemma}
\label{lemma:sp1}
ECTAI and WiPD are computationally efficient.
\end{lemma}
\begin{proof}
The running time of ECTAI and WiPD is the running time of  Algorithm \ref{algo:3} and Algorithm \ref{algo:14} respectively. So, let us determine the running time of each of the algorithms one by one.\\
\indent In line 1 of Algorithm \ref{algo:3}, the initialization is done and will take $O(1)$ time. For each iteration of \emph{while} loop in lines 2-18, line 3 will take $O(n)$ time. Line 4 is bounded above by $O(n)$. Lines 5-8 are bounded above by $O(g)$. Line 9 sorts the peak values of the IoT devices and takes $O(g \lg g)$, where $g$ is the number of IoT devices present in $\mathcal{R}$. Lines 10-14 will take $O(1)$ time. Line 15 will take $O(f)$ time. Line 16  will take constant time. For removing $\eta$ IoT devices from $\mathcal{I}'$, it will take $O(n)$ time. So, the time taken by lines 2-18 for each iteration of \emph{while} loop is given as $O(n) + O(g) + O(g \lg g) + O(f) + O(1) + O(n) = O(g \lg g)$, if $g$ is a function of $n$ then it can be written as $O(n \lg n)$. As the \emph{while} loop will iterate for $n$ times, so lines 2-18 are bounded above by $O(n^2 \lg n)$. So, the running time of Algorithm \ref{algo:3} is $O(1) + O(n^2 \lg n) + O(1) = O(n^2 \lg n)$.\\    
\indent In Algorithm \ref{algo:14}, line 1 will take constant time. Lines 2-4 will iterate for $m$ times and are bounded above by $O(m)$. Lines 5-7 will iterate for $\mathcal{N}$ times and are bounded above by $O(\mathcal{N})$, where $\mathcal{N}$ is the number of quality IoT devices in set $\boldsymbol{\mathcal{O}}$. Let us say the \emph{while} loop in lines 8-25 iterates for $\lambda$ times. For each iteration of \emph{while} loop, the \emph{for} loop in lines 9-24 will iterate for $\mathcal{N}$ times. Line 10 is bounded above by $\mathcal{N}$. Lines 11-17 will take $O(1)$ time. Lines 18 and 19 are bounded above by $O(n^2)$. Lines 20-22 will take constant time. So, lines 9-24 are bounded above by $O(m) + O(\mathcal{N}) + O(n^2) = O(n^2)$ for each iteration of \emph{while} loop. So, \emph{while} loop in lines 8-25 takes $O(xn^2)$. If $x$ is a fraction of $n$ then it is rewritten as $O(n^3)$. The time taken by Algorithm \ref{algo:14} is $O(1) + O(m) + O(n) + O(n^3) + O(1) = O(n^3)$.\\   \indent From above it can be seen that Algorithm \ref{algo:3} and Algorithm \ref{algo:14} are bounded above by $O(n^2 \lg n)$ and $O(n^3)$. Hence, ECTAI and WiPD are \emph{computationally efficient}.   
\end{proof}

\begin{lemma}
\label{lemma:sp2}
ECTAI works correctly.
\end{lemma}
\begin{proof}
The proof of correctness of ECTAI is done using loop invariant technique \cite{Coreman_2009}. \\

\indent \textbf{Proof of Correctness of Algorithm \ref{algo:3}:} To prove that Algorithm \ref{algo:3} is correct, the following loop invariant is considered:
    \begin{mdframed}[backgroundcolor=gray230]
\textbf{Loop invariant:} In each iteration of \emph{while} loop of lines 2-18, a set of quality IoT devices is determined and is added in $\boldsymbol{\mathcal{O}}$.
\end{mdframed}
\begin{itemize}
\item \textbf{Initialization:} We can start by showing that loop invariant hold before the first iteration of \emph{while} loop when $\boldsymbol{\mathcal{O}} = \phi$. The output set $\boldsymbol{\mathcal{O}}$ has no quality IoT devices before the first iteration. So, the loop invariant holds.
\item \textbf{Maintenance:} For the loop invariant to be true, it is to be shown that before any $l^{th}$ iteration of the \emph{while} loop and after $l^{th}$ iteration of the \emph{while} loop the loop invariant holds. Before $l^{th}$ iteration, $i.e.$ till $(l-1)^{th}$ iteration there will be $(l-1)$ quality IoT devices in a set $\boldsymbol{\mathcal{O}}$. After $l^{th}$ iteration, the set $\boldsymbol{\mathcal{O}}$ will contain $l$ IoT devices. So, the loop invariant holds.
\item \textbf{Termination:} From the construction of Algorithm \ref{algo:3}, it is clear that the \emph{while} loop will terminate only when all the IoT devices are processed for determining the set of quality IoT devices.
\end{itemize}
Hence, Algorithm \ref{algo:3} is correct and so is the ECTAI.
\end{proof}

\begin{lemma}
\label{lemma:sp3}
WiPD works correctly.
\end{lemma}
\begin{proof}
The proof of correctness of WiPD is done using loop invariant technique \cite{Coreman_2009}. \\

\indent \textbf{Proof of Correctness of Algorithm \ref{algo:14}:} To prove that Algorithm \ref{algo:14} is correct, the following loop invariant is considered:
    \begin{mdframed}[backgroundcolor=gray230]
\textbf{Loop invariant:} In each iteration of \emph{while} loop of lines 8-25, a set of IoT devices  as winners for each of the tasks are determined and the winner's payment is decided.
\end{mdframed}
\begin{itemize}
\item \textbf{Initialization:} We can start by showing that loop invariant holds before the first iteration of \emph{while} loop, the sets  $\mathcal{A} = \phi$ and $\boldsymbol{\rho} = \phi$. The output sets $\mathcal{A} = \phi$ and $\boldsymbol{\rho} = \phi$ have no IoT devices as winners and no payment. So, the loop invariant holds.
\item \textbf{Maintenance:} For the loop invariant to be true, it is to be shown that before any $l^{th}$ iteration of the \emph{while} loop and after $l^{th}$ iteration of the \emph{while} loop the loop invariant holds. Before $l^{th}$ iteration, say, some $k$ sets of tasks of $k$ IoT devices are held in $\mathcal{A}$, and $\boldsymbol{\rho}$ holds the payment of $k$ IoT devices. After $l^{th}$ iteration, the set $\mathcal{A}$ may contain some $(k+x)$ sets of tasks of $(k+x)$ IoT devices, and the set $\boldsymbol{\rho}$ contains the payment of $k+x$ IoT devices. So, the loop invariant holds.
\item \textbf{Termination:} From the construction of Algorithm \ref{algo:14}, it is clear that the \emph{while} loop will terminate only when the requested set of additional tasks $i.e.$ $\mathcal{F}_i$ for all $\mathcal{I}_i \in \boldsymbol{\mathcal{O}}$ is $\phi$.
\end{itemize}
Hence, Algorithm \ref{algo:14} is correct, so as the WiPD\\
\end{proof}

\begin{proposition}
The median voting rule $\cite{T.roughgarden_20163}$ is truthful.
\end{proposition}
\begin{lemma}
\label{lemma:8}
ECTAI is truthful.
\end{lemma}
\begin{proof}
To prove that ECTAI is \emph{truthful}, it is sufficient to prove that the participating IoT devices are not gaining by misreporting their privately held peak value. Consider any $i^{th}$ IoT device. Let us say that the reported peak value of $i^{th}$ IoT device lies to the left of the median $\mathcal{\hat{R}}$ on the scale of $[0, 1]$. $\mathcal{\hat{R}}$ is obtained between 0 and 1 on the scale of $[0, 1]$ when all the participating IoT devices are reporting their peak value \emph{truthfully}. Any $i^{th}$ IoT device can misreport its peak value in two different ways (considered as two cases in the proof): (1) by reporting a lower peak value from its true peak value, and (2) by reporting a higher peak value from its true peak value. Let us consider the two cases one by one.
\begin{itemize}
    \item \textbf{Case 1 ($\alpha_i' < \alpha_i$):} In this case, it is considered that any $i^{th}$ IoT device reports a lower peak value from its true peak value. In the case when true peak value was reported, the utility of $i^{th}$ IoT device is given as $u_i = |\alpha_i - \hat{\mathcal{R}}|$. The pictorial representation of the true peak value scenario is depicted in Figure \ref{fig:ex111}. Now, when $i^{th}$ IoT device lowers its peak value $i.e.~ \alpha_i' < \alpha_i$, in such case the resultant peak value will be the same as the resultant peak value when $i^{th}$ IoT device reports its peak value in a truthful manner $i.e.$ $\hat{R}' = \hat{R}$. If that is the case, then the utility of $i^{th}$ IoT device will be same as the utility of $i^{th}$ IoT device when reported truthfully $i.e.$ $u_i' = |\alpha_i - \hat{\mathcal{R}}'| = u_i$.      
    \begin{figure}[H]
     \centering
     \begin{subfigure}[h]{0.42\textwidth}
         \centering
         \includegraphics[width=\textwidth]{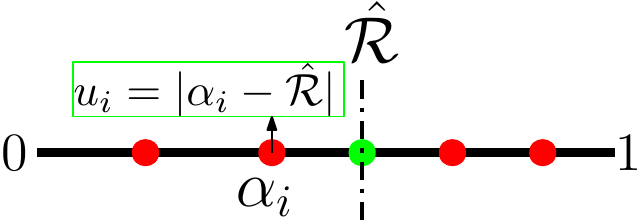}
         \caption{$i^{th}$ IoT device reporting true peak value}
         \label{fig:ex111}
     \end{subfigure}
     \hfill
     \begin{subfigure}[h]{0.42\textwidth}
         \centering
         \includegraphics[width=\textwidth]{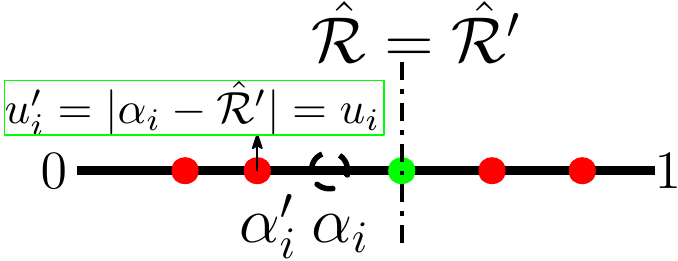}
         \caption{$i^{th}$ IoT device misreporting peak value ($\alpha_i' < \alpha_i$)}
         \label{fig:ex1112}
     \end{subfigure}
     \caption{Pictorial representation of the case when $i^{th}$ IoT device reports true peak value and lower peak value from the true peak value}
         \label{fig:ex1113}
     \end{figure}
    \item \textbf{Case 2 ($\alpha_i' > \alpha_i$):} In this case, it is considered that any $i^{th}$ IoT device reports a higher peak value from its true peak value. In the case when true peak value was reported, the utility of $i^{th}$ IoT device is given as $u_i = |\alpha_i - \hat{\mathcal{R}}|$. The pictorial representation of the true peak value scenario is depicted in Figure \ref{fig:ex111}. Now, when $i^{th}$ IoT device increases its peak value $i.e.~ \alpha_i' > \alpha_i$ again the two cases can occur. In the first case, it may happen that the $i^{th}$ IoT device increases its peak value and the increased peak value still lies on to the left of the resultant peak value obtained when reported \emph{truthfully} as depicted in Figure \ref{fig:ex111b}. In such case, the resultant peak value will be the same as the resultant peak value when the true peak value was reported by $i^{th}$ IoT device $i.e.$ $\hat{R}' = \hat{R}$. If that is the case, then the utility of $i^{th}$ IoT device will be same as the utility of $i^{th}$ IoT device when reported \emph{truthfully} $i.e.$ $u_i' = |\alpha_i - \hat{\mathcal{R}}'| = u_i$. Another case could be that when $i^{th}$ IoT device reported an increased peak value $i.e.~ \alpha_i' > \alpha_i$ and it crosses the resultant peak value obtained when all the IoT devices were reporting truthfully as depicted in Figure \ref{fig:ex1112b}. In such case the resultant peak value will be shifted to the right of the current resultant peak value as shown in Figure \ref{fig:ex1112b} and the utility of $i^{th}$ IoT device will be $u_i' = |\alpha_i - \hat{\mathcal{R}}'| > u_i$. It means that the resultant peak value moved away from the true peak value. Hence, it's a loss. So, with the increase in peak value, the utility of $i^{th}$ IoT device is remaining the same or is becoming worse.         
    \begin{figure}[H]
     \centering
     \begin{subfigure}[h]{0.42\textwidth}
         \centering
         \includegraphics[width=\textwidth]{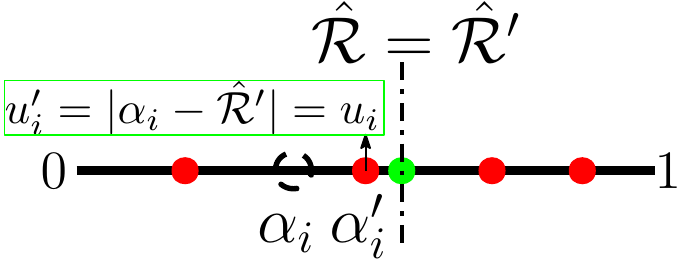}
         \caption{$i^{th}$ IoT device misreporting peak value ($\alpha_i' > \alpha_i$)}
         \label{fig:ex111b}
     \end{subfigure}
     \hfill
     \begin{subfigure}[h]{0.45\textwidth}
         \centering
         \includegraphics[width=\textwidth]{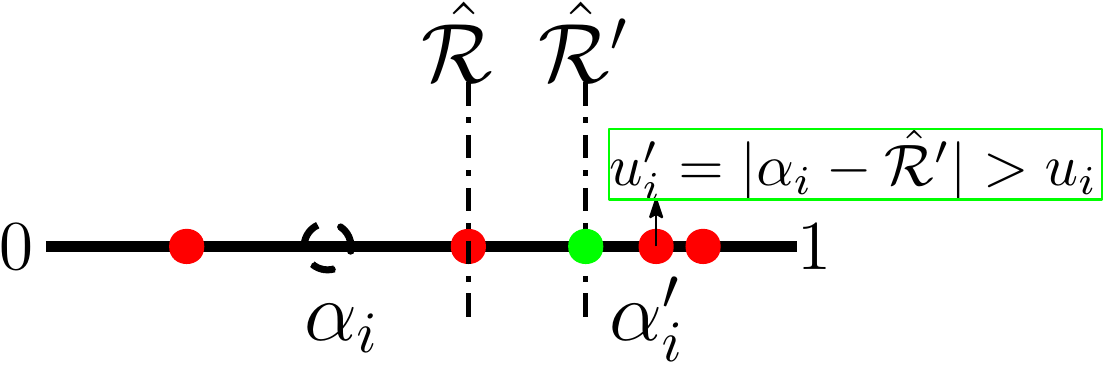}
         \caption{$i^{th}$ IoT device misreporting peak value ($\alpha_i' > \alpha_i$)}
         \label{fig:ex1112b}
     \end{subfigure}
     \caption{Pictorial representation of the case when $i^{th}$ IoT device reports a higher peak value from the true peak value.}
         \label{fig:ex1113b}
     \end{figure}
    \end{itemize}
    From the above-discussed two cases, it can be seen that the participating IoT devices are not gaining by misreporting their peak value. Hence, ECTAI is \emph{truthful}.
    \end{proof}
    
    \begin{lemma}
    \label{lemma:ttu}
    WiPD is truthful.
    \end{lemma}
    \begin{proof}
    Fix an IoT device $\mathcal{I}_i$. To prove that WiPD is \emph{truthful}, the two cases are considered: (1) \textbf{\emph{underbid}}, and (2) \textbf{\emph{overbid}}. In the first case, the $i^{th}$ IoT device decreases its bid value for a set of tasks $\mathcal{S}$ such that $v_i'(\mathcal{S}) < v_i(\mathcal{S})$. In \emph{overbid} case, the $i^{th}$ IoT device increases its bid value for a set of tasks $\mathcal{S}$ such that $v_i'(\mathcal{S}) > v_i(\mathcal{S})$. Let us illustrate the two cases.
    \begin{itemize}
        \item \textbf{Underbid:} Let us suppose that $i^{th}$ IoT device misreported his valuation for the set of tasks say $\mathcal{S}$ such that $v_i'(\mathcal{S}) < v_i(\mathcal{S})$. In such case two things can happen: (1) it may happen to be the last iteration and for any $i^{th}$ IoT device $\mathcal{F}_i = \phi$. It means that none of the IoT devices wants to have an additional task. In this the $i^{th}$ IoT device is allocated its requested tasks at the price $\sum\limits_{j \in \mathcal{S}} \boldsymbol{\rho}_{i}'(j)$ and hence the utility is $\sum\limits_{j \in \mathcal{S}} \boldsymbol{\rho}_{i}'(j) - v_i(S) < \sum\limits_{j \in \mathcal{S}} \boldsymbol{\rho}_{i}(j) - v_i(S) = u_i(\mathcal{S}, \boldsymbol{\rho})$. (2) Another case could be that due to the lowering of bid value, the $i^{th}$ IoT device has not received any tasks and hence the loser. In that case, the utility will be 0.  
        \item \textbf{Overbid:} Let us suppose that $i^{th}$ IoT device misreported his valuation for the set of tasks $\mathcal{S}$ such that $v_i'(S) > v_i(S)$. In this case, the two scenarios can occur: (1) if the current iteration is the last, and after that mechanism terminates. In this case, the $i^{th}$ IoT device is the winner and its utility will be $ u_i'(\mathcal{S}, \boldsymbol{\rho}) = \sum\limits_{j \in \mathcal{S}} \boldsymbol{\rho}_{i}(j) - v_i(\mathcal{S}) = u_i(\mathcal{S}, \boldsymbol{\rho})$. (2) Another case could be with the increase in bid values the IoT device fetched some additional sets of tasks but its tasks got fetched by other IoT devices, so its utility will be 0.    
    \end{itemize}
    From the above discussion, it can be inferred that IoT devices cannot gain by misreporting their private value. Hence, WiPD is truthful. 
    \end{proof}

\section{Experimental Analysis}
\label{sec:Sim}
In this section, the simulation for the two tiers is carried out independently. For the first tier of the proposed framework, through simulation, TENM is compared with the already existing mechanisms, namely, \emph{\textbf{\underline{n}}on-\textbf{\underline{t}}ruthful \textbf{\underline{b}}udget \textbf{\underline{f}}easible \textbf{\underline{m}}echanism} (NTBFM) and \emph{\textbf{\underline{p}}roportional \textbf{\underline{s}}hare \textbf{\underline{m}}echnaism} (PSM) \cite{Singer_2016}. The comparison is done based on (1) truthfulness, (2) budget feasibility, (3) the number of IoT devices selected as initial notifiers in the social network, and (4) running time. The evaluation metric mentioned in point 1 will help us to show that TENM and PSM are truthful and NTBFM is vulnerable to manipulation. The budget feasibility metric mentioned in point 2 will help us to show that NTBFM, PSM, and TENM are budget feasible. The estimate on the number of IoT devices getting selected as initial notifiers in the case of TENM, PSM, and NTBFM is taken care of by the metric mentioned in point 3. The evaluation metric mentioned in point 4 will help us to show that TENM is scalable.\\
\indent In the second tier of the proposed framework, the experiments are carried out to compare ECTAI and WiPD with the already existing mechanism \emph{\textbf{\underline{a}}verage \textbf{\underline{v}}oting \textbf{\underline{r}}ule} (AVR) and greedy mechanism respectively on the ground of \emph{truthfulness}. It will help us to show that both ECTAI and WiPD are truthful as compared to the respective benchmark mechanism. The second metric that is considered for comparing ECTAI and WiPD with the respective benchmark mechanism is the running time of the mechanisms. It will help us to show that ECTAI and WiPD are scalable.   
\subsection{Simulation Setup}
In this section, the simulation setup for the two independent tiers is discussed one by one. The simulation is carried out by using the software and programming language mentioned in Table \ref{tab:data}.
\begin{table}[H]
\renewcommand{\arraystretch}{0.99}  
\caption{Specifications of the software used in simulation}
\label{tab:data}
\centering
  \begin{tabular}{|c||c|}
\hline
\textbf{Software} & \textbf{Version}\\
\hline
Ubuntu & 20.04.4 LTS \\
\hline
Python & 3.8.10\\
\hline
Processor & Intel Core™ i5-8250U~ CPU~ @ 1.60GHz × 8\\ 
\hline
Gnuplot & 5.2 \\
\hline
\end{tabular}
\end{table}
\begin{enumerate}
\item \textbf{Simulation setup for first tier -} For the simulation purpose we use three real-world data sets by considering the social connections of the users in Facebook \cite{fb}, Twitter \cite{twit}, and Google+ \cite{gplus}, with statistics given in Table \ref{tab:second table}. It shows the \emph{number of nodes}, \emph{number of edges}, and the \emph{cost} range utilized for three different social networks. Each of the three mechanisms is executed for 5 rounds (or iterations) on the data given in Table \ref{tab:second table}. 
\begin{table}[H]
\caption{Simulation parameters used in the first tier}
\label{tab:second table}
\centering
\begin{tabular}{c|c|c|c}
\hline
\textbf{Social network name} & \textbf{No. of nodes} & \textbf{No. of edges} & \textbf{Cost range (random distribution)}\\
\hline
Facebook & 4039 & 88234 & $[20,~50]$\\
\hline
Twitter & 81306 & 1768149 & $[20,~50]$\\
\hline
Google+ & 107614 & 13673453 & $[20,~50]$\\
\hline
\end{tabular}
\end{table}
\begin{figure}[H]
     \centering
     \begin{subfigure}[h]{0.35\textwidth}
         \centering
         \includegraphics[scale = 0.13]{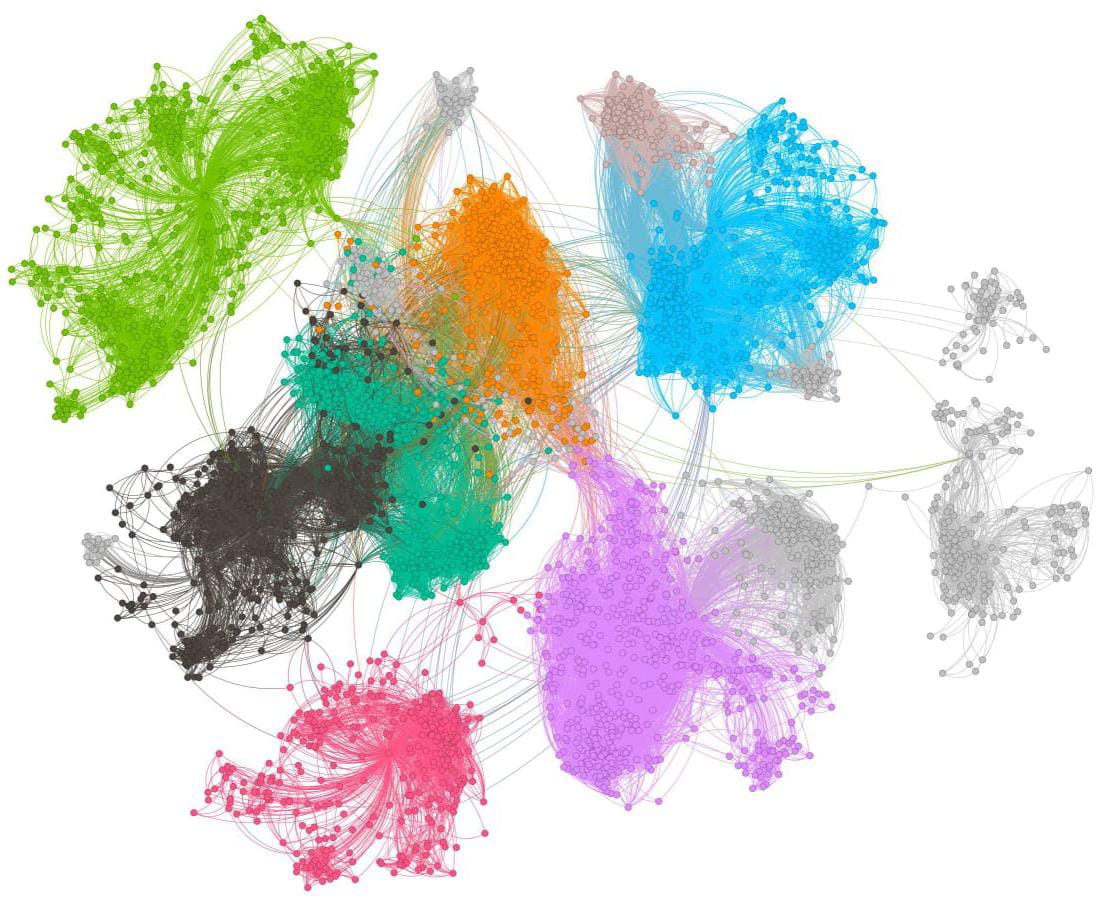}
         \label{fig:sim1a}
     \end{subfigure}
      \hspace*{-7mm}
     \begin{subfigure}[h]{0.35\textwidth}
         \centering
         \includegraphics[scale = 0.13]{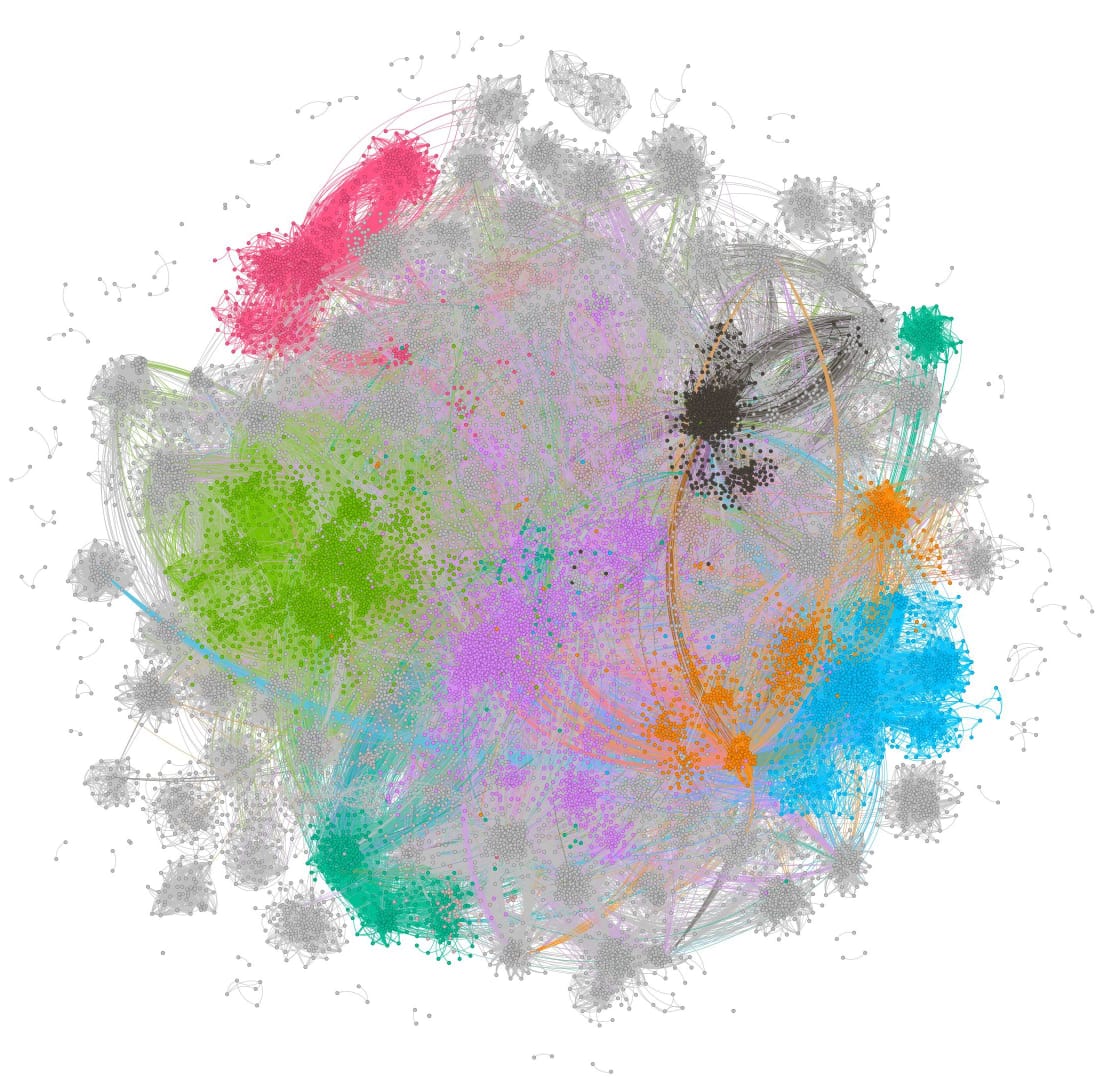}
         \label{fig:sim1b}
     \end{subfigure}
     \hspace*{-7mm}
      \begin{subfigure}[h]{0.35\textwidth}
         \centering
         \includegraphics[scale = 0.13]{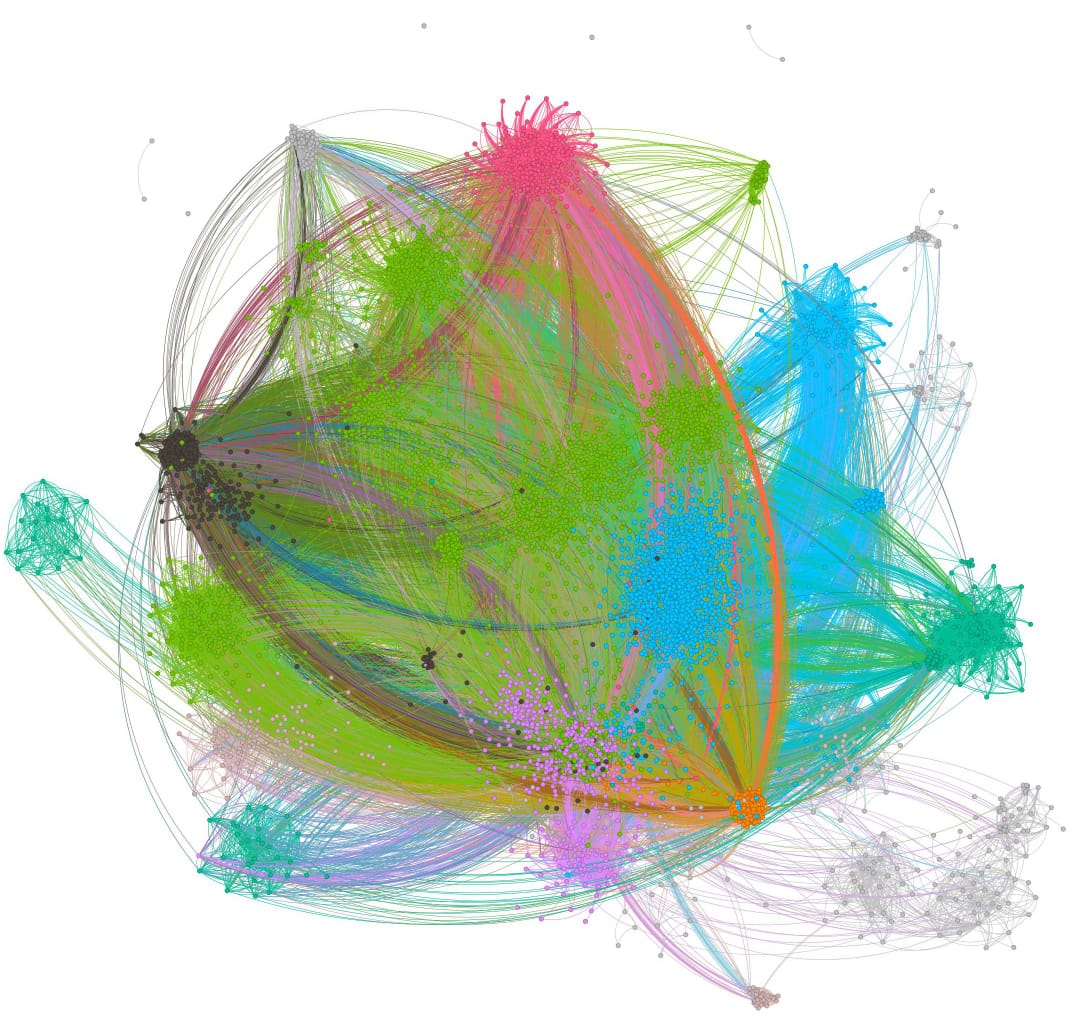}
         \label{fig:sim1c}
     \end{subfigure}
     \caption{Pictorial representation of users' social connections on Facebook, Twitter, and Google+}
     \label{fig:pr}
\end{figure}
For comparison purposes, one parameter that is kept intact for the three distinct social networks is the cost range of IoT devices. For each IoT device, a cost for notifying other IoT devices is picked up randomly from the given cost range. The unit of cost is the dollar. The pictorial representation of the social connections of the users in the three different social networking sites is depicted in Figure \ref{fig:pr}. The reason behind utilizing three different data sets is to strengthen the claim mentioned in the theoretical analysis part. For comparing TENM, PSM, and NTBFM on the ground of truthfulness it is considered that 30\% of the IoT devices are misreporting their true cost in the case of NTBFM and is represented as NTBFM-D in the simulation results. Here, misreporting the cost means that the IoT devices will lower their cost by some fixed amount so that their chance of getting selected will increase.
\item \textbf{Simulation setup for second tier -} For simulation purposes, each time some 50 IoT devices are selected randomly for providing ranking over 100 other randomly selected IoT devices that are placed on the scale of $[0,1]$. To generate the peak value of the IoT devices we have utilized \emph{random distribution} and \emph{normal distribution}. For the normal distribution, the mean value is taken as $\mu = 0.6$ and standard deviation $\sigma = 0.3$ as shown in Table \ref{tab:sim2}. Similarly, for the second phase of the second tier, the bid value range that is utilized for random and normal distributions is depicted in Table \ref{tab:sim2}. The reason behind considering the two different probability distributions is to strengthen the claim made in the theoretical analysis. 

\begin{table}[H]
\renewcommand{\arraystretch}{0.99}  
\caption{Data set utilized for simulation in case of ECTAI mechanism}
\label{tab:sim2}
\centering
  \begin{tabular}{c||c|c}
\hline
\textbf{Parameters} & \textbf{Values} & \textbf{Description}\\
\hline
$f$ & $100$ & Number of IoT devices that are to be ranked.\\
\hline
$g$ & $50$ & Number of IoT devices that provide a ranking.\\
\hline
$\alpha_i$ & $[0,~1]$ & Peak value selected randomly for random distribution.\\
\hline
$\alpha_i$ & $[\mu = 0.6, \sigma = 0.3]$ & Peak value determination for normal distribution.\\
\hline
$v_i (S)$ & $[30,45]$ & True valuation range for random distribution.\\
\hline
$v_i (S)$ & $[\mu = 37, \sigma = 8]$ & True valuation range for normal distribution.\\
\hline
\end{tabular}
\end{table}  
\end{enumerate}
\subsection{Baselines}
\label{subsec:bl}
In this section, the baselines that are used for comparing the proposed mechanisms in two tiers are discussed.
\begin{enumerate}
\item \textbf{Baseline for first tier -} For comparing the efficiency of TENM, we considered two baseline methods, namely, NTBFM and PSM \cite{Singer_2016}. In this, NTBFM is non-truthful and PSM is truthful in nature by their construction and is shown in the simulation results. Both the mechanisms are budget feasible.   
\begin{itemize}
\item NTBFM $-$ It consists of \emph{allocation} and \emph{payment} rules. In the \emph{allocation rule}, firstly, the IoT devices are sorted in descending order of marginal notification per cost of IoT devices. From the sorted ordering, each time an IoT device is picked up and a check is made that, whether, the cost of the picked-up IoT device is less than or equal to the available budget or not. If it satisfies this condition, then the IoT device is included in the winning set ($i.e.$ considered as the initial notifiers), and the available budget is reduced by the amount equal to the cost of the IoT device considered otherwise an \emph{allocation} rule terminate. The above process continues till the stopping condition in the allocation rule is satisfied. After that, coming to the \emph{payment rule}, each of the winning IoT devices is paid an amount equal to their respective reported cost.          
\item PSM \cite{Singer_2016} $-$ The PSM consists of \emph{allocation} and \emph{payment} rules. In the \emph{allocation rule}, the IoT devices are sorted in descending order of their costs so that $c_1 \leq c_2 \leq \ldots \leq c_n$. Further in sorted ordering, the largest index $k$ is determined for which $c_k \leq \frac{\boldsymbol{\mathcal{B}}}{k}$. The payment of any winning IoT device $i$ is given as $\min\{\frac{\boldsymbol{\mathcal{B}}}{k}, c_{k+1}\}$.  
\end{itemize}
\item \textbf{Baseline for second tier -} For comparing the efficiency of ECTAI, we consider a baseline method, namely, AVR \cite{T.roughgarden_20163}. In this, AVR is non-truthful and the manipulative behavior of IoT devices in the case of AVR is reflected in the simulation results. The idea of AVR is quite similar to that of ECTAI. Steps 1-3 of AVR are exactly the same as for the ECTAI. ECTAI differs from AVR only in steps 4 and 5. In step 4 of AVR, calculate the average of the peak values reported by $g$ number of IoT devices. In the next step determine the IoT device among the IoT devices that are placed on the scale of $[0, 1]$ whose peak value lies closer to the average peak value. It will be considered a quality IoT device. Repeat steps 1-5 of AVR until all the IoT devices are ranked.\\
\indent For comparing the efficiency of WiPD, we consider the baseline method as a greedy mechanism (in simulation graphs it is named GREEDY). It is to be noted that the proposed greedy mechanism is non-truthful in nature. The idea of the greedy mechanism is: (1) IoT devices are sorted in ascending order of the ratio of valuation to square root of the number of reported sets of tasks. (2) Each time, an $i^{th}$ IoT device is picked-up from the sorted ordering, and a check is made whether the intersection of tasks of the IoT devices present in the winning set and currently considered IoT device requested set of task is $\phi$ or not. If it is $\phi$ then $i^{th}$ IoT device is considered in the winning set, otherwise not. (3) The payment of the winning IoT devices will be their respective bid value.             
\end{enumerate}
  
\subsection{Performance Analysis} 
In this section, the performance analysis of TENM, ECTAI, and WiPD is carried out based on the metrics mentioned in Section \ref{sec:Sim} above. First, the performance analysis of TENM is presented and after that, the performance analysis of ECTAI and WiPD is discussed.  
\begin{enumerate}
\item \textbf{Performance analysis of TENM on Facebook, Twitter, and Google+ data sets.}
The performance of TENM is compared with the existing works (NTBFM and PSM) in the scenario discussed in the first tier of the proposed model. The discussion mainly circumvents around the metrics given in Section \ref{sec:Sim} for the first tier.
\begin{itemize}
\item \textbf{Truthfulness -} The simulation results shown in Figure \ref{fig:sim1} show the comparison of TENM with the two benchmark mechanisms $i.e.$ NTBFM and PSM on the ground of utility of notifiers. The x-axis of the graphs represents the iteration (or round) number and the y-axis of the graphs represents the utility of notifiers. Fixing the available budget $\mathcal{B}$ to 15000\$. It is evident from Figure \ref{fig:sim1} that in the case of TENM, the utility of IoT devices is sometimes higher and sometimes lower than in the case of PSM for all three data sets. This nature of the graphs is appearing from the construction of mechanisms. 
     
\begin{figure}[H]
     \centering
     \begin{subfigure}[h]{0.35\textwidth}
         \centering
         \includegraphics[scale = 0.44]{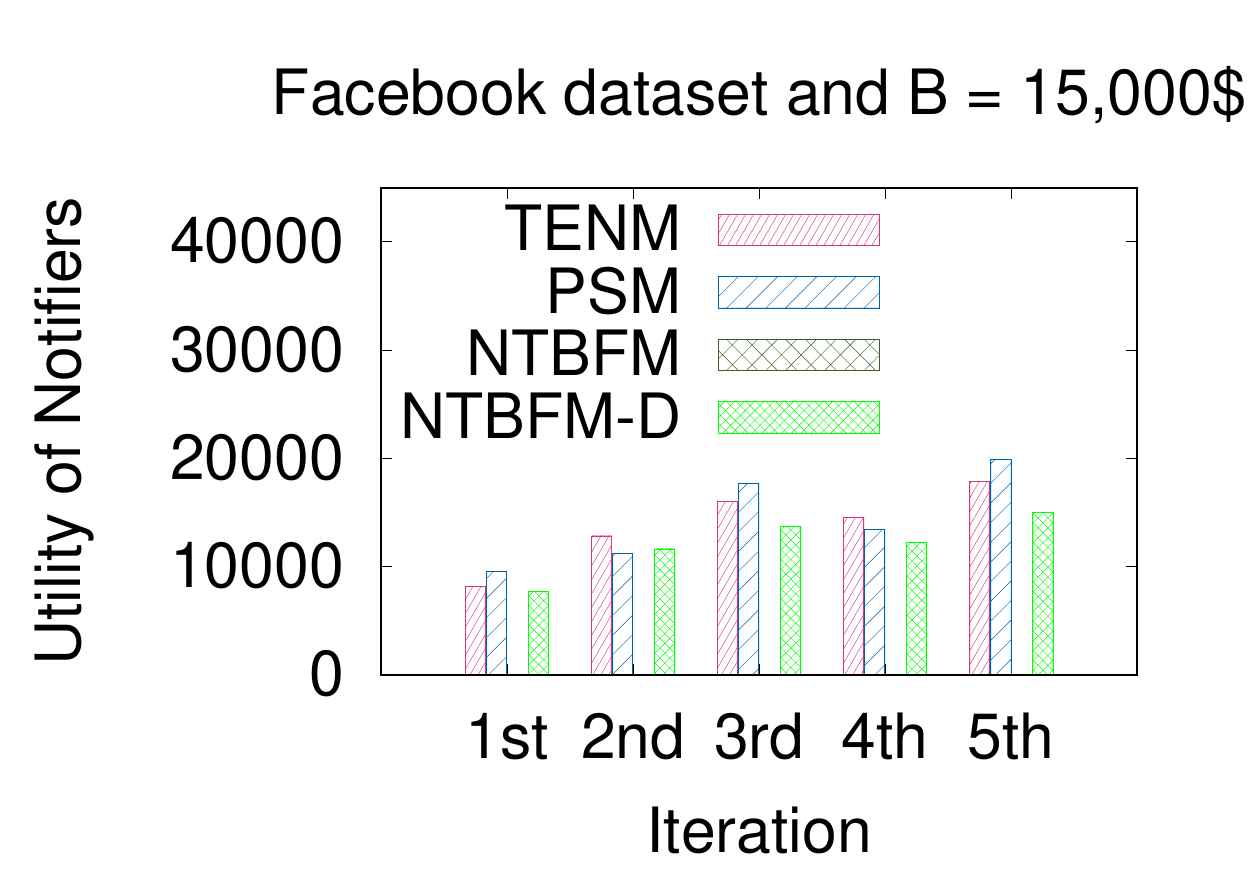}
         \label{fig:sim1a}
     \end{subfigure}
      \hspace*{-7mm}
     \begin{subfigure}[h]{0.35\textwidth}
         \centering
         \includegraphics[scale = 0.44]{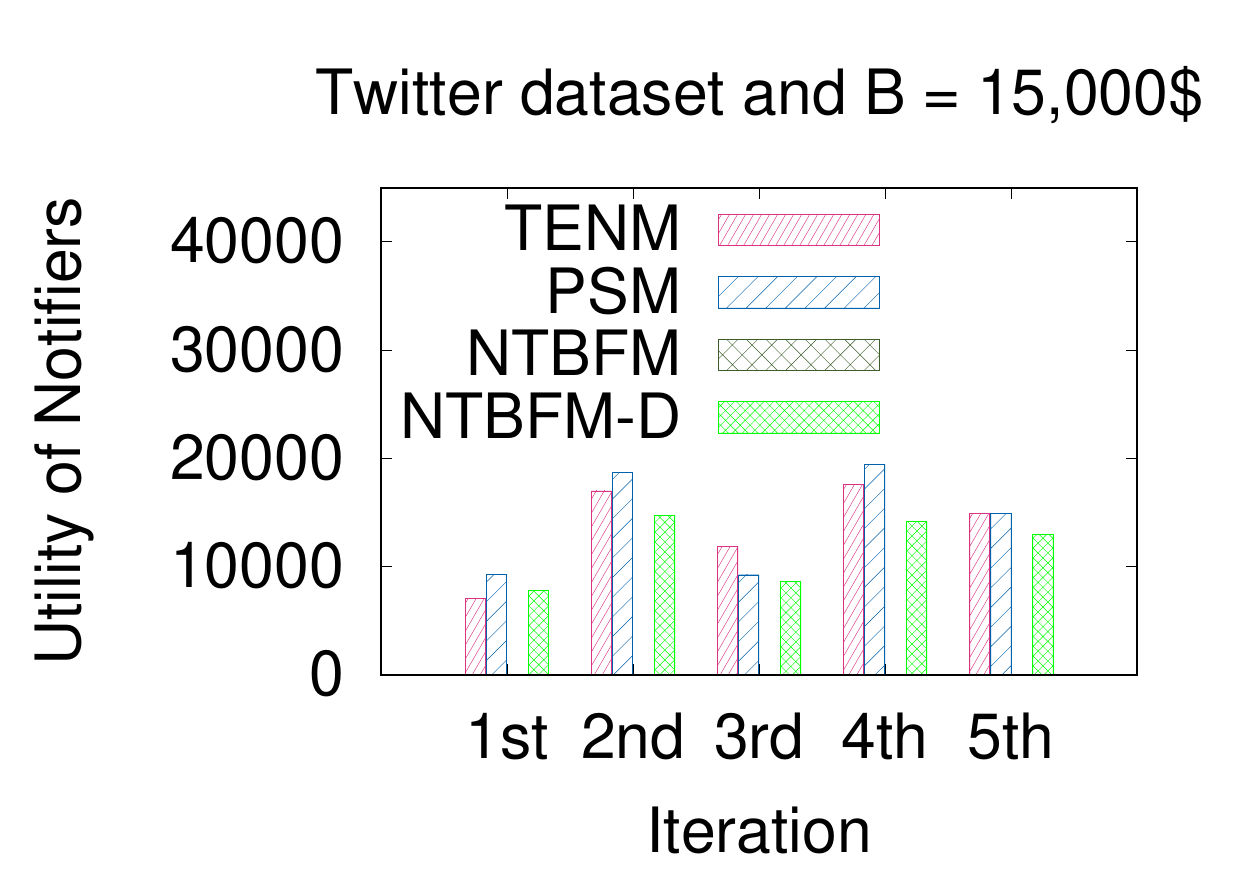}
         \label{fig:sim1b}
     \end{subfigure}
     \hspace*{-7mm}
      \begin{subfigure}[h]{0.35\textwidth}
         \centering
         \includegraphics[scale = 0.44]{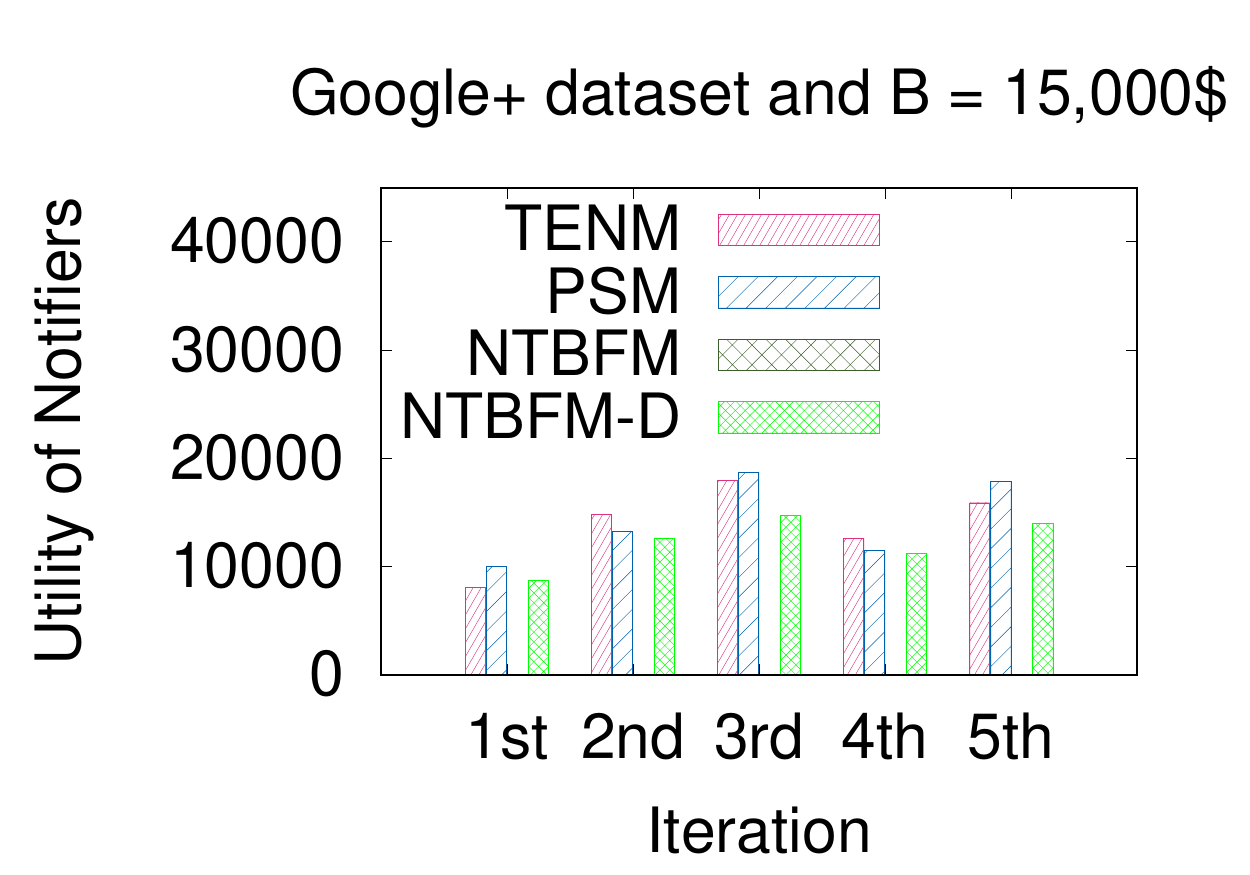}
         \label{fig:sim1c}
     \end{subfigure}
     \caption{Comparison of TENM, NTBFM, NTBFM-D, and PSM on Facebook, Twitter, and Google+ data in terms of the utility of notifiers.}
     \label{fig:sim1}
\end{figure}
For NTBFM, the IoT devices will have \emph{zero} utility because the payment made to the winning IoT devices is equal to the respective cost revealed by them as shown in Figure \ref{fig:sim1}. However, NTBFM is vulnerable to manipulation, which means that the IoT devices can gain by misreporting (in this case, 30\% of the IoT devices are reporting their cost below the true cost by some amount) their privately held cost. It can be easily seen in Figure \ref{fig:sim1} that for all three data sets the IoT devices to have higher utility in the case of NTBFM when misreporting (it is represented as NTBFM-D in the graphs) than the case when all the participating IoT devices report truthfully. So, the participating IoT devices gain by misreporting their true value in the case of NTBFM. From the simulation results, it can be inferred that TENM and PSM are truthful and NTBFM is vulnerable to manipulation. 
\end{itemize}    
\begin{itemize}
\item \textbf{Budget feasibility:} The simulation results presented in Figure \ref{fig:sim2} shows the comparison of TENM with the benchmark mechanisms $i.e.$ NTBFM and PSM on the ground of \emph{budget feasibility}. The x-axis of the graphs represents the iteration number and the y-axis of the graphs represent the total payment in dollars. In Figure \ref{fig:sim2} it can be seen that the total payment made to the winning IoT devices in the case of PSM is sometimes higher and sometimes lower than the total payment made to the IoT devices in the case of TENM. This nature of the graphs is evident from the construction of the two mechanisms $i.e.$ TENM and PSM. On the other hand, NTBFM has the highest value for the total payment made to the IoT devices among the three mechanisms. It is due to multiple reasons: (1) the full available budget is utilized for deciding the winners, and (2) the available budget is draining out slowly which leads to a larger number of IoT devices getting selected as initial notifiers. As the number of IoT devices getting selected is higher in the case of NTBFM and so is the total payment. From the construction of the three mechanisms it is clear that the total payment made to the IoT devices will be below the available budget and is also depicted in Figure \ref{fig:sim2}. So, TENM, PSM, and NTBFM are budget feasible.  
\end{itemize} 
\begin{figure}[H]
     \centering
     \begin{subfigure}[h]{0.35\textwidth}
         \centering
         \includegraphics[scale = 0.28]{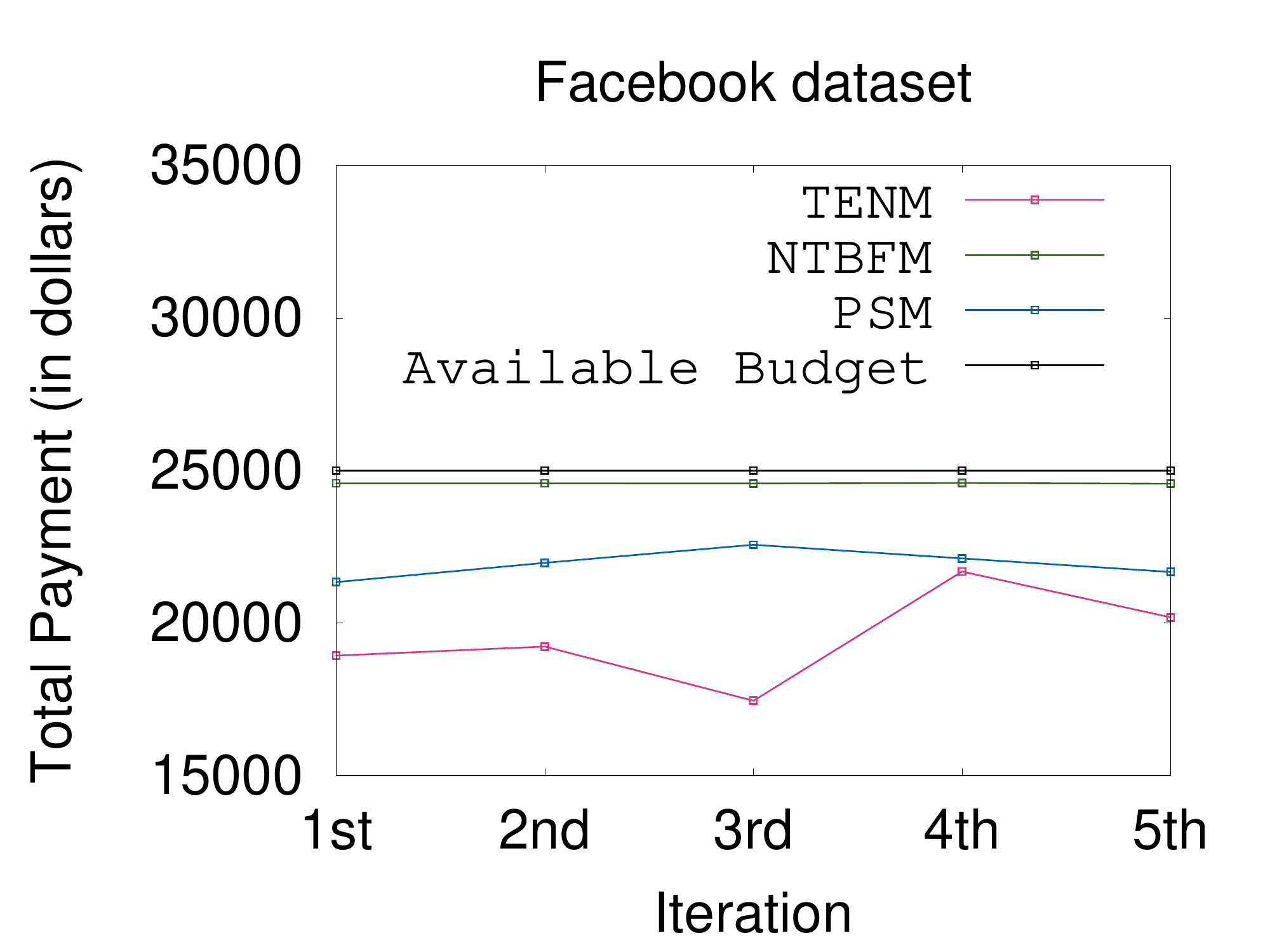}
         \label{fig:sim2a}
     \end{subfigure}
      \hspace*{-7mm}
     \begin{subfigure}[h]{0.35\textwidth}
         \centering
         \includegraphics[scale = 0.28]{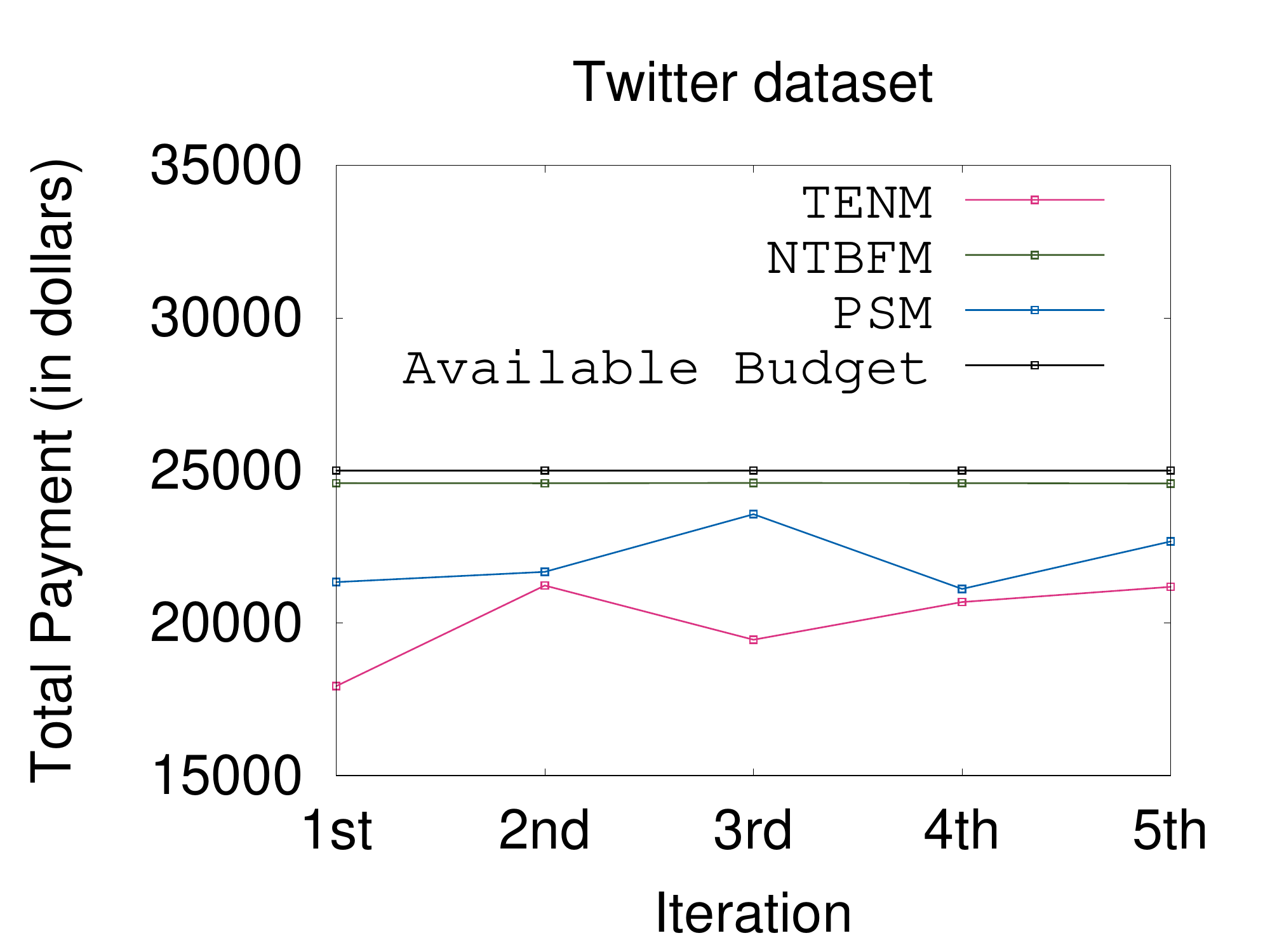}
         \label{fig:sim2b}
     \end{subfigure}
     \hspace*{-7mm}
      \begin{subfigure}[h]{0.35\textwidth}
         \centering
         \includegraphics[scale = 0.28]{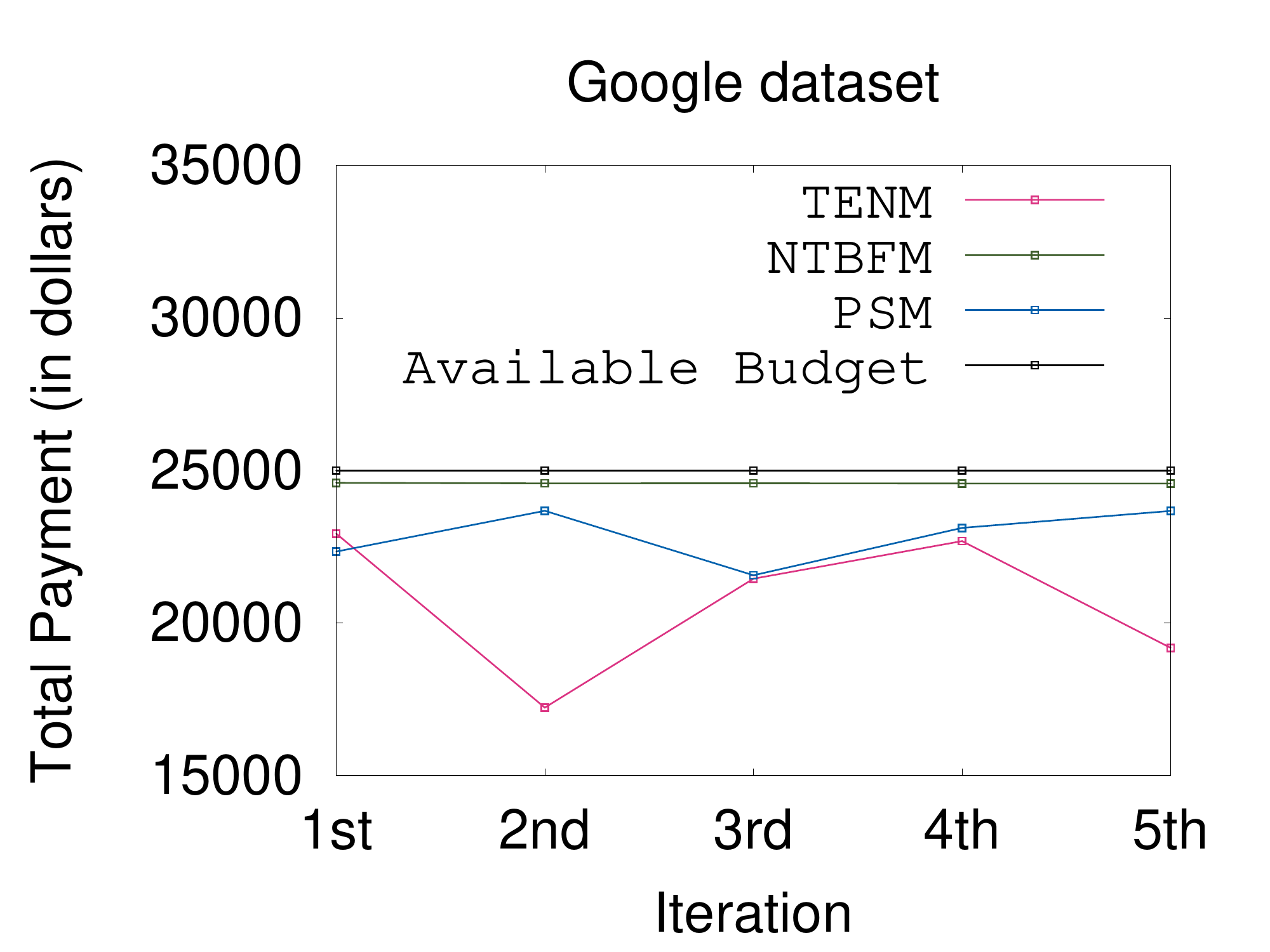}
         \label{fig:sim2c}
     \end{subfigure}
          \caption{Comparison of TENM, NTBFM, NTBFM-D, and PSM on Facebook, Twitter, and Google+ data in terms of budget feasibility.}
     \label{fig:sim2}
\end{figure}
\begin{itemize} 
\item \textbf{Number of IoT devices selected as initial notifiers in the social network:} In Figure \ref{fig:sim3} the comparison of TENM, NTBFM, and PSM is done on the ground of the number of IoT devices got selected as initial notifiers. In Figure \ref{fig:sim3} the x-axis of the graphs represents the budget and the y-axis of the graphs represents the number of notifiers selected. The comparison is carried out for three different social networks $i.e$ Facebook, Twitter, and Google+. From Figure \ref{fig:sim3}, it can be seen that the number of IoT devices selected as notifiers in the case of NTBFM is more than the number of IoT devices selected as notifiers in the case of TENM and is more than the number of IoT devices selected as notifiers in case of PSM for all the three data sets. This nature is due to the fact that in the case of NTBFM, the overall budget is utilized in the allocation rule and it is eaten away slowly than in the case of TENM where the budget allocation rule starts with half of the available budget and as the algorithm progresses the budget drains out very fast and less number of IoT devices get selected. On the other hand, in the case of PSM, the allocation rule starts with the full budget but in each iteration, the available budget drains out faster than in the case of TENM and NTBFM. So, due to this reason less number of IoT devices get selected. So, fewer IoT devices get notified about the task execution process.        
\end{itemize} 
\begin{figure}[H]
     \centering
     \begin{subfigure}[h]{0.35\textwidth}
         \centering
         \includegraphics[scale =0.29]{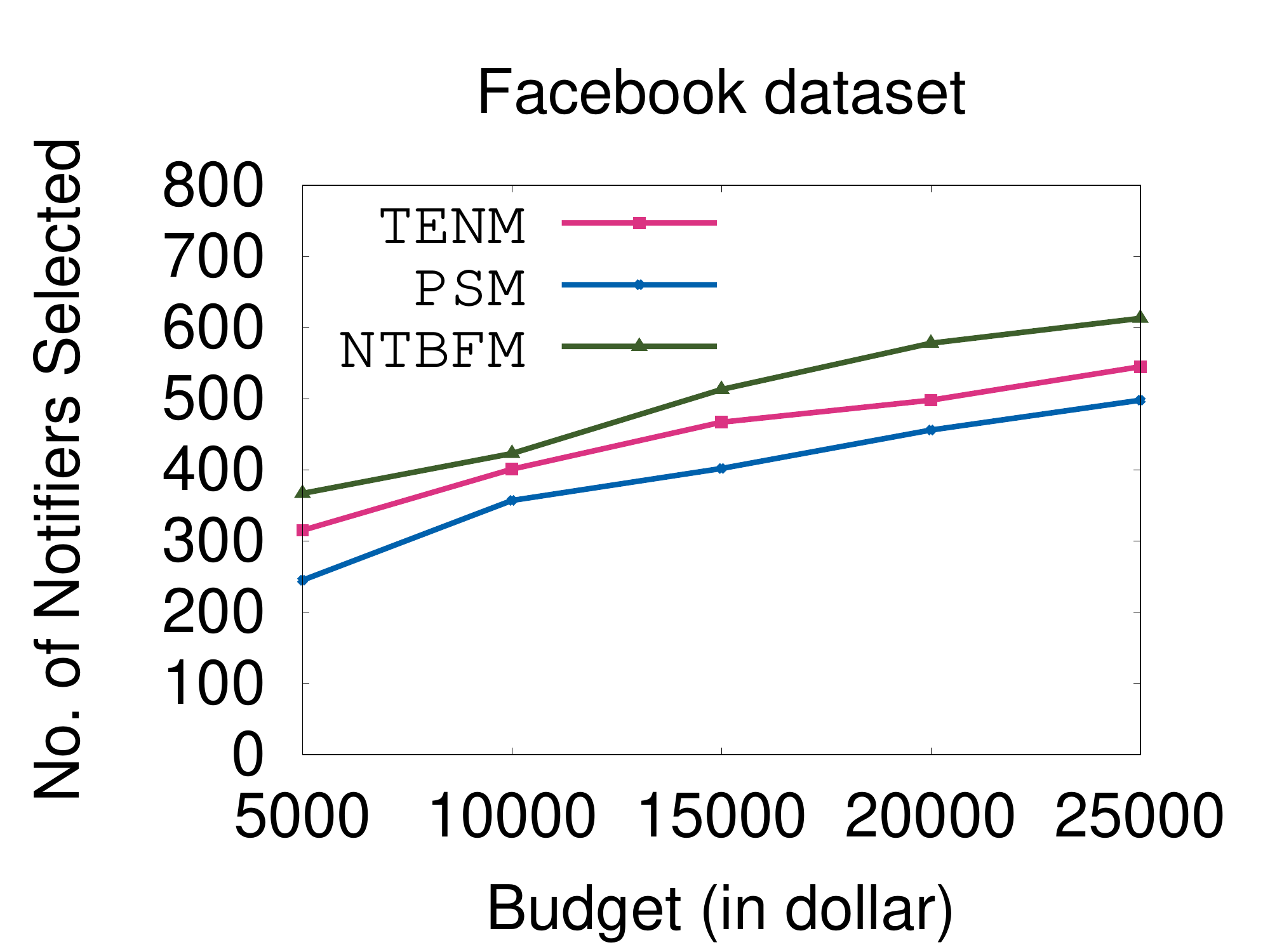}
         \label{fig:sim3a}
     \end{subfigure}
      \hspace*{-7mm}
     \begin{subfigure}[h]{0.35\textwidth}
         \centering
         \includegraphics[scale = 0.29]{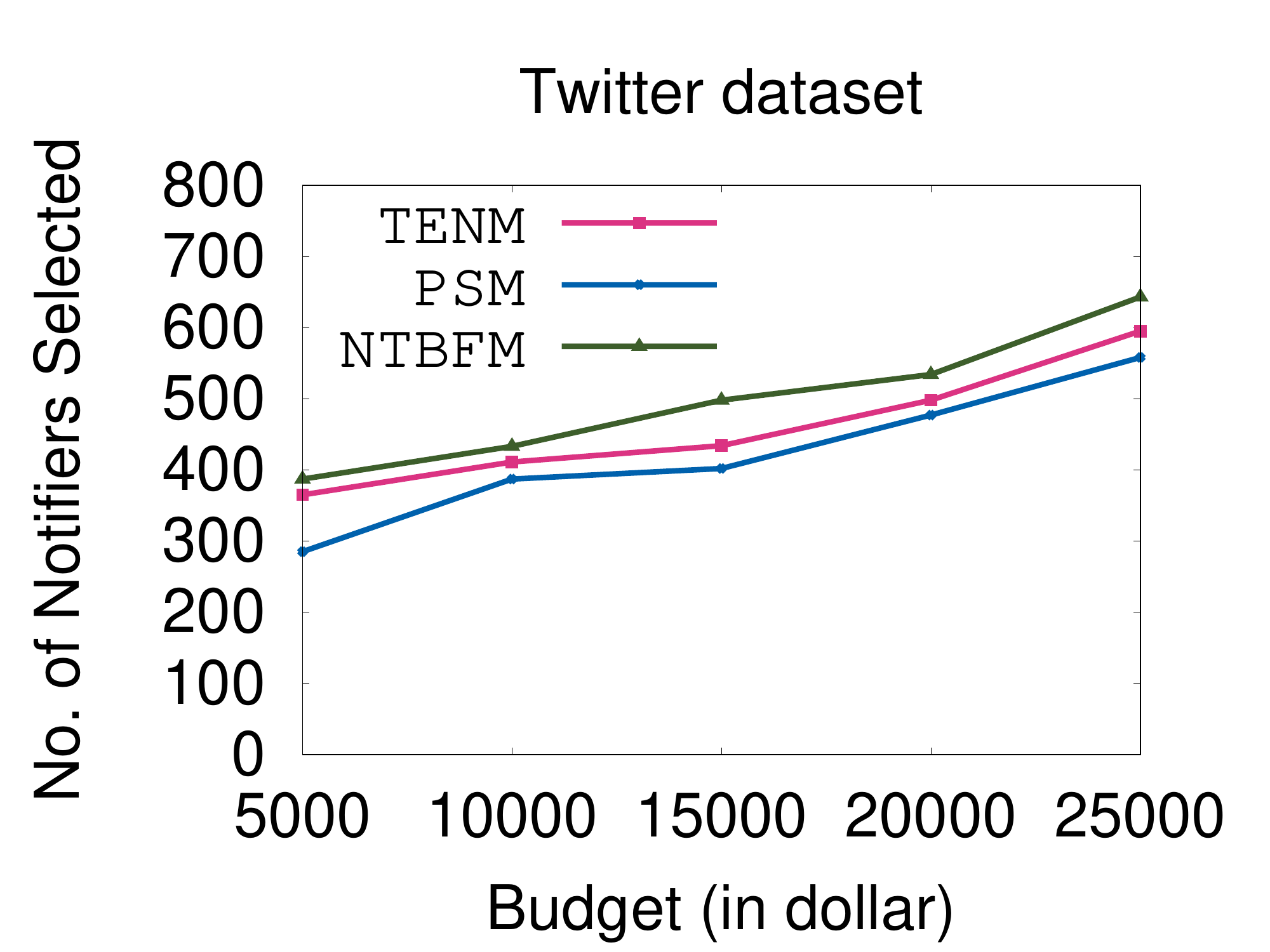}
         \label{fig:sim3b}
     \end{subfigure}
     \hspace*{-7mm}
      \begin{subfigure}[h]{0.35\textwidth}
         \centering
         \includegraphics[scale = 0.29]{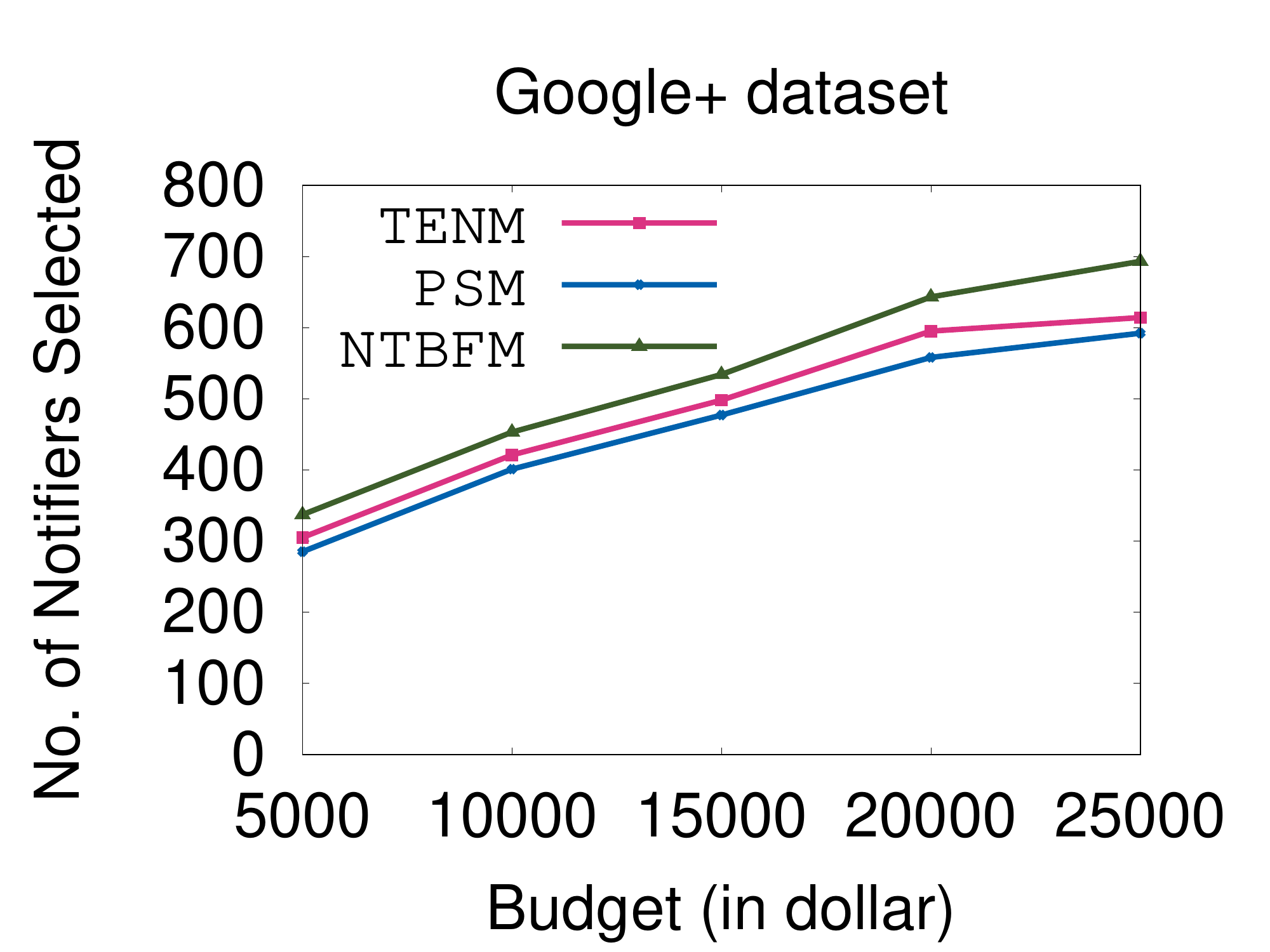}
         \label{fig:sim3c}
     \end{subfigure}
          \caption{Comparison of TENM, NTBFM, and PSM on Facebook, Twitter, and Google+ data in terms of the number of IoT devices selected as initial notifiers in the social network.}
     \label{fig:sim3}
\end{figure}
\begin{itemize} 
\item \textbf{Running time:} In Figure \ref{fig:sim4} the execution time (or running time) of TENM is compared with the execution time of NTBFM and PSM. The x-axis of the graph shown in Figure \ref{fig:sim4} represents the budget in dollars and the y-axis represents running time in milliseconds. The simulation results show that TENM takes more time than PSM and NTBFM for providing the desired results. The reason behind the higher running time for TENM is both allocation and pricing mechanisms. In the allocation rule of TENM, most of the time is killed up in calculating the marginal notification of each of the nodes in the social graph. Similar is the reason for the pricing mechanism of TENM, in this, the IoT device for whom the payment is calculated is dragged out of the crowdsourcing market and again the allocation rule is executed on the remaining IoT devices. The process is repeated for all the winning IoT devices. In the case of PSM, the payment is the minimum of the two quantities $\frac{\mathcal{B}}{k}$ and $c_{k+1}$. So, this payment calculation takes time less than the payment calculation for TENM. On the other hand, in the case of NTBFM, the payment is the reported cost of the IoT devices and so takes the lowest time. From the above discussion, it can be inferred that TENM running time is the highest among all the three mechanisms.  From the simulation graphs, it can be seen that  all three mechanisms are scalable.     
\end{itemize} 
\begin{figure}[H]
     \centering
     \begin{subfigure}[h]{0.35\textwidth}
        \centering
         \includegraphics[scale = 0.27]{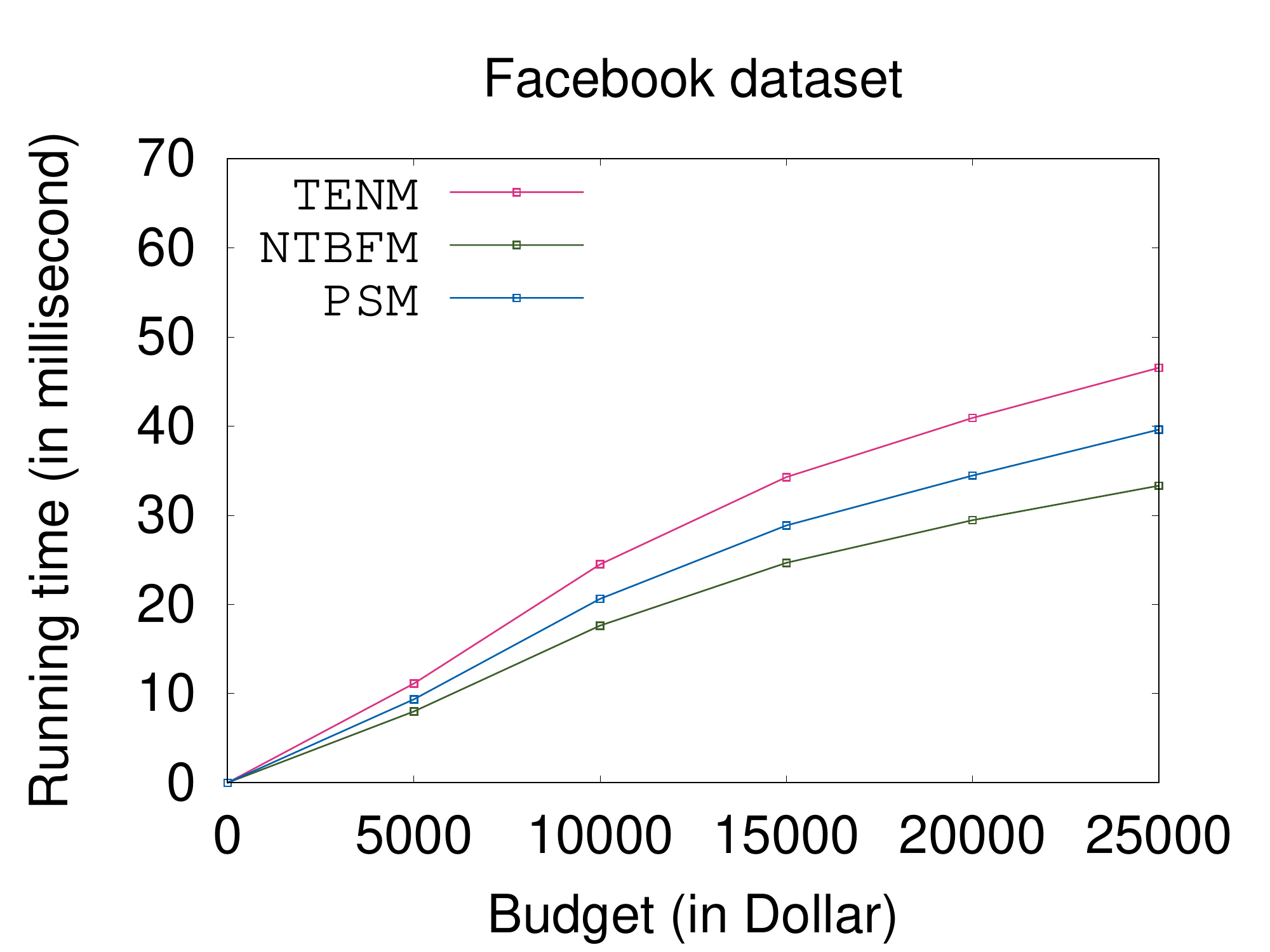}
         \label{fig:sim4a}
     \end{subfigure}
      \hspace*{-7mm}
     \begin{subfigure}[h]{0.35\textwidth}
        \centering
         \includegraphics[scale = 0.27]{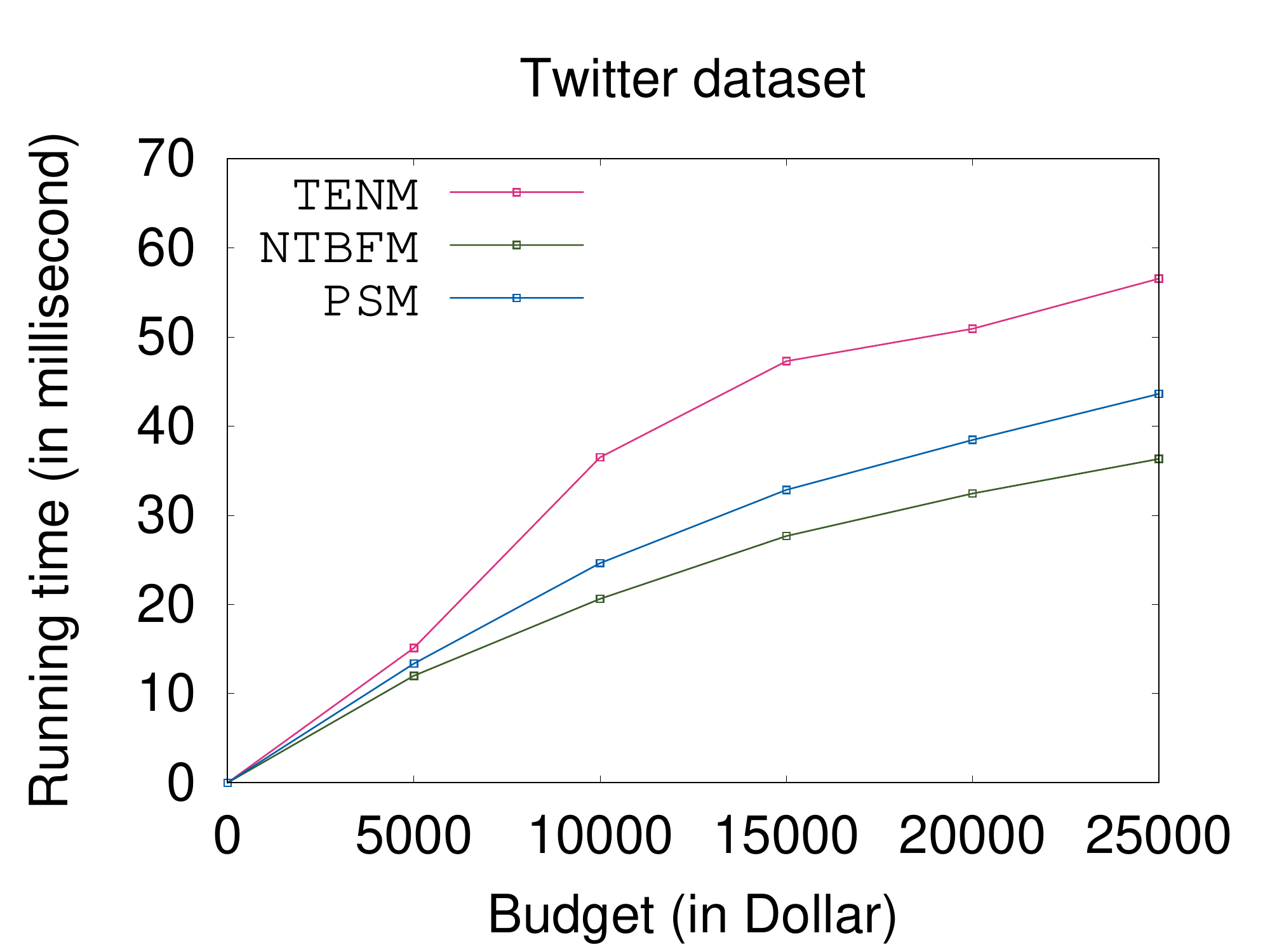}
         \label{fig:sim4b}
     \end{subfigure}
     \hspace*{-7mm}
      \begin{subfigure}[h]{0.35\textwidth}
        \centering
         \includegraphics[scale = 0.27]{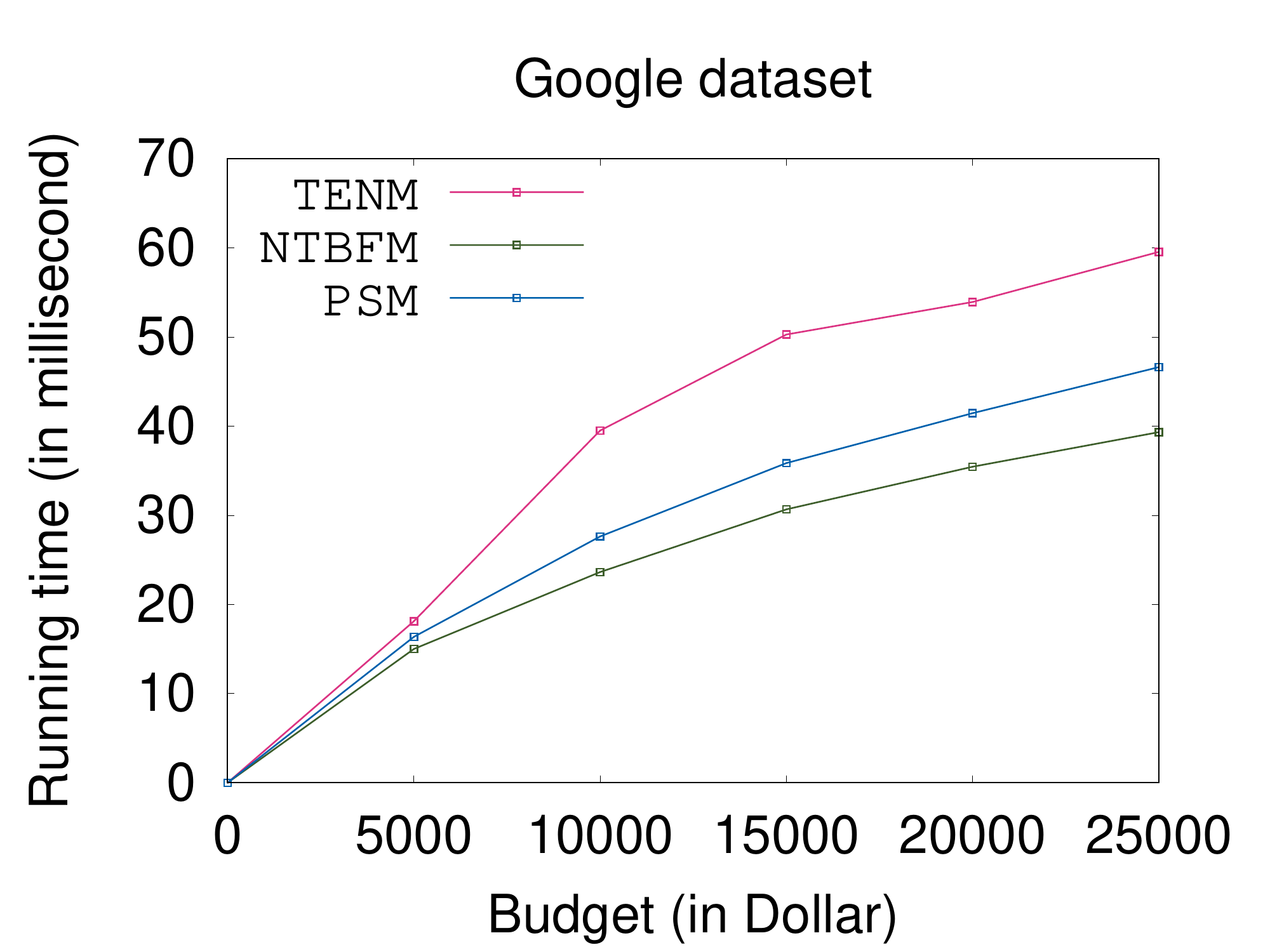}
         \label{fig:sim4c}
     \end{subfigure}
     \caption{Comparison of TENM, NTBFM, NTBFM-D, and PSM on Facebook, Twitter, and Google+ data in terms of execution time.}
     \label{fig:sim4}
\end{figure}
\item \textbf{Performance analysis of ECTAI and WiPD.} The performance of ECTAI and WiPD is compared with AVR and GREEDY respectively. The discussion mainly circumvents around the metrics given in Section \ref{sec:Sim} for the second tier.
\begin{itemize}
\item \textbf{Truthfulness:} The simulation results shown in Figure \ref{fig:sim44} show the comparison of ECTAI with AVR on the ground of truthfulness for both \emph{random} and \emph{normal} distributions. The x-axis of the graphs represents the number of IoT devices and the y-axis represents the utility of IoT devices that are ranked. In the case of ECTAI, the resultant peak value will be sometimes lower and sometimes higher than in the case of AVR, and so as the utility of the IoT devices as shown in Figure \ref{fig:sim44}. This nature of the graphs is from the construction of the mechanisms. In the case of ECTAI and AVR, the higher the utility value lower will be the quality of the IoT devices, and the lower the utility value higher will be the quality of IoT devices. Further, in the graphs of Figure \ref{fig:sim44} the manipulative behavior of IoT devices can be seen in the case of AVR. If 30\% of the IoT devices misreport (reporting the peak value closer to its favorite IoT device) their peak value then the total utility of the IoT devices will drop down and will be sometime less as compared to the case when all the IoT devices are reporting their true peak value. This manipulative nature of AVR can be seen for both random and normal distributions in Figure \ref{fig:sim44}. From the above discussion, one can say that AVR is vulnerable to manipulation but ECTAI is not. 
\begin{figure}[H]
\centering
\begin{subfigure}[b]{0.49\textwidth}
                \centering
                \includegraphics[scale=0.60]{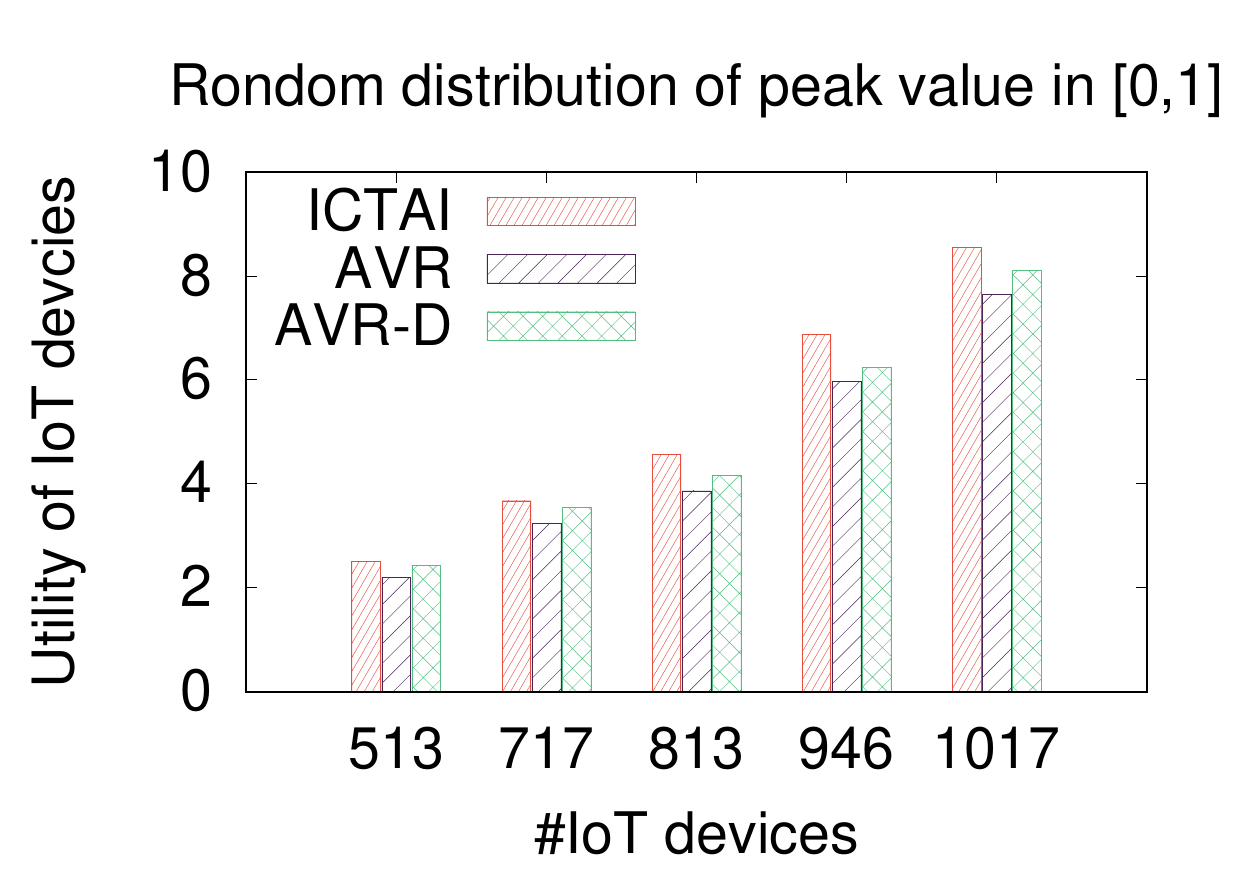}
                \label{fig:sim44A}
        \end{subfigure}%
        \begin{subfigure}[b]{0.49\textwidth}
                \centering
                \includegraphics[scale=0.60]{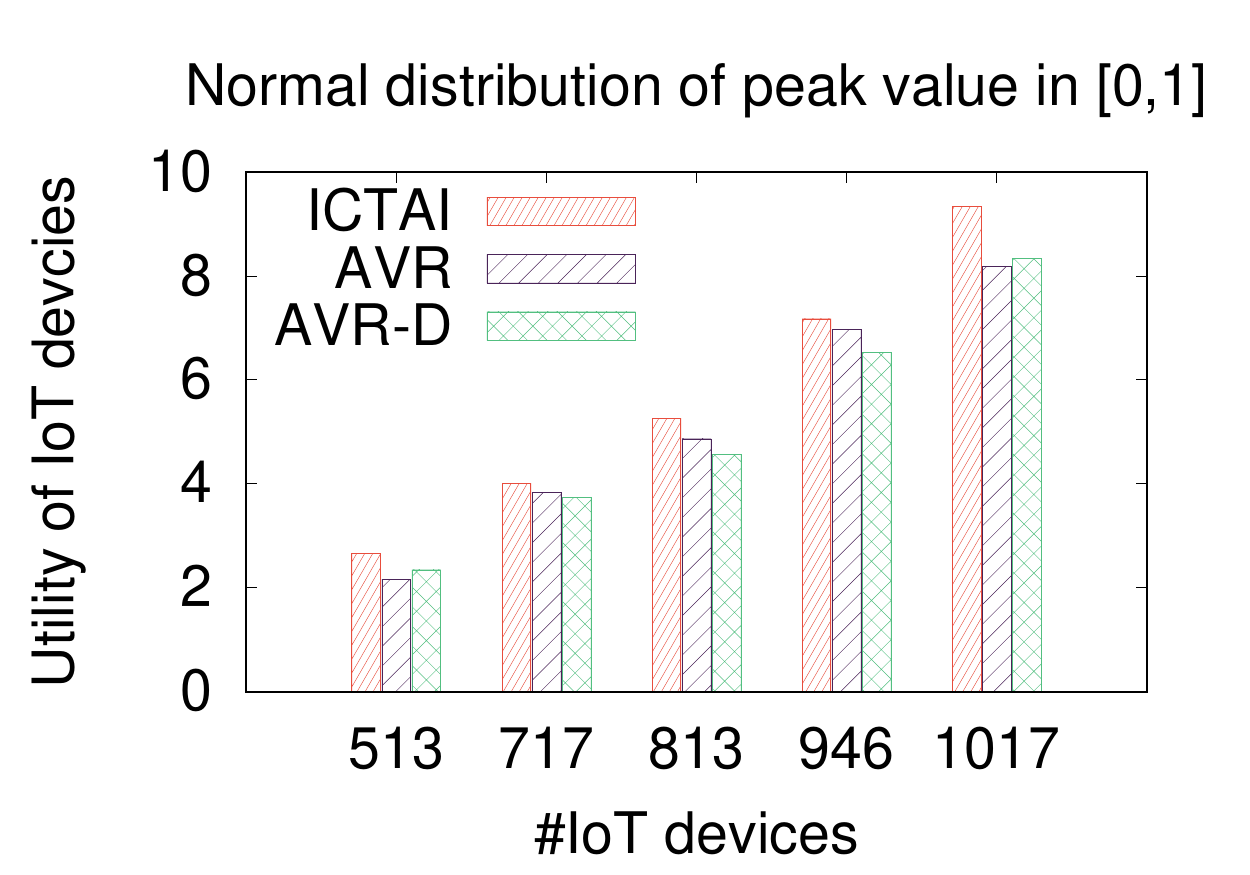}
    \label{fig:sim44B}
        \end{subfigure}
        \caption{Comparison of ECTAI, AVR, and AVR-D in terms of the utility of IoT devices in random and normal distribution cases.}
        \label{fig:sim44}
\end{figure} 
\begin{figure}[H]
\centering
\begin{subfigure}[b]{0.49\textwidth}
                \centering
                \includegraphics[scale=0.60]{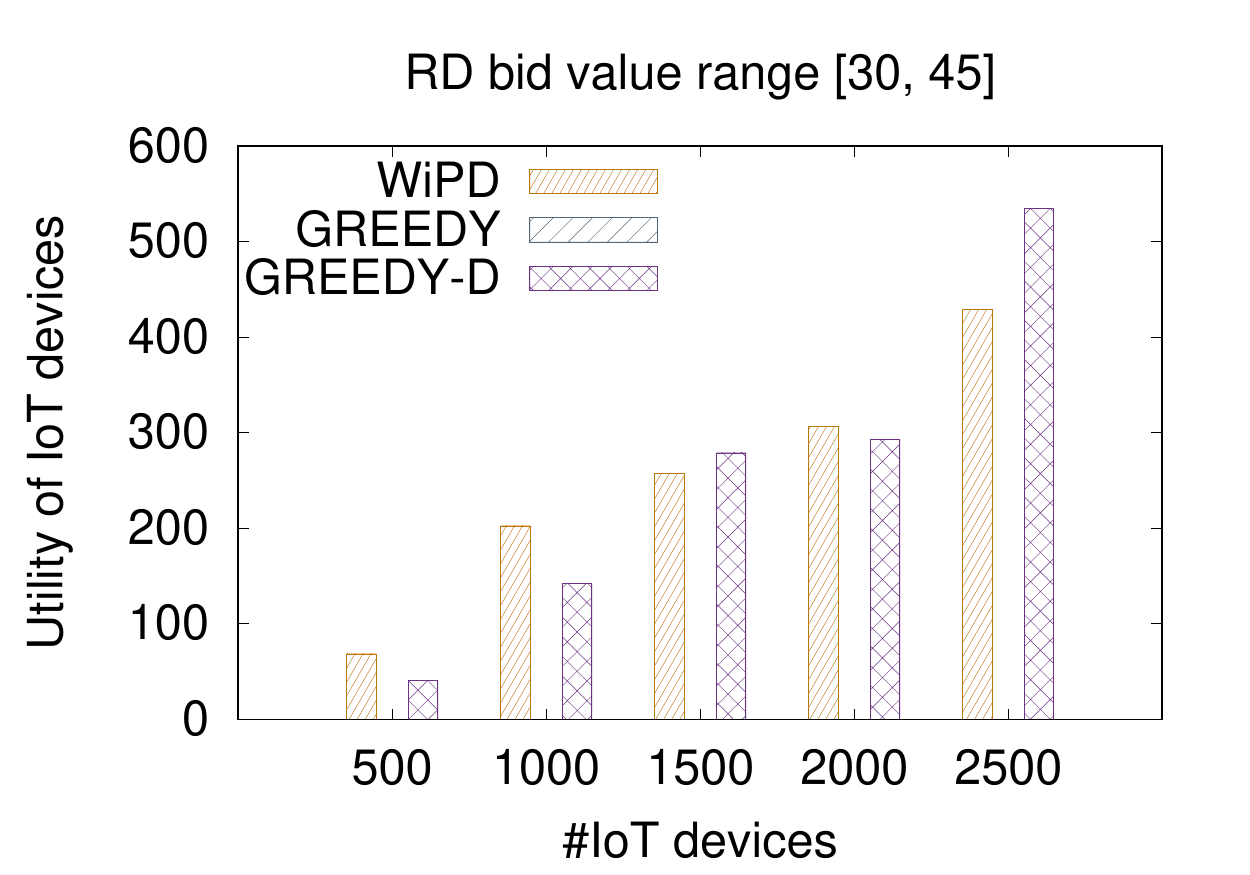}
                \label{fig:sim54A}
        \end{subfigure}%
        \begin{subfigure}[b]{0.49\textwidth}
                \centering
                \includegraphics[scale=0.60]{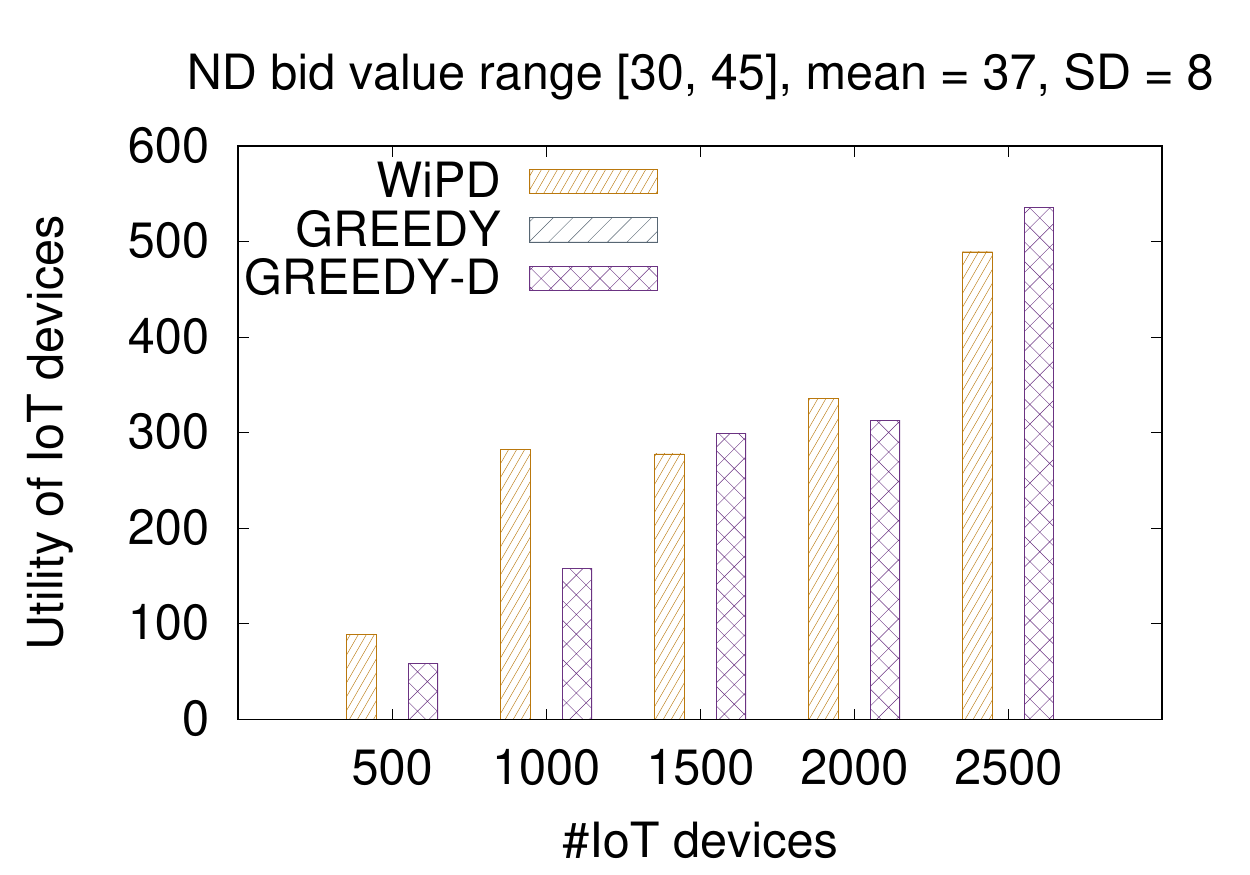}
        \end{subfigure}
        \caption{Comparison of WiPD, GREEDY, GREEDY-D in terms of the utility of IoT devices in random and normal distribution cases.}
        \label{fig:sim54}
\end{figure}  
\indent The simulation results shown in Figure \ref{fig:sim54} show the comparison of WiPD with GREEDY on the ground of truthfulness for both \emph{random} and \emph{normal} distributions. The x-axis of the graphs represents the number of IoT devices and the y-axis represents the utility of IoT devices. In this case, the utility of IoT devices is calculated by utilizing the formula given in equation \ref{equ:1a}. It can be seen from Figure \ref{fig:sim54} that the utility of IoT devices in the case of WiPD is higher than in the case of GREEDY. It is zero in the case of GREEDY. It is due to the reason that in the case of GREEDY, the payment made to the IoT devices is their reported bid value and so the utility is 0. Further, in the graphs of Figure \ref{fig:sim54} the manipulative behavior of IoT devices can be seen in the case of the greedy mechanism (named GREEDY-D in the graphs of Figure \ref{fig:sim54}). If 30\% of the IoT devices misreport (reporting the bid value higher than their true valuation) their true valuation then the total utility of the IoT devices will go up and will be more as compared to the case when all the IoT devices were reporting their true bid value. This manipulative nature of GREEDY can be seen in both random and normal distributions. Hence, GREEDY is vulnerable to manipulation and WiPD is not vulnerable to manipulation.
\item \textbf{Running time:}  In Figure \ref{fig:sim64a} the execution time (or running time) of ECTAI is compared with the execution time of AVR. The x-axis of the graph shown in Figure \ref{fig:sim64a} represents the number of IoT devices and the y-axis represents running time in milliseconds. The simulation results show that ECTAI and AVR take almost the same time for providing the desired results. The reason is that in both the mechanism a simple arithmetic operation takes place and it takes constant time as mentioned in Lemma \ref{lemma:sp1}. The above-discussed scenario is depicted in Figure \ref{fig:sim64a}. From the simulation graphs, it can be seen that both ECTAI and AVR are scalable.
\end{itemize} 
\begin{figure}[H]
     \centering
     \begin{subfigure}[h]{0.48\textwidth}
        \centering
         \includegraphics[scale = 0.32]{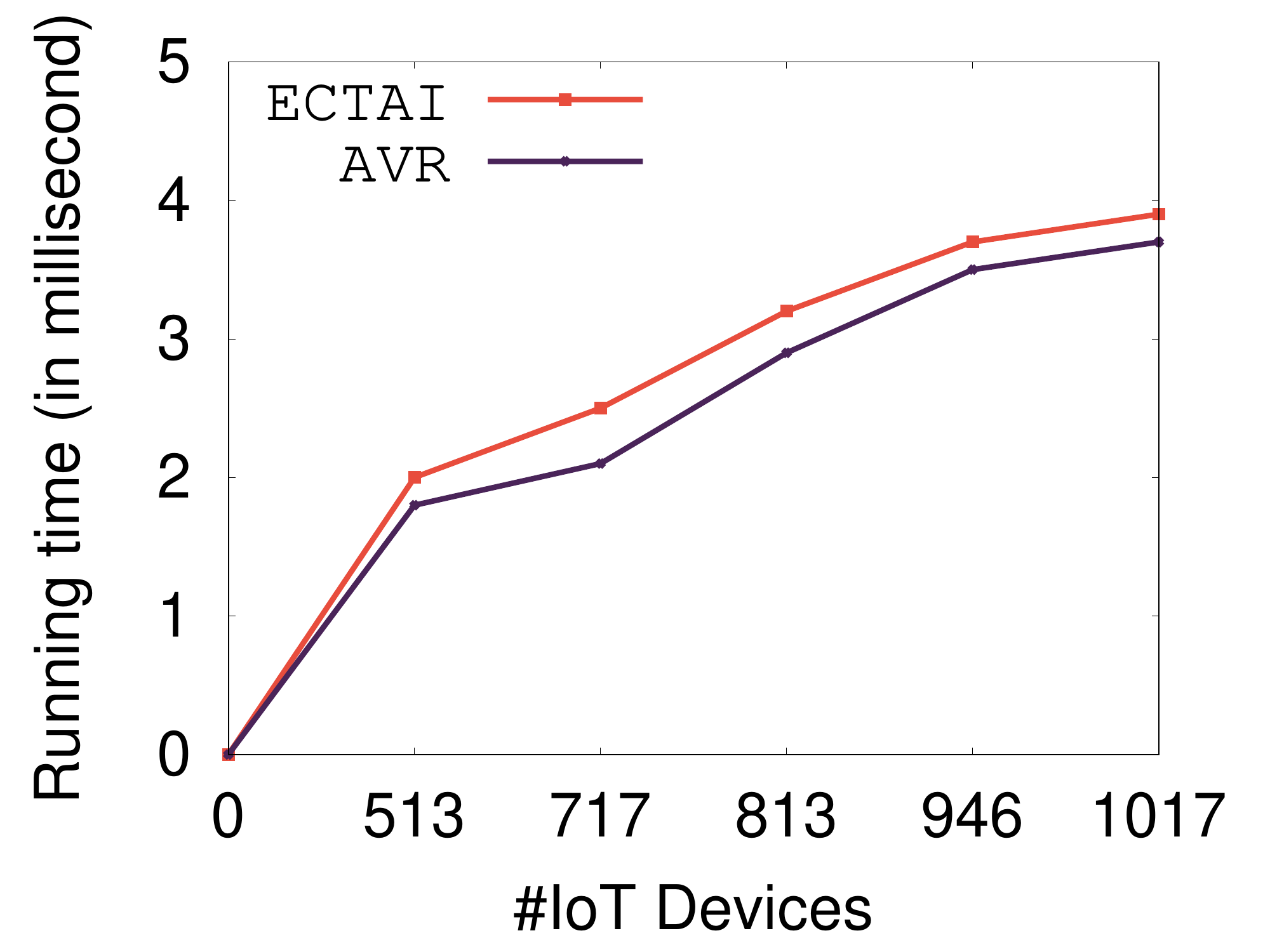}
      \caption{Running time analysis of ECTAI and AVR}
         \label{fig:sim64a}
     \end{subfigure}
      \hspace*{-7mm}
     \begin{subfigure}[h]{0.48\textwidth}
        \centering
         \includegraphics[scale = 0.32]{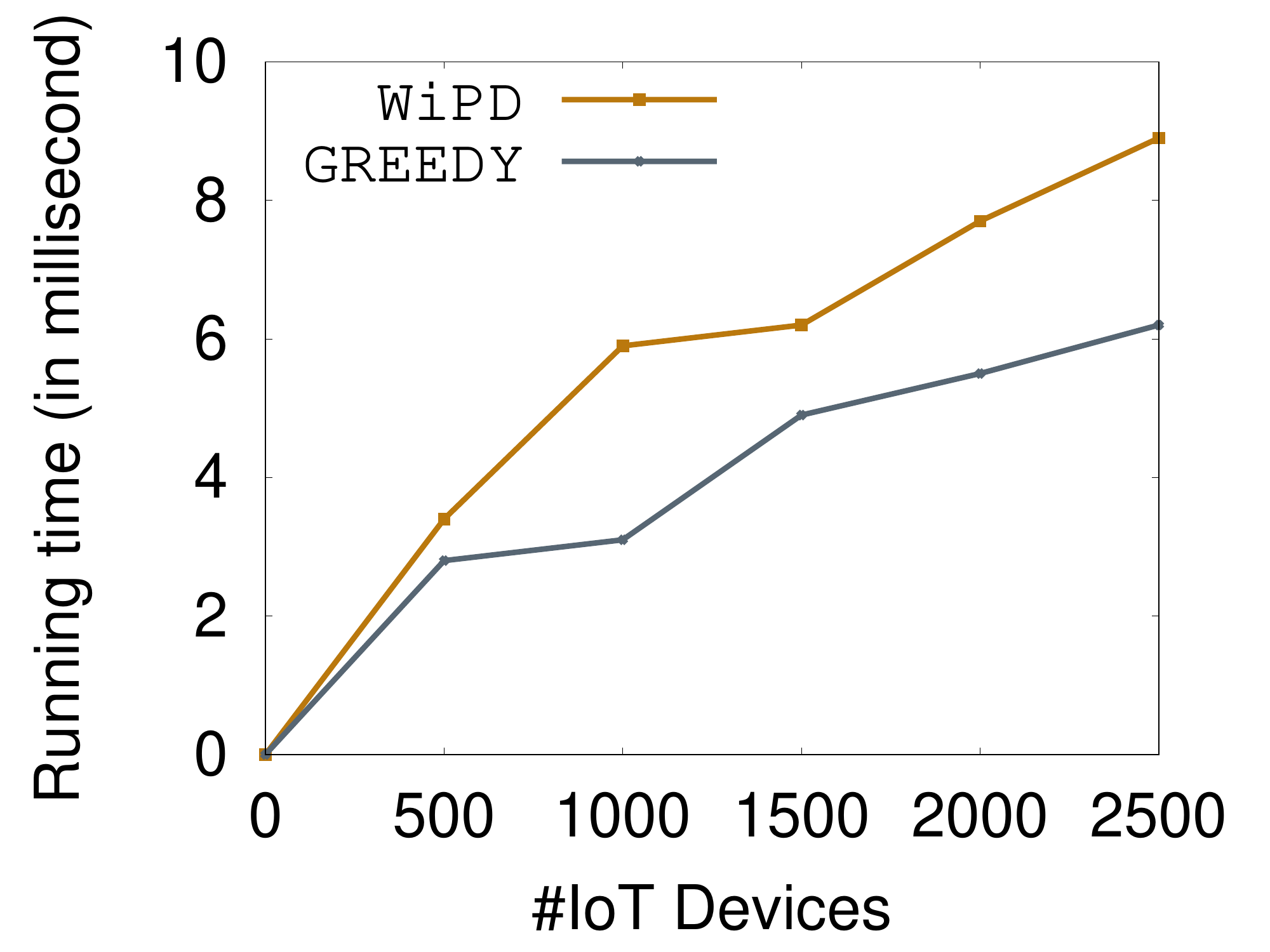}
        \caption{Running time analysis of WiPD and GREEDY}
         \label{fig:sim64b}
     \end{subfigure}
          \caption{Comparison of ECTAI and AVR, WiPD and GREEDY in terms of execution time.}
     \label{fig:sim64}
\end{figure}
In the case of WiPD and GREEDY, the comparison based on running time is shown in Figure \ref{fig:sim64b}. The x-axis of the graph shown in Figure \ref{fig:sim64b} represents the number of IoT devices and the y-axis represents running time in milliseconds. The simulation results show that WiPD takes more time than GREEDY. The reason is that in the case of WiPD deciding the set of tasks to be assigned to the IoT devices and the payment to be made to them takes more time. In the case of GREEDY the allocation rule takes time but the payment rule is simple and takes less time. The above-discussed scenario is depicted in Figure \ref{fig:sim64b}. From the simulation graphs, it can be seen that both WiPD and GREEDY are scalable.
\end{enumerate}

\section{Conclusion and Future Works}
\label{se:conc}
In this paper, one of the crowdsourcing scenarios is studied as a two-tiered process in \emph{strategic} setting. In the first tier of the proposed framework, the social connections of the IoT devices are utilized to make other IoT devices aware of the task execution event. For this purpose a \emph{truthful} mechanism namely TENM is proposed for identifying the initial notifiers such that the total payment made to the initial notifiers for their services is within the available budget. Once a substantial number of IoT devices got notified about the task execution process, in the second tier, the \emph{truthful} mechanisms namely ECTAI and WiPD are proposed that determine the quality of IoT devices and assign a subset of tasks to each of the quality IoT devices respectively. For the second tier of the proposed model, we assumed that the valuation function is \emph{gross substitute}.\\
\indent Through theoretical analysis, it is shown that the proposed mechanisms are \emph{correct}, \emph{computationally efficient}, and \emph{truthful}. Further, through probabilistic analysis, the estimate is done on the number of task executors that get notified about the task execution process. The simulation result shows that in the first tier, in the case of NTBFM the IoT devices are gaining by misreporting their costs, whereas in the case of TENM and PSM the IoT devices can maximize their utility only by reporting a true cost. The reason is that NTBFM is vulnerable to manipulation and TENM is not vulnerable to manipulation. Further, the results show that TENM, NTBFM, and PSM are budget feasible. With the increase in budget, the number of IoT devices selected as initial notifiers in a social network is increasing for all three mechanisms i.e. TENM, PSM, and NTBFM. Further, the comparison is done in terms of execution time, and found that TENM, NTBFM, and PSM are scalable. In the second tier, ECTAI and WiPD are compared to their respective benchmark mechanism on the ground of truthfulness and running time. In the case of ECTAI the IoT devices can maximize their utility only by reporting their true peak value. On the other hand in the case of AVR, the IoT devices gain by misreporting their true peak value. It is due to the reason that AVR is vulnerable to manipulation. In terms of running time, both ECTAI and AVR take almost the same time. In the case of WiPD the IoT devices can maximize their utility only by reporting their valuation. On the other hand in the case of GREEDY, the IoT devices gain by misreporting their true valuation. It is because of this reason that GREEDY is vulnerable to manipulation. In terms of running time, GREEDY outperforms WiPD but both are scalable.   \\
\indent In the future it will be interesting to see if the above-discussed set-up can be extended to the case where IoT devices have general valuations (not GS). It poses the challenge of designing a \emph{truthful} mechanism (right now the \emph{truthful} mechanism is guaranteed only when the IoT devices have GS valuations). Another direction could be, say, in addition to the above set-up each of the tasks has a start time and finish time associated with it. In such cases designing a time-bound truthful mechanism will be a challenging task. In our upcoming work, the focus will be on designing a truthful mechanism for the above-mentioned scenario that also takes care of the quality of the IoT devices and the completion of tasks within the given start and finish times.\\
\indent The results, the code, the real-world data, and the synthetic data that are utilized in the paper are made available at \url{https://github.com/chbhargavi/IoT_Elseiver}.

\section*{Acknowledgment}
The authors would like to thank the Centre of Excellence for the Internet of Things (CoE-IoT) of VIT-AP University, Amaravati, India for providing valuable suggestions and support.

\bibliographystyle{alpha}
\bibliography{phd}

\newpage 
\appendix
\end{document}